\documentclass[Journal,twocolumn]{IEEEtran}
\usepackage[T1]{fontenc}
\usepackage[latin9]{inputenc}
\usepackage{amsthm}
\usepackage{amsmath}
\usepackage{amssymb}
\usepackage{graphicx}

\makeatletter
 \theoremstyle{plain}
\newtheorem{thm}{Theorem}
  \theoremstyle{remark}
  \newtheorem{rem}{Remark}
   \theoremstyle{plain}

  \theoremstyle{definition}
  \newtheorem{defn}{Definition}
  \theoremstyle{definition}
  \newtheorem{example}{Example}
 \theoremstyle{plain}
  \newtheorem{lem}{Lemma}
 \theoremstyle{plain}
  \newtheorem{cor}{Corollary}

\begin{document}

\title{Combinatorial Message Sharing and a New Achievable Region for Multiple
Descriptions}


\author{Kumar Viswanatha$^{\ast}$, Emrah Akyol$^{\dagger}$, \IEEEmembership{Member,~IEEE}, and \\ Kenneth Rose$^{\ddagger}$, \IEEEmembership{Fellow,~IEEE,}
\thanks{The work was supported by the NSF under grants CCF-1016861, CCF-1118075 and CCF-1320599. At the time of this work, all authors were with the Department of Electrical and Computer Engineering, University of California - Santa Barbara, CA. The material in this paper was presented in part at the IEEE international symposium on information theory at Saint Petersburg, Russia, 2011 \cite{our_ISIT}, IEEE information theory workshop at Paraty, Brazil, 2012 \cite{our_ITW} and the IEEE international symposium on information theory at Istanbul, Turkey, 2013 \cite{our_isit_13}.} 
\thanks{$^{\ast}$K. Viswanatha is currently with the Qualcomm research center, San Diego, CA, USA. (e-mail: kumar@ece.ucsb.edu)}
\thanks{$^{\dagger}$E. Akyol is with the Department of Electrical and Computer Engineering, University of Illinois Urbana-Champaign, CA, USA (e-mail: akyol@illinois.edu).}
\thanks{$^{\ddagger}$K. Rose is with the Electrical and Computer Engineering Department, University of California - Santa Barbara, CA, USA. (e-mail: rose@ece.ucsb.edu)}}

\markboth{IEEE Transactions on Information Theory,~Vol.~XX, No.~XX, XXX~XXXX}%
{}


\maketitle
\begin{abstract}
This paper presents a new achievable rate-distortion region for the
general L channel multiple descriptions problem. A well known general
region for this problem is due to Venkataramani, Kramer and Goyal
(VKG) \cite{VKG}. Their encoding scheme is an extension of the El-Gamal-Cover
(EC) and Zhang-Berger (ZB) coding schemes to the L channel case and
includes a combinatorial number of refinement codebooks, one for each
subset of the descriptions. As in ZB, the scheme also allows for a
single common codeword to be shared by all descriptions. This paper
proposes a novel encoding technique involving `Combinatorial Message
Sharing' (CMS), where every subset of the descriptions may share a
distinct common message. This introduces a combinatorial number of
shared codebooks along with the refinement codebooks of \cite{VKG}.
These shared codebooks provide a more flexible framework to trade-off
redundancy across the messages for resilience to descriptions loss.
We derive an achievable rate-distortion region for the proposed technique,
and show that it subsumes the VKG region for general sources and distortion
measures. We further show that CMS provides a strict improvement of
the achievable region for any source and distortion measures for which
some 2-description subset is such that ZB achieves points outside
the EC region. We then show a more surprising result: CMS outperforms
VKG for a general class of sources and distortion measures, including
scenarios where the ZB and EC regions coincide for all 2-description
subsets. In particular, we show that CMS strictly improves on VKG,
for the $L-$channel quadratic Gaussian MD problem, for all $L\geq3$,
despite the fact that the EC region is complete for the corresponding
2-descriptions problem. Consequently, the `correlated quantization'
scheme (an extreme special case of VKG), that has been proven to be
optimal for several cross-sections of the $L-$channel quadratic Gaussian
MD problem, is strictly suboptimal in general. Using the encoding
principles derived, we show that the CMS scheme achieves the complete
rate-distortion region for several asymmetric cross-sections of the
$L-$channel quadratic Gaussian MD problem. \end{abstract}

\section{Introduction\label{sec:Introduction}}

The Multiple Descriptions (MD) problem was proposed in the late seventies
and has been studied extensively since, yielding a series of advances,
ranging from the derivation of asymptotic bounds \cite{VKG,EGC,ZB,Ahlswede,Ozarow,wang,Ramchandran,Ramchandran-1,wang_Vishwanath,Jun_Chen_ind_central}
to practical approaches for multiple descriptions quantizer design
\cite{Vaishampayan,MC_DA}. It was originally viewed as a method to
cope with channel failures, where multiple source descriptions are
generated and sent over different paths. The encoder generates $L$
descriptions for transmission over $L$ available channels. It is
assumed that the decoder receives a subset of the descriptions perfectly
and the remaining are lost, as shown in Fig. \ref{fig:Basic_MD}.
The objective of the MD problem is to design the encoders (for each
description) and decoders (for each possible received subset of the
descriptions), with respect to an overall rate-distortion (RD) trade-off.
The subtlety of the problem is due to the balance between the full
reconstruction quality versus quality of individual descriptions;
or noise free quality versus the amount of redundancy across descriptions
needed to achieve resilience to descriptions loss.

\begin{figure}
\centering\includegraphics[scale=0.55]{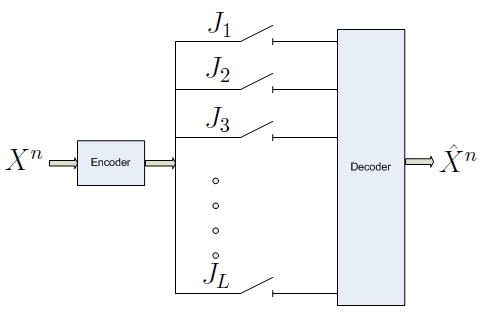}\caption{$L-$Channel Multiple Descriptions Setup : Each description is received
error free or is completely lost at the decoder\label{fig:Basic_MD}}
\end{figure}

One of the first achievable regions for the 2-channel MD problem was
derived by El-Gamal and Cover (EC) in 1982 \cite{EGC}. It follows
from Ozarow's results in \cite{Ozarow} that the EC region is complete
for the 2-channel quadratic Gaussian MD problem, i.e., when the source
is Gaussian and the distortion measure is mean squared error (MSE).
It was further shown by Ahlswede in \cite{Ahlswede} that the EC region
is complete for a cross-section of the 2-description MD setup, called
the `no-excess rate regime', a scenario wherein the central decoder
receives information at the minimum sum rate. This led to the popular
belief that the EC achievable region is complete for the general 2-channel
MD problem. However Zhang and Berger (ZB) in \cite{ZB} proved a then
surprising result, that the EC scheme is strictly sub-optimal in general.
In particular, they showed that for a binary source under Hamming
distortion, sending a common codeword in both descriptions can achieve
points that are strictly outside the EC region. While introducing
a common codeword implies explicit redundancy among the two descriptions,
this codeword assists in better coordination between the descriptions
leading to a strictly larger RD region.

Several researchers have since focused on extending EC and ZB to the
$L-$channel MD problem \cite{VKG,Ramchandran,Ramchandran-1,wang_Vishwanath,Jun_Chen_ind_central,Viswanatha_2_levels,Asymmetric_MD}.
An achievable scheme, due to Venkataramani, Kramer and Goyal (VGK)
\cite{VKG}, directly builds on EC and ZB, and introduces a combinatorial
number of refinement codebooks, one for each subset of the descriptions.
Motivated by ZB, a \textit{single} common codeword is also shared
by all the descriptions, which assists in better coordination of the
messages, improving the RD trade-off. For the $L-$channel quadratic
Gaussian problem, it was shown by Wang and Viswanath in \cite{wang_Vishwanath}
that a special case of the VKG coding scheme, where no common codeword
is sent, achieves the minimum sum rate when only the individual and
the central distortion constraints are imposed. In particular, they
showed that a `correlated quantization' based encoding scheme, which
is an extension of the Ozarow's encoding mechanism to the $L$-descriptions
problem, achieves the minimum sum rate for the cross-section involving
constraints only on the individual and central distortions. It was
also shown recently by Chen in \cite{Jun_Chen_ind_central} that,
in fact, this approach leads to the complete region for this particular
cross-section.

Pradhan, Puri and Ramachandran (PPR) considered a practically interesting
cross-section of the general $L-$channel MD problem in \cite{Ramchandran,Ramchandran-1}
called the `symmetric MD problem' wherein it is assumed that the rates
of all the descriptions are equal and the distortion is a function
only of the `number' of descriptions received rather than which particular
subset is received. They proposed a new coding scheme leveraging principles
from distributed source coding \cite{SW,WZ76}, and particularly Slepian
and Wolf's random binning techniques, and showed that, for this symmetric
cross-section, the proposed encoding scheme improves upon the VKG
region. Tian and Chen derived a new coding scheme in \cite{Tian_Chen_symmetric_K_descriptions}
for the symmetric MD problem which further extends the PPR region.
It was also shown in \cite{Approximating} that this region is very
close to complete for the symmetric quadratic Gaussian MD problem.
Recently, Wang and Viswanath \cite{Viswanatha_2_levels} derived a
coding scheme based on the VKG and the PPR encoding principles and
showed its sum-rate optimality for certain cross-sections of the quadratic
Gaussian problem wherein only 2 layers of distortions are imposed. Moreover, in \cite{Song_Shuo_Chen} , Song et.al. showed that, for the quadratic Gaussian setting, when the  minimum  sum  rate  is attained subject to  two  levels  of  distortion  constraints  (with  the  second  level imposed  on  the  complete  set  of  descriptions), the PPR scheme also leads to the minimum achievable distortion at the intermediate levels.

In this paper we present a new encoding scheme involving ``Combinatorial
Message Sharing'' (CMS), where a unique common codeword is sent in
(shared by) each subset of the descriptions, thereby introducing a
combinatorial number of \textit{shared codebooks}, along with the
refinement codebooks of \cite{VKG}. The common codewords enable better
coordination between descriptions, providing an improved overall RD
region. We derive an achievable region for CMS and show that it subsumes
VKG for general sources and distortion measures. Moreover, we show
that CMS achieves a strictly larger region than VKG for all $L>2$,
if there exists a 2-description subset for which ZB achieves points
strictly outside the EC region. In particular, CMS achieves strict
improvement for a binary source under Hamming distortion. We then
show a surprising result: CMS strictly outperforms VKG for a general
class of sources and distortion measures, which includes several scenarios
in which, for every 2-description subset, the ZB and the EC regions
coincide. In particular, we show that for a Gaussian source under
MSE, CMS achieves points strictly outside the VKG region. This result
is in striking contract to the corresponding 2-descriptions setting.
Optimality of EC for the 2-descriptions setup has led to a natural
conjecture that common codewords do not play a necessary role in quadratic
Gaussian MD coding, and all the achievable regions characterized so
far neglect the common layer (see, e.g., \cite{VKG,Jun_Chen_ind_central,Viswanatha_2_levels}).
In this paper, we show that the common codewords in CMS play a critical
role in achieving the complete RD region for several asymmetric cross-sections
of the general $L-$channel quadratic Gaussian MD problem. We note
that the principles underlying CMS have been shown in precursor work
to be useful in related applicational contexts of routing for networks
with correlated sources \cite{DIR} and data storage for selective
retrieval \cite{fusion}. We further note that a preliminary version
of the results in this paper appeared in \cite{our_ISIT}, \cite{our_ITW} and \cite{our_isit_13}.
We also note that, in this paper, we focus on generalizing the VKG
coding scheme using a combinatorial number of shared messages for
the MD problem with general sources and distortion measures. However,
the CMS principle can be extended to incorporate random binning based
encoding techniques, to utilize the underlying symmetry in the problem
setup. Preliminary results in this direction have recently appeared
in \cite{Binned_CMS,Binned_CMS_ITW} and form part of our current
research focus. The CMS scheme introduces
a structured mechanism for generating random codebooks that lead to
controlled redundancy across transmitted messages. The CMS scheme
is only one of the several possible generalizations of the ZB scheme
to $L$-descriptions. The fundamental idea of introducing structured
codes and controlled redundancy for $L-$channel multiple descriptions
problem is not new and has been explored in several publications in
the past, including \cite{Ramchandran,Ramchandran-1,Tian_Chen_symmetric_K_descriptions,Shirani_Pradhan_14}
and some of these schemes have been shown to outperform VKG for certain
sources and distortion measures. However, to the best of our knowledge,
for a Gaussian source, under MSE, CMS is the only generalization of the VKG scheme that
achieves a strictly larger rate-distortion region. There are other encoding schemes that can achieve points strictly outside the VKG region, such as the PPR scheme, but are not strict generalizations VKG  in the sense that their achievable region has not been shown to always subsume the entire VKG rate-distortion region. Hence there are fundamentally new insights to be gained by studying CMS. 

The rest of the paper is organized as follows. In Section \ref{sec:Prior_results},
we formally state the $L-$channel MD setup and briefly describe the
approaches and regions of EC \cite{EGC}, ZB \cite{ZB} and VKG \cite{VKG}.
To keep the notation simple, we first describe in Section \ref{sub:3-Descriptions-scenario},
the CMS scheme for the 3 descriptions scenario and extend it to the
general case in Section \ref{sub:L-Channel-case}. We then prove in
section \ref{sec:Strict-Improvement1} that CMS achieves strict improvement
over VKG whenever there exists a 2-description subset for which ZB
achieves points outside the EC region. In Section \ref{sec:Strict-Improvement2}
we prove that CMS outperforms VKG for a fairly general class of source-distortion
measure pairs, including a Gaussian source under MSE. Finally in Section
\ref{sec:Gaussian-MSE-Setting}, we derive new results for the $L-$channel
quadratic Gaussian MD setup and show that CMS achieves the complete
RD region for several asymmetric cross-sections of the general problem.

\section{Formal Definitions and Prior results\label{sec:Prior_results}}

We use uppercase letters to denote random variables (e.g., $X$) and
lowercase letters to denote their realizations (e.g., $x$). We use
script letters to denote sets, alphabets and subscript indices (e.g.,
$\mathcal{S},\mathcal{X},\mathcal{K}$). A sequence of $n$ independent
and identically distributed (iid) random variables is denoted by $X^{n}$
and its realization by $x^{n}$. $2^{\mathcal{S}}$ denotes the set
of all subsets (power set) of any set $\mathcal{S}$ and $|\mathcal{S}|$
denotes the set cardinality. Note that $|2^{\mathcal{S}}|=2^{|\mathcal{S}|}$.
$\mathcal{S}^{c}$ denotes the set complement (the universal set will
be explicitly specified when not obvious). For two sets $\mathcal{S}_{1}$
and $\mathcal{S}_{2}$, we denote the set difference by $\mathcal{S}_{1}-\mathcal{S}_{2}=\{k:k\in\mathcal{S}_{1},k\notin\mathcal{S}_{2}\}$.
All constant random variables, that take a single deterministic value, are denoted by $\Phi$. 

Throughout the paper, random variables, distortions and rates are
indexed using sets, e.g., $U_{\mathcal{K}},D_{\mathcal{K}},R_{\mathcal{K}}$,
where $\mathcal{K}$ takes values such as $\{1\}$, $\{2\}$, $\{1,2\}$,
$\{2,3\}$ etc. These subscript indices are always subsets of the
set $\{1,2,\ldots,L\}$, where $L$ denotes the number of descriptions,
and satisfy all the standard properties of sets. For ease of notation,
we often drop the braces and the comma whenever it is obvious. For
example, $U_{\{1\}}$, $U_{\{2\}}$, $U_{\{1,2\}}$ and $U_{\{1,2,3\}}$
are abbreviated as $U_{1},U_{2},U_{12}$ and $U_{123}$, respectively.
Next, let $\mathcal{S}$ be a set of subscript indices. Then, we use
the shorthand $\{U\}{}_{\mathcal{S}}$ to denote the set of variables
$\{U_{\mathcal{K}}:\mathcal{K}\in\mathcal{S}\}$ (e.g., if $\mathcal{S}=\left\{ \{1\},\{2\},\{1,2\}\right\} $,
then $\{U\}_{\mathcal{S}}=\{U_{\{1\}},U_{\{2\}},U_{\{1,2\}}\}=\{U_{1},U_{2},U_{12}\})$.
Note the difference between $\{U\}_{\mathcal{S}}$ and $U_{\mathcal{K}}$.
$\{U\}_{\mathcal{S}}$ is a set of variables, whereas $U_{\mathcal{K}}$
is a single variable. We use the notation in \cite{Cover-book} to
denote all standard information theoretic quantities. With a slight
abuse of notation, we use $H(\cdot)$ to denote the entropy of a discrete
random variable or the differential entropy of a continuous random
variable. 

We first give a formal definition of the $L-$channel MD problem.
A source produces an iid sequence $X^{n}=\left(X^{(1)},X^{(2)}\ldots,X^{(n)}\right)$,
taking values over a finite alphabet $\mathcal{X}$. We denote $\mathcal{L}=\{1,\ldots,L\}$.
There are $L$ encoding functions, $f_{l}(\cdot),\,\, l\in\mathcal{L}$,
which map $X^{n}$ to the descriptions $J_{l}=f_{l}(X^{n})$, where
$J_{l}$ takes values in the set $\{1,\ldots, B_{l}\}$. The rate of
description $l$ is defined as:
\begin{equation}
R_{l}=\frac{1}{n}\log_{2}(B_{l})\label{eq:defn_rate}
\end{equation}
Description $l$ is sent over channel $l$ and is either received
at the decoder error-free or is completely lost. There are $2^{L}-1$
decoding functions for each possible received combination of the descriptions
$\hat{X}_{\mathcal{K}}^{n}=\left(\hat{X}_{\mathcal{K}}^{(1)},\hat{X}_{\mathcal{K}}^{(2)}\ldots,\hat{X}_{\mathcal{K}}^{(n)}\right)=g_{\mathcal{K}}(J_{l}:l\in\mathcal{K})$,
$\forall\mathcal{K}\subseteq\mathcal{L},\mathcal{K}\neq\emptyset$, where
$\hat{X}_{\mathcal{K}}$ takes values in a finite set $\hat{\mathcal{X}}_{\mathcal{K}}$,
and $\emptyset$ denotes the null set. The distortion at the decoder when
a subset $\mathcal{K}$ of the descriptions is received is:
\begin{equation}
D_{\mathcal{K}}=E\left[\frac{1}{n}\sum_{t=1}^{n}d_{\mathcal{K}}(X^{(t)},\hat{X}_{\mathcal{K}}^{(t)})\right]\label{eq:dist_main_constraint}
\end{equation}
where $d_{\mathcal{K}}:\mathcal{X}\times\hat{\mathcal{X}}_{\mathcal{K}}\rightarrow\mathcal{R}$,
is a well-defined bounded single letter distortion measure. We say
that the RD tuple $(R_{i},D_{\mathcal{K}}:i\in\mathcal{L},\mathcal{K}\subseteq\mathcal{L},\mathcal{K}\neq\emptyset)$
is achievable if there exist $L$ encoding functions with rates $(R_{1},\ldots,R_{L})$
and $2^{L}-1$ decoding functions yielding distortions $D_{\mathcal{K}}$.
The closure of the set of all achievable RD tuples is defined as the
`\textit{$L$-channel multiple descriptions RD region}'%
\footnote{Note that this region has $L+2^{L}-1$ dimensions.%
}, denoted hereafter by $\mathcal{RD}^{L}$.

\subsection{2 - Channel MD}

\begin{figure}
\centering\includegraphics[scale=0.65]{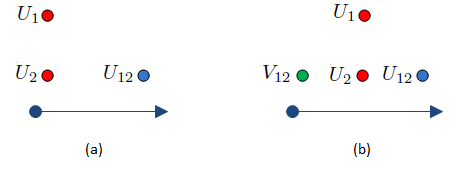}\caption{Codebook generation for EC (a) and ZB (b) coding schemes. The arrow
mark indicates the order in which the codebooks are generated. \label{fig:EGC_ZB}}
\end{figure}

\subsubsection{El-Gamal-Cover region}

The first achievable region for the 2-channel MD problem was given
by El-Gamal and Cover in 1982 \cite{EGC}. Their region is denoted
here by $\mathcal{RD}_{EC}$. It is given by the convex closure over
all tuples $(R_{1},R_{2},D_{1},D_{2},D_{12})$ for which there exist
auxiliary random variables $U_{1}$, $U_{2}$ and $U_{12}$ (defined
over arbitrary finite alphabets $\mathcal{U}_{1}$, $\mathcal{U}_{2}$
and $\mathcal{U}_{12}$, respectively), jointly distributed with $X$
and functions $\psi_{\mathcal{K}},\mathcal{K}\subseteq\{1,2\},\mathcal{K}\neq\emptyset$
such that,
\begin{eqnarray}
R_{1} & \geq & I(X;U_{1})\nonumber \\
R_{2} & \geq & I(X;U_{2})\nonumber \\
R_{1}+R_{2} & \geq & I(X;U_{1},U_{2},U_{12})+I(U_{1};U_{2})\nonumber \\
D_{\mathcal{K}} & \geq & E\left[d_{\mathcal{K}}(X,\psi_{\mathcal{K}}(U_{\mathcal{K}}))\right],\,\,\mathcal{K}\subseteq\{1,2\},\mathcal{K}\neq\emptyset\label{eq:EGC}
\end{eqnarray}
Note that the original EC description assumes $\psi_{\mathcal{K}}$
to be identity functions (i.e., $\psi_{\mathcal{K}}(U_{k})=U_{\mathcal{K}}$).
However, the two characterizations are equivalent and we use the one
described above as it is easier to relate to the $L$-channel achievable
regions. We also note that the function $\psi_{12}(\cdot)$ can, in
general, be made to depend on $U_{1},U_{2}$ and $U_{12}$. However,
 the overall RD region remains unchanged. To see this, let us say, for some joint distribution, $\psi_{12}(\cdot)$ depends on $U_{1},U_{2}$ and $U_{12}$. We can always construct a new joint distribution with the random variable $U_{12}$ set equal to $\psi_{12}(U_{1},U_{2},U_{12})$. This new joint distribution satisfies all the rate and distortion conditions in (\ref{eq:EGC}) with the new $\psi_{12}(\cdot)$ being an identity function. 

The order of codebook generation is shown in Fig. \ref{fig:EGC_ZB}(a).
As stated in the introduction, the EC region can be shown to be complete
for the 2-descriptions quadratic Gaussian MD problem \cite{Ozarow}
and for the `no excess rate' regime for general sources and distortion
measures \cite{Ahlswede}. The achievability of the above region is described as follows. $2^{nR_{1}}$
codewords of $U_{1}$ and $2^{nR_{2}}$ codewords of $U_{2}$ are
generated, each of length $n$, using the marginal distributions of
$U_{1}$ and $U_{2}$, respectively. For each pair $(u_{1}^{n},u_{2}^{n})$
of codewords, a codeword for $U_{12}$ is generated according to the
conditional density $\prod_{t=1}^{n}P_{U_{12}|U_{1},U_{2}}(u_{12}^{(t)}|u_{1}^{(t)},u_{2}^{(t)})$.
Given a sequence $X^{n}$, the encoder looks for a triplet of codewords
that are jointly typical with the observed sequence. Information at
rate $R_{1}$ is sent in description 1 and at rate $R_{2}$ in description
2. On receiving either of the the two descriptions, the decoder estimates
the source based on the corresponding codeword of $U_{1}$ or $U_{2}$.
However, if it receives both the descriptions, it estimates the source
based on $U_{12}$. El-Gamal and Cover showed that, for joint typicality
encoding, the probability of error asymptotically approaches zero
if the rates satisfy conditions in (\ref{eq:EGC}). As the distortion
measures are bounded, joint typicality assures the distortion constraints
to be satisfied.

\subsubsection{Zhang-Berger region}

Given the optimality of EC for the no-excess rate scenario and for
the quadratic Gaussian setup, it was naturally conjectured that it
is optimal in general. However, Zhang and Berger \cite{ZB} proposed
a new coding scheme and an associated rate-distortion region for the
2-channel MD problem and showed that it strictly subsumes $\mathcal{RD}_{EC}$.
In particular, they showed that for a binary symmetric source under
Hamming distortion, points outside $\mathcal{RD}_{EC}$ can be achieved
by sending a common codeword in both the descriptions. This extended
region is denoted here by $\mathcal{RD}_{ZB}$. They showed that,
any tuple is achievable for which there exist auxiliary random variables%
\footnote{Our reason for the different notational choice for $V_{12}$ will
become evident later%
} $V_{12},U_{1},U_{2},U_{12}$, jointly distributed with $X$, for
which there exist functions $\psi_{\mathcal{K}},\mathcal{K}\subseteq\mathcal{L},\mathcal{K}\neq\emptyset$
such that,

\begin{eqnarray}
R_{1} & \geq & I(X;U_{1},V_{12})\nonumber \\
R_{2} & \geq & I(X;U_{2},V_{12})\nonumber \\
R_{1}+R_{2} & \geq & 2I(X;V_{12})+I(X;U_{1},U_{2},U_{12}|V_{12})\nonumber \\
 &  & +I(U_{1};U_{2}|V_{12})\nonumber \\
D_{\mathcal{K}} & \geq & E\left[d_{\mathcal{K}}(X,\psi_{\mathcal{K}}(U_{\mathcal{K}}))\right],\,\,\mathcal{K}\subseteq\{1,2\},\mathcal{K}\neq\emptyset\label{eq:ZB}
\end{eqnarray}

The ZB region is given by the closure of all such rate-distortion
tuples. Again note that the ZB region is described differently in
\cite{ZB}, and does not include $U_{12}$ in its characterization.
However, it was recently shown in \cite{wang} that the above characterization
is equivalent to the original one in \cite{ZB}. Further note that
the above region is convex and does not require `convexification'
as is the case for $\mathcal{RD}_{EC}$.

The codebooks are generated as follows: first, $2^{nR_{12}^{''}}$
codewords of $V_{12}$, each of length $n$, are generated according
to the marginal density of $V_{12}$. Conditioned on each codeword
of $V_{12}$, $2^{nR_{1}^{'}}$ codewords of $U_{1}$ and $2^{nR_{2}^{'}}$
codewords of $U_{2}$ are generated according to the conditional densities
$\prod_{t=1}^{n}P_{U_{1}|V_{12}}(u_{1}^{(t)}|v_{12}^{(t)})$ and $\prod_{t=1}^{n}P_{U_{2}|V_{12}}(u_{2}^{(t)}|v_{12}^{(t)})$,
respectively. For each codeword tuple $(v_{12}^{n},u_{1}^{n},u_{2}^{n})$,
a codeword for $U_{12}$ is generated according to the conditional
density $\prod_{t=1}^{n}P_{U_{12}|U_{1},U_{2},V_{12}}(u_{12}^{(t)}|u_{1}^{(t)},u_{2}^{(t)},v_{12}^{(t)})$.
Similar to EC, descriptions 1 and 2 carry information about codewords
of $U_{1}$ and $U_{2}$, respectively. However, along with these
`\textit{private}' messages, both the descriptions carry a common
component, which is the codeword of $V_{12}$. Although this common
codeword introduces explicit redundancy, it helps to co-ordinate the
two messages well and therefore provides better overall efficiency.
The encoding structure of the random variables is shown in Fig. \ref{fig:EGC_ZB}(b). 

Note the functional difference of the two random variables $U_{12}$
and $V_{12}$. The codeword corresponding to $V_{12}$ is sent in
\textit{both} the descriptions, whereas the information corresponding
to the codeword of $U_{12}$ is (implicitly) \textit{split} between
the two descriptions. We therefore call $V_{12}$ the `\textit{common
random variable'} and $U_{12}$ as the `\textit{refinement random
variable}'. Random variables $U_{1}$ and $U_{2}$ form the so called
`\textit{base layer random variables}'.

\subsection{L-Channel MD}

\begin{figure}
\centering\includegraphics[scale=0.5]{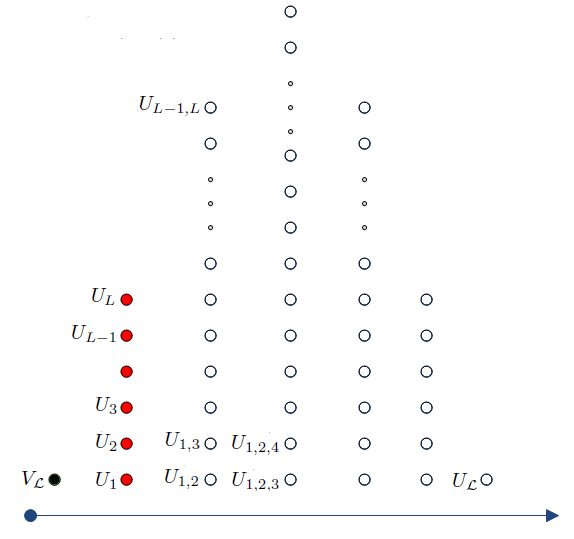}\caption{Codebook generation for VKG coding scheme. Green ($V_{\mathcal{L}}$)
indicates `common random variable'. Red ($U_{l}\,\, l\in\mathcal{L}$)
indicates `base layer random variables' and Blue ($U_{\mathcal{S}}$
$|\mathcal{S}|>1$) indicates `refinement random variables'. \label{fig:VKG-1}}
\end{figure}

An extension of $\mathcal{RD}_{EC}$ and $\mathcal{RD}_{ZB}$ to the
$L-$channel setup was proposed by Venkataramani et al. in \cite{VKG}.
The resulting region is denoted here by $\mathcal{RD}_{VKG}$, and
described as follows. Let $(V_{\mathcal{L}},\{U\}_{2^{\mathcal{L}}-\emptyset})$
be any set of $2^{L}$ random variables distributed jointly with $X$.
Then, an RD tuple is said to be achievable if there exist functions
$\psi_{\mathcal{K}}(\cdot)$ such that:
\begin{eqnarray}
\sum_{l\in\mathcal{K}}R_{l} & \geq & |\mathcal{K}|I(X;V_{\mathcal{L}})-H(\{U\}_{2^{\mathcal{K}}-\emptyset}|X,V_{\mathcal{L}})\nonumber \\
 &  & +\sum_{\mathcal{A}\subseteq\mathcal{K}}H\left(U_{\mathcal{A}}|\{U\}_{2^{\mathcal{A}}-\emptyset-\mathcal{A}}\right)\label{eq:VKG_rate}\\
D_{\mathcal{K}} & \geq & E\left[d_{\mathcal{K}}\left(X,\psi_{\mathcal{K}}(U{}_{\mathcal{K}})\right)\right]\label{eq:VKG}
\end{eqnarray}
$\forall\mathcal{K}\subseteq\mathcal{L}$. The closure of the achievable
tuples over all such $2^{L}$ random variables is $\mathcal{RD}_{VKG}$.
Here, we only present an overview of the encoding scheme. The order
of codebook generation of the auxiliary random variables is shown
in Fig. \ref{fig:VKG-1}. First, $2^{nR_{\mathcal{L}}^{''}}$ codewords
of $V_{\mathcal{L}}$ are generated using the marginal distribution
of $V_{\mathcal{L}}$. Conditioned on each codeword of $V_{\mathcal{L}}$,
$2^{nR_{l}^{'}}$ codewords of $U_{l}$ are generated according to
their respective conditional densities. Next, for each $j\in(1,\ldots,2^{n(R_{\mathcal{L}}^{''}+\sum_{l\in\mathcal{K}}R_{l}^{'})})$,
a single codeword is generated for $U_{\mathcal{K}}(j)$ conditioned
on $(v_{\mathcal{L}}(j),\{u(j)\}_{2^{\mathcal{K}}-\emptyset-\mathcal{K}})$
$\,\,\forall\mathcal{K}\subseteq\mathcal{L},\,\,|\mathcal{K}|>1$.
Note that to generate the codebook for $U_{\mathcal{K}}$, we first
need the codebooks for all $\{U\}_{2^{\mathcal{K}}-\emptyset-\mathcal{K}}$
and $V_{\mathcal{L}}$. 

On observing a typical sequence $X^{n}$, the encoder tries to find
a jointly typical codeword tuple one from each codebook. Codeword
index of $U_{l}$ (at rate $R_{l}^{'}$) is sent in description $l$.
Along with the `\textit{private}' message, each description also carries
a `\textit{shared message}' at rate $R_{\mathcal{L}}^{''}$, which
is the codeword index of $V_{\mathcal{L}}$. Hence, the rate for each
description is $R_{l}=R_{l}^{'}+R_{\mathcal{L}}^{''}$. VKG showed
that, to ensure finding a set of jointly typical codewords with the
observed sequence, the rates must satisfy (\ref{eq:VKG_rate}). It
then follows from standard arguments (see, e.g., \cite{Gamal_notes})
that, if the random variables also satisfy (\ref{eq:VKG}), then the
distortion constraints are met. Note that $V_{\mathcal{L}}$ is the
\textit{only} shared random variable. $U_{l}:l\in\mathcal{L}$ form
the base layer random variables and all $U_{\mathcal{K}}:|\mathcal{K}|\geq2$
form the refinement layers. Observe that the codebook generation follows
the order: shared layer $\rightarrow$ base layer $\rightarrow$ refinement
layer. Also, observe that the random variable $V_{\mathcal{L}}$ plays
the role of $V_{12}$ in the ZB scheme%
\footnote{It has been recently shown in \cite{wang} that the last layer of
refinement random variables ($U_{\mathcal{L}}$) can be set to be
a deterministic function of all the remaining random variables without
affecting the overall achievable region (for the 2-channel MD case,
this corresponds to setting $U_{12}$ to be a function of $V_{12},U_{1}$
and $U_{2}$). A similar result can be proven even for the proposed
coding scheme. However, we continue to use this last layer in our
theorems to avoid additional notation. %
}.

\section{Combinatorial Message Sharing\label{sec:CMS}}

In this section, we describe the proposed encoding scheme and derive
the new achievable region. To simplify notation and understanding,
we first describe the scheme for the 3-descriptions case and offer
intuitive arguments to show the achievability of the new region. We
then extend the arguments to the $L-$channel case and provide formal
proofs as part of Theorem \ref{thm:main}.

\subsection{3-Descriptions scenario\label{sub:3-Descriptions-scenario}}

\begin{figure}
\centering\includegraphics[scale=0.32]{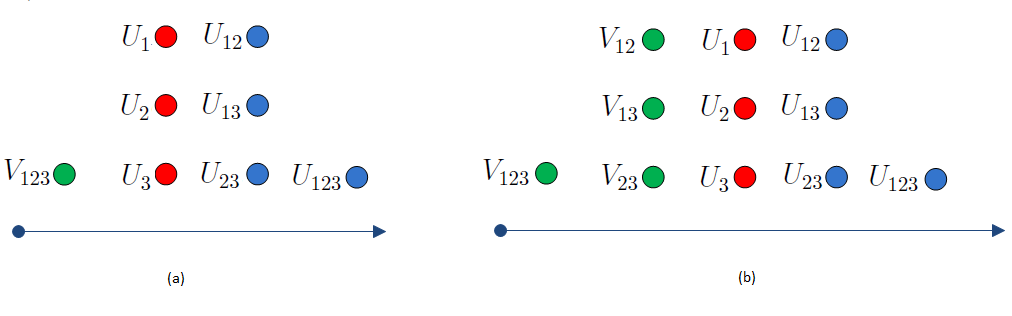}\caption{(a) Denotes the codebook generation order for VKG scheme. (b) represents
the codebook generation for the proposed coding scheme. \label{fig:3_des}}
\end{figure}

The encoding order for the 3-descriptions VKG scheme is shown in Fig.
\ref{fig:3_des}(a). Recall that the common codeword helps in coordinating
the 3 descriptions. The VKG encoding scheme employs \textit{one} common
codeword ($V_{\mathcal{L}}$) that is sent in all the $L$ descriptions.
However, when dealing with $L>2$ descriptions, restricting to a single
shared message could be suboptimal. The CMS scheme therefore allows
for `combinatorial message sharing', i.e., a common codeword is sent
in each (non-empty) subset of the descriptions.

The shared random variables are denoted by `$V$'. The base and the
refinement layer random variables are denoted by `$U$'. For the 3
descriptions scenario, we have 11 auxiliary random variables which
include 3 additional variables over the VKG scheme. These variables
are denoted by $V_{12},V_{13}$ and $V_{23}$. The codeword corresponding
to $V_{\mathcal{K}}$ is sent in \textit{all} the descriptions $l\in\mathcal{K}$.
For example, the codeword of $V_{12}$ is sent in both the descriptions
$1$ and $2$. This introduces a new layer in the encoding structure
as shown in Fig. \ref{fig:3_des} (b). This extra encoding layer provides
an additional degree of freedom in controlling the redundancy across
the messages, and hence leads to an improved rate-distortion region. 

The codebook generation is done as follows. First, the codebook for
$V_{123}$ is generated containing $2^{nR_{123}^{''}}$ independently
generated codewords, each generated according to the marginal $P(V_{123})$.
Then codebooks for $V_{12},V_{13}$ and $V_{23}$ (each containing
$2^{nR_{12}^{''}}$, $2^{nR_{13}^{''}}$ and $2^{nR_{23}^{''}}$ codewords,
respectively) are independently generated conditioned on each codeword
of $V_{123}$, according to the respective conditional densities.
Next, the base layer codebooks for $U_{l}$, $l\in\{1,2,3\}$ (each
containing $2^{nR_{l}^{'}}$ codewords) are generated conditioned
on the codewords of all $V_{\mathcal{K}}$ such that $l\in\mathcal{K}$.
For example, the codebooks for $U_{1}$ are generated conditioned
on the codewords of $V_{12},V_{13}$ and $V_{123}$. Note that each
codebook of $U_{1}$ contains $2^{nR_{1}^{'}}$ codewords and there
are such $2^{n(R_{12}^{''}+R_{13}^{''}+R_{123}^{''})}$ codebooks. 

The refinement layer codewords are generated similar to the VKG scheme.
However, the codewords for $U_{\mathcal{K}}$ are now generated conditioned
not only on the codewords of $\{U\}_{2^{\mathcal{K}}-\emptyset}$ and $V_{\mathcal{L}}$,
but also on the codewords of all $V_{\mathcal{A}}$ such that $|\mathcal{A}\cap\mathcal{K}|>0$.
For example, a single codeword for $U_{12}$ is generated conditioned
on each codeword tuple of $\{U_{1},U_{2},V_{12},V_{13},V_{23},V_{123}\}$.
Note that, overall, there are $2^{n(R_{1}^{'}+R_{2}^{'}+R_{12}^{''}+R_{13}^{''}+R_{123}^{''})}$
codewords of $U_{12}$ that are generated. Similarly codewords for
$U_{13}$, $U_{23}$ and $U_{123}$ are generated conditioned on codewords
of $\{U_{1},U_{3},V_{12},V_{13},V_{23},V_{123}\}$, $\{U_{2},U_{3},V_{12},V_{13},V_{23},V_{123}\}$
and $\{U_{1},U_{2},U_{3},V_{12},V_{13},V_{23},V_{123}\}$, respectively. 

The encoder, on observing $X^{n}$, tries to find a codeword from
the codebook of $V_{123}$ that is jointly typical with $X^{n}$.
Using typicality arguments, it is possible to show that the probability
of not finding such a codeword approaches zero if $R_{123}^{''}>I(X;V_{123})$.
Let us denote the selected codeword of $V_{123}$ by $v_{123}$. The
encoder next looks at the codebooks of $V_{12}$, $V_{13}$ and $V_{23}$,
which were generated conditioned on $v_{123}$, to find a triplet
of codewords which are jointly typical with $(X^{n},v_{123})$. It
can be shown using arguments similar to \cite{VKG,Gamal_notes} (formal
proof is given in Theorem 1), that the probability of not finding
such a triplet approaches zero if the following conditions are satisfied
$\forall\mathcal{S}\subseteq\left\{ \{1,2\},\{1,3\},\{2,3\}\right\} $:
\begin{equation}
\sum_{\mathcal{K}\in\mathcal{S}}R_{\mathcal{K}}^{''}>\sum_{\mathcal{K}\in\mathcal{S}}H\left(V_{\mathcal{K}}|V_{123}\right)-H(\{V\}_{\mathcal{S}}|V_{123},X)\label{eq:3d_1}
\end{equation}
Let us denote these codewords by $v_{12}$, $v_{13}$ and $v_{23}$.
The encoders next step is to find an index tuple $(i_{1},i_{2},i_{3})$
such that $(U_{1}(i_{1}),U_{2}(i_{2}),U_{3}(i_{3}),U_{12}(i_{1},i_{2}),$
$U_{13}(i_{1},i_{3}),U_{23}(i_{2},i_{3}),U_{123}(i_{1},i_{2},i_{3}))$
is jointly typical with $(X^{n},v_{123},v_{12},v_{13},v_{23})$. Here,
$U_{1}(i_{1})$ denotes the $i_{1}$th codeword from the codebook
of $U_{1}$ and $U_{12}(i_{1},i_{2})$ denotes the codeword of $U_{12}$
generated conditioned on $(v_{123},v_{12},v_{13},v_{23},U_{1}(i_{1}),U_{2}(i_{2}))$.
Similar notation is followed for $U_{2}(i_{2}),U_{3}(i_{3}),U_{13}(i_{1},i_{3})$,
etc. Again using arguments similar to \cite{VKG}, we can show that
the probability of not finding such an index tuple approaches zero
if $\forall\mathcal{S}\subseteq\left\{ 1,2,3\right\} $:
\begin{eqnarray}
\sum_{\mathcal{K}\in\mathcal{S}}R_{\mathcal{K}}^{'} & > & \sum_{\mathcal{K}\subseteq\mathcal{S}}H\left(U_{\mathcal{K}}|\{U\}_{2^{\mathcal{K}}-\emptyset-\{\mathcal{K}\}},\{V\}_{\mathcal{J}(\mathcal{K})}\right)\nonumber \\
 &  & -H\left(\{U\}_{2^{\mathcal{S}}-\emptyset}|V_{12},V_{13},V_{23},V_{123},X\right)\label{eq:3d_2}
\end{eqnarray}
where $\mathcal{J}(\mathcal{K})=\{\mathcal{A}:\,\,\mathcal{A}\in\{12,13,23,123\},\,\,|\mathcal{K}\bigcap\mathcal{A}|>0\}$.
We denote the codewords corresponding to this index tuple by $(u_{1},u_{2},u_{3},u_{12},u_{13},u_{23},u_{123})$.
Note that, in the above illustration, we assumed that the encoder
finds jointly typical codewords in a sequential manner, i.e., it first
finds a codeword of $V_{123}$, then finds jointly typical codewords
from the codebooks of $(V_{12},V_{13},V_{23})$ and so on. This was
done only for the ease of understanding. In Theorem \ref{thm:main},
we derive the conditions on rates for the encoder to find typical
sequences from all the codebooks jointly (at once). The conditions
on the rates for joint encoding is generally weaker (the region is
larger) than that for sequential encoding. 

The encoder sends the index of the base layer codewords in the corresponding
descriptions. It also sends the index of codewords corresponding to
$V_{\mathcal{K}}$ in \textit{all} the descriptions $l\in\mathcal{K}$.
For example, in description 1, the encoder sends the indices corresponding
to $v_{123},v_{12},v_{13}$ and $u_{1}$. Similarly, descriptions
2 and 3 carry indices corresponding to $(v_{123},v_{12},v_{23},u_{2})$
and $(v_{123},v_{13},v_{23},u_{3})$, respectively. Therefore, rate
for description $l$ is:
\begin{equation}
R_{l}=R_{l}^{'}+\sum_{\mathcal{K}\in\mathcal{J}(l)}R_{\mathcal{K}}^{''}
\end{equation}
Conditions on $R_{l}$ can be obtained by substituting bounds from
(\ref{eq:3d_1}) and (\ref{eq:3d_2}). 

The decoder, on receiving a subset of descriptions, estimates $X^{n}$
based on the refinement layer codeword corresponding to the received
subset. For example, if the decoder receives the descriptions 1 and
2, it estimates $X^{n}$ as $\psi_{12}(u_{12})$ for some function
$\psi_{12}(\cdot)$. It follows from standard arguments \cite{Gamal_notes}
that, if there exist functions $\psi_{\mathcal{K}}(\cdot)$ satisfying
the following constraints, then distortion vector $\{D_{\mathcal{K}},\,\,\forall\mathcal{K}\in2^{\mathcal{L}}-\emptyset\}$
is achievable. 
\begin{equation}
D_{\mathcal{K}}\geq E\left[d_{\mathcal{K}}\left(X,\psi_{\mathcal{K}}\left(U_{\mathcal{K}}\right)\right)\right]\label{eq:3d_4}
\end{equation}

An achievable RD region for the 3-descriptions MD setup is obtained
by taking the closure of the achievable tuples over all such 11 auxiliary
random variables, jointly distributed with $X$.

\subsection{$L$ Channel case\label{sub:L-Channel-case}}

\begin{figure}
\centering\includegraphics[scale=0.43]{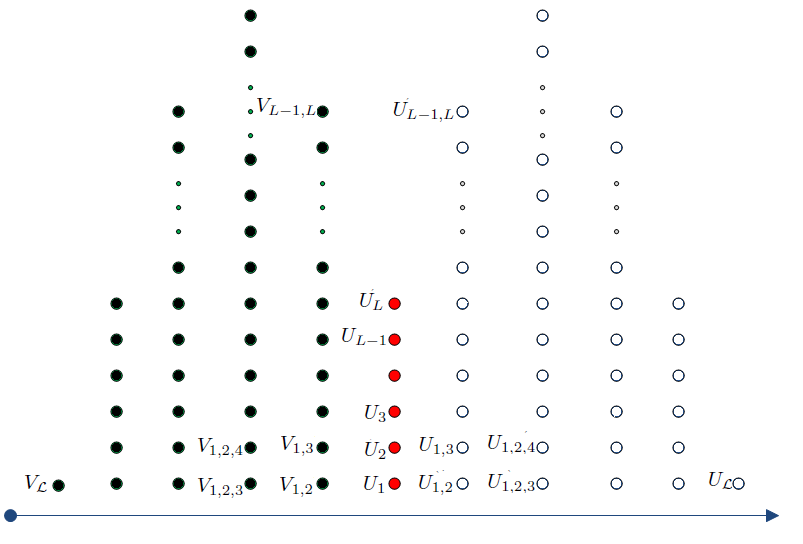}\caption{Codebook generation for the proposed coding scheme. Black ($V_{\mathcal{S}}\,\,\,|\mathcal{S}|>1$)
indicates `shared random variable'. Red ($U_{l}\,\, l\in\mathcal{L}$)
indicates `base layer random variables' and White ($U_{\mathcal{S}}\,\,\,$$|\mathcal{S}|>1$)
indicates `refinement random variables'. \label{fig:L_Channel_CMS}}
\end{figure}

The fundamental idea here is same as the 3-descriptions scenario,
albeit involving more complex notation. The codebook generation is
done in an order as shown in Fig. \ref{fig:L_Channel_CMS}. First,
the codebook for $V_{\mathcal{L}}$ is generated. Then, the codebooks
for $V_{\mathcal{K}}$, $|\mathcal{K}|=W$ are generated in the order
$W=L-1,L-2\ldots2$. This is followed by the generation of the base
layer codebooks, i.e, $U_{\mathcal{K}}$, $|\mathcal{K}|=1$. Then,
the refinement layer codebooks corresponding to $U_{\mathcal{K}}$,
$|\mathcal{K}|=W$ are generated in the order $W=2,3\ldots,L$. Each
codebook is generated conditioned on a subset of the previously generated
codewords. The specifics of codebook generation will be described
as part of the proof of Theorem \ref{thm:main}. 

Before stating the theorem, we define the following subsets of $2^{\mathcal{L}}$:
\begin{eqnarray}
\mathcal{I}_{W} & = & \{\mathcal{K}:\mathcal{K}\in2^{\mathcal{L}},\,\,|\mathcal{K}|=W\}\nonumber \\
\mathcal{I}_{W+} & = & \{\mathcal{K}:\mathcal{K}\in2^{\mathcal{L}},\,\,|\mathcal{K}|>W\}\label{eq:Iw}
\end{eqnarray}
Let $\mathcal{B}$ be any non-empty subset of $\mathcal{L}$ with
$|\mathcal{B}|\leq W$. We define the following subsets of $\mathcal{I}_{W}$
and $\mathcal{I}_{W+}$:
\begin{eqnarray}
\mathcal{I}_{W}(\mathcal{B}) & = & \{\mathcal{K}:\mathcal{K}\in\mathcal{I}_{W},\,\,\mathcal{B}\subseteq\mathcal{K}\}\nonumber \\
\mathcal{I}_{W+}(\mathcal{B}) & = & \{\mathcal{K}:\mathcal{K}\in\mathcal{I}_{W+},\,\,\mathcal{B}\subseteq\mathcal{K}\}\label{eq:Iwb}
\end{eqnarray}
We also define: 
\begin{equation}
\mathcal{J}(\mathcal{B})=\{\mathcal{K}:\,\,\mathcal{K}\in2^{\mathcal{L}},\,\,|\mathcal{B}\bigcap\mathcal{K}|>0\}
\end{equation}
Note that $\mathcal{J}(\mathcal{L})=2^{\mathcal{L}}-\emptyset$. 

Next we consider subsets of $2^{\mathcal{L}}-\emptyset$ and define some
important notation. Let $\mathcal{Q}$ be any subset of $2^{\mathcal{L}}-\emptyset$.
We denote by $[\mathcal{Q}]_{1}$ the set of all elements of $\mathcal{Q}$
of cardinality $1$, i.e.,: 
\begin{equation}
[\mathcal{Q}]_{1}=\{\mathcal{K}:\mathcal{K}\in Q,\,|\mathcal{K}|=1\}
\end{equation}
Finally, we say that $\mathcal{Q}\in\mathcal{Q}^{*}$ if it satisfies
the following property.
\begin{defn}
Let $\mathcal{Q}$ be a subset of $2^\mathcal{L}-\emptyset$. If, for every set $\mathcal{K}$ that belongs to $\mathcal{Q}$, all the sets $\mathcal{I}_{|\mathcal{K}|+}(\mathcal{K})$ also belong to $\mathcal{Q}$, then we say that $\mathcal{Q}\in\mathcal{Q}^{*}$, i.e., $\mathcal{Q}\in\mathcal{Q}^{*}$ if $\forall\mathcal{K\in\mathcal{Q}}$:

\begin{equation}
\mathcal{K}\in\mathcal{Q}\,\,\Rightarrow\,\,\,\mathcal{I}_{|\mathcal{K}|+}(\mathcal{K})\subset\mathcal{Q}\label{eq:cond_Q_main}
\end{equation}
\end{defn}

Let $(\{V\}_{\mathcal{J}(\mathcal{L})},\{U\}_{2^{\mathcal{L}}-\emptyset})$
be any set of $2^{L+1}-L-2$ random variables jointly distributed
with $X$. Let $\mathcal{Q}\subseteq2^{\mathcal{L}}-\emptyset$ such that
$\mathcal{Q}\in\mathcal{Q}^{*}$. We define: 
\begin{eqnarray}
\alpha(\mathcal{Q}) & = & \sum_{\mathcal{K}\in\mathcal{Q}-[\mathcal{Q}]_{1}}H\left(V_{\mathcal{K}}|\{V\}_{\mathcal{I}_{|\mathcal{K}|+}(\mathcal{K})}\right)\nonumber \\
 &  & +\sum_{\mathcal{K}\in2^{[\mathcal{Q}]_{1}}-\emptyset}H\left(U_{\mathcal{K}}|\{V\}_{\mathcal{I}_{1+}(\mathcal{K})},\{U\}_{2^{\mathcal{K}}-\emptyset-\mathcal{K}}\right)\nonumber \\
 &  & -H\left(\{V\}_{\mathcal{Q}-[\mathcal{Q}]_{1}},\{U\}_{2^{[\mathcal{Q}]_{1}}-\emptyset}|X\right)\label{eq:alpha_defn_1-1}
\end{eqnarray}
We follow the convention $\alpha(\emptyset)=0$. Next we state the rate-distortion
region achievable by the CMS scheme for the $L-$descriptions framework. 
\begin{thm}
\label{thm:main}Let $(\{V\}_{\mathcal{J}(\mathcal{L})},\{U\}_{2^{\mathcal{L}}-\emptyset})$
be any set of $2^{L+1}-L-2$ random variables jointly distributed
with $X$, where $U_{\mathcal{K}}$ and $V_{\mathcal{K}}$ take values
in some finite alphabets $\mathcal{U}_{\mathcal{K}}$ and $\mathcal{V}_{\mathcal{K}}$,
respectively $\forall\mathcal{K}$. Let $\mathcal{Q}^{*}$ be the
set of all subsets of $2^{\mathcal{L}}-\emptyset$ satisfying (\ref{eq:cond_Q_main})
and let $R_{\mathcal{S}}^{''},\,\,\mathcal{S}\in\mathcal{I}_{1+}$
and $R_{l}^{'},\,\, l\in\mathcal{L}$ be $2^{L}-1$ auxiliary rates
satisfying:

\textup{
\begin{equation}
\sum_{\mathcal{S}\in\mathcal{Q}-[\mathcal{Q}]_{1}}R_{\mathcal{S}}^{''}+\sum_{l\in[\mathcal{Q}]_{1}}R_{l}^{'}>\alpha(\mathcal{Q})\,\,\,\,\forall\mathcal{Q}\in\mathcal{Q}^{*}\label{eq:aux_rate_cond_thm}
\end{equation}
}Then, the RD region for the $L-$channel MD problem contains the
rates and distortions for which there exist functions $\psi_{\mathcal{K}}(\cdot)$,
such that, 
\begin{eqnarray}
R_{l} & = & R_{l}^{'}+\sum_{\mathcal{K}\in\mathcal{J}(l)}R_{\mathcal{S}}^{''}\label{eq:rate_condition_thm}\\
D_{\mathcal{K}} & \geq & E\left[d_{\mathcal{K}}\left(X,\psi_{\mathcal{K}}\left(U_{\mathcal{K}}\right)\right)\right]\label{eq:dist_condition_thm}
\end{eqnarray}
The closure of the achievable tuples over all such $2^{L+1}-L-2$
random variables is denoted by $\mathcal{RD}_{CMS}$.\end{thm}
\begin{rem}
$\mathcal{RD}_{CMS}$ can be extended to continuous random variables
and well-defined distortion measures using techniques similar to \cite{Wyner}.
We omit the details here and assume that the above region continues
to hold even for well behaved continuous random variables (for example,
a Gaussian source under MSE). 
\end{rem}

\begin{rem}
$\mathcal{RD}_{CMS}$ is convex, as a time sharing random variable
can be embedded in $V_{\mathcal{L}}$.\end{rem}
\begin{proof}
\textit{Codebook Generation} : Suppose we are given $P\left(\{v\}_{2^{\mathcal{L}}-\{(1),(2),\ldots(L)\}},\{u\}_{2^{\mathcal{L}}}|x\right)$
and $\psi_{\mathcal{K}}(\cdot)$ satisfying (\ref{eq:dist_condition_thm}).
The codebook generation begins with $V_{\mathcal{L}}$. We independently
generate $2^{nR_{\mathcal{L}}^{''}}$ codewords of $V_{\mathcal{L}}$,
denoted by $v_{\mathcal{L}}^{n}(j_{\mathcal{L}})\,\, j_{\mathcal{L}}\in\{1\ldots2^{nR_{\mathcal{L}}^{''}}\}$,
according to the density $\prod_{t=1}^{n}P_{V_{\mathcal{L}}}(v_{\mathcal{L}}^{(t)})$.
For each codeword $v_{\mathcal{L}}^{n}(j_{\mathcal{L}})$, we independently
generate $2^{nR_{\mathcal{K}}^{''}}$ codewords of $V_{\mathcal{K}}$
$\forall\mathcal{K}\in\mathcal{I}_{L-1}$, according to $\prod_{t=1}^{n}P_{V_{\mathcal{K}}|V_{\mathcal{L}}}(v_{\mathcal{K}}^{(t)}|v_{\mathcal{L}}^{(t)})$.
We denote these codewords by $v_{\mathcal{K}}^{n}(j_{\mathcal{L}},j_{\mathcal{K}})$
$j_{\mathcal{K}}\in\{1\ldots2^{nR_{\mathcal{K}}^{''}}\}$. This procedure
for generating the codebooks of the shared random variables continues.
$2^{nR_{\mathcal{K}}^{''}}$ codewords of $V_{\mathcal{K}}$ are independently
generated for each codeword tuple of $\{V\}_{\mathcal{I}_{W+}(\mathcal{K})}$
according to $\prod_{t=1}^{n}P_{V_{\mathcal{K}}|\{V\}_{\mathcal{I}_{W+}(\mathcal{K})}}(v_{\mathcal{K}}^{(t)}|\{v\}_{\mathcal{I}_{W+}(\mathcal{K})}^{(t)})$.
These codewords are denoted by $v_{\mathcal{K}}^{n}(\{j\}_{\mathcal{I}_{W+}(\mathcal{K})},j_{\mathcal{K}})\,\, j_{\mathcal{K}}\in\{1\ldots2^{nR_{\mathcal{K}}^{''}}\}$.
Note that to generate the codebooks for $V_{\mathcal{K}}$ $\forall\mathcal{K}\in\mathcal{I}_{W}$,
we need the codebooks of $V_{\mathcal{A}}\,\,\forall\mathcal{A}\in\mathcal{I}_{W+}(\mathcal{K})$.
The codebook generation follows the order indicated in Fig. \ref{fig:L_Channel_CMS}. 

Once all the codebooks of shared random variables are generated, the
codebooks for the base layer random variables are generated. For each
codeword tuple of $\{V\}_{\mathcal{I}_{1+}(l)}$, $2^{nR_{l}^{'}}$
codewords of $U_{l}$ are generated independently according to $\prod_{t=1}^{n}P_{U_{l}|(\{V\}_{\mathcal{I}_{1+}(l)})}(u_{l}^{(t)}|\{v\}_{\mathcal{I}_{1+}(l)}^{(t)})$
and are denoted by $u_{l}^{n}\left(\{j\}_{\mathcal{I}_{1+}(l)},i_{l}\right)$
$i_{l}\in\{1\ldots2^{nR_{l}^{'}}\}$. Then the codebooks for the refinement
layers are formed by assigning a single codeword $u_{\mathcal{K}}^{n}(\{j\}_{\mathcal{J}(\mathcal{K})},\{i\}_{\mathcal{K}})$
to each $\mathcal{K}\in2^{\mathcal{L}}-\{\{1\},\{2\},\ldots\{L\}\}$
and $\forall\{j\}_{\mathcal{J}(\mathcal{K})},\{i\}_{\mathcal{K}}$.
These codewords are generated according to $\prod_{t=1}^{n}P_{U_{\mathcal{K}}|\{V\}_{\mathcal{J}(\mathcal{K})},\{U\}_{2^{\mathcal{K}}-\{\mathcal{K}\}}}(u_{\mathcal{K}}^{(t)}|\{v\}_{\mathcal{J}(\mathcal{K})}^{(t)},\{u\}_{2^{\mathcal{K}}-\{\mathcal{K}\}}^{(t)})$. 

The encoder, on observing a typical sequence $X^{n}$, attempts to
find a set of codewords, one for each variable, such that they are
all jointly typical. If the encoder succeeds in finding such a set,
from the typical average lemma \cite{Gamal_notes}, it follows that
the average distortions are smaller than $D_{\mathcal{K}},\,\,\forall\mathcal{K}\in2^{\mathcal{L}}$.
However, if the encoder fails to find such a set, the average distortions
are upper bounded by $d_{max}$ (as the distortion measures are assumed
to be bounded). Hence, if the probability of finding a set of jointly
typical codewords approaches 1, the distortion conditions (\ref{eq:dist_condition_thm})
are satisfied. We show in Appendix A that if the rates $R_{\mathcal{K}}^{''}$
and $R_{l}^{'}$ satisfy conditions (\ref{eq:aux_rate_cond_thm}),
this probability approaches 1. Next, recall that the encoder sends
the codewords of $V_{\mathcal{K}}$ (at rate $R_{\mathcal{K}}^{''}$)
in all the descriptions $l\in\mathcal{K}$. It also sends the codewords
of $U_{l}$ (at rate $R_{l}^{'}$) in description $l$. Therefore
the rate of description $l$ is given by (\ref{eq:rate_condition_thm}),
proving the Theorem. 
\end{proof}
We note that the characterization of the achievable RD region in Theorem
\ref{thm:main} involves auxiliary rates $R_{l}^{'}$ and $R_{\mathcal{K}}^{''}$.
This makes it hard to prove general converse results. However deriving
an explicit characterization only in terms of $R_{l}$ is more involved
and deriving such bounds will be considered as part of future work.

\section{Strict Improvement }

\subsection{Binary Symmetric Source under Hamming Distortion Measure\label{sec:Strict-Improvement1}}

In this section we show that the CMS scheme achieves points strictly
outside $\mathcal{RD}_{VKG}$ whenever there exists a 2-description
subset for which the ZB region achieves points outside the EC region.
As a result of this theorem, it follows that for a Binary symmetric
source under Hamming distortion measure, CMS strictly outperforms
VKG.
\begin{thm}
\label{thm:ZB_imp}(i) The rate-distortion region achievable by the
CMS scheme is always at least as large as the VKG achievable region,
i.e.:
\begin{equation}
\mathcal{RD}_{VKG}\subseteq\mathcal{RD}_{CMS}
\end{equation}
(ii) $\forall L>2$, the CMS scheme achieves points that are not achievable
by VKG for any source and distortion measures for which there exists
a 2-description subset such that the Zhang-Berger scheme achieves
points outside the El-Gamal-Cover region. i.e.:
\begin{equation}
\mathcal{RD}_{VKG}\subset\mathcal{RD}_{CMS}\,\,\,\mbox{if }\mathcal{RD}_{EC}\subset\mathcal{RD}_{ZB}
\end{equation}
for some 2-description subset. Specifically, for a binary symmetric
source under Hamming distortion, CMS achieves a strictly larger rate-distortion
region compared to VKG $\forall L>2$. \end{thm}
\begin{rem}
In Section \ref{sec:Strict-Improvement2} we will show a more surprising
result that CMS achieve points strictly outside $\mathcal{RD}_{VKG}$
even in scenarios where $\mathcal{RD}_{ZB}=\mathcal{RD}_{EC}$ for
every 2-description subset. Moreover, we will show that CMS outperforms
VKG for a Gaussian source under MSE, a setting for which it is well
known that the EC and ZB regions coincide.\end{rem}
\begin{proof}
Part (i) of the theorem is a straightforward corollary of Theorem
\ref{thm:main}. It follows directly by setting $V_{\mathcal{K}}=\Phi$, 
$\forall\mathcal{K}$ such that $|\mathcal{K}|<L$ in $\mathcal{RD}_{CMS}$
(where $\Phi$ denotes a constant). Substituting in (\ref{eq:rate_condition_thm}),
and setting the auxiliary rate $R_{\mathcal{L}}^{''}=I(X;V_{\mathcal{L}})$,
we get the same conditions as (\ref{eq:VKG}). 

Part (ii) of the theorem, in fact, follows directly from Zhang and
Berger's result. To see this, without loss of generality, let descriptions
$\{1,2\}$ be the 2-description subset for which ZB achieves points
outside EC region. Let us consider a cross-section of the $L-$channel
MD problem where we only consider $R_{1}$, $R_{2}$, $D_{1}$, $D_{2}$
and $D_{12}$ and set the remaining rates to $0$ while ignoring all
the remaining distortion constraints. Clearly, the CMS scheme leads
to the ZB region as the common layer codeword $V_{12}$ is sent as
part of both descriptions 1 and 2. However, observe that, VKG leads
to EC as $V_{\mathcal{L}}$ must be set to constant to ensure $R_{3}=\cdots=R_{L}=0$.
Although this example is sufficient to prove Part (ii) of the theorem,
it is arguable that this particular cross-section of the $L-$channel
MD setup is degenerate and hence this proof has little value in practice.
We therefore provide a more general proof where all the descriptions
carry non-trivial information and all rates $R_{1},\ldots,R_{L}$
are greater than $0$.

We prove (ii) for $L=3$. Note that once we prove that the CMS scheme
achieves a strictly larger region for some $L=l>2$, then it must
be true for all $L\geq l$. Consider a 3-descriptions MD problem where
descriptions $\{1,2\}$ form the 2-description subset for which ZB
achieves points outside EC region. We denote the VKG and CMS achievable
regions by $\mathcal{RD}_{VKG}^{3}$ and $\mathcal{RD}_{CMS}^{3}$,
respectively. We now consider a particular cross-section of these
regions where we only constrain $D_{1},D_{2},D_{12}$ and $D_{13}$.
We remove the constraints on all other distortions, i.e., we allow
$D_{3},D_{23}$ and $D_{123}$ to be $\infty$. Equivalently, we can
think of a 3 descriptions MD problem with a particular channel failure
pattern, wherein only one of the following sets of descriptions can
reach the decoder reliably: $(1,2,\{1,2\},\{1,3\})$ as shown in Fig.
\ref{fig:3_des-1}. We denote the set of all achievable points for
this setup using VKG and CMS by $\tilde{\mathcal{RD}}_{VKG}^{3}$
and $\tilde{\mathcal{RD}}_{CMS}^{3}$, respectively. 

\begin{figure}
\centering\includegraphics[scale=0.45]{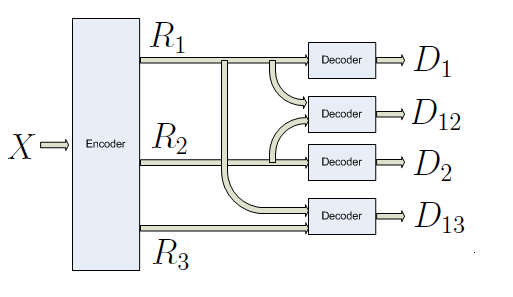}\caption{\label{fig:3_des-1}(a) Equivalent model for the cross-section considered
to prove that the CMS scheme achieves strict improvement over the
VKG scheme for a binary symmetric source under Hamming distortion
measure.}
\end{figure}

Observe that with respect to the first 2 descriptions, we have a simple
2-descriptions problem, while descriptions $\{1,3\}$ represent a
successive refinement setup. We know from the results of Zhang and
Berger that a common codeword sent among the first two descriptions
leads to a strictly better RD trade-off for these two descriptions.
However, the VKG coding scheme cannot use this common codeword unless
it is also sent as part of the third description, which is redundant
for the successive refinement framework involving descriptions $\{1,3\}$,
since description $3$ is never received without description $1$.
On the other hand, CMS allows for a unique common codeword to be sent
within every subset of the transmitted descriptions. Specifically,
it can send a common codeword which is only shared by the first two
descriptions, achieving the ZB advantage without adding redundancy
to description $3$. This argument is in fact sufficient to prove
the claim. We rewrite it formally below. 

Consider a point $P^{*}=\{R_{1},R_{2},D_{1},D_{2},D_{12}\}$ in the
ZB region that is not in the EC region. Let the rate of the common
codeword that is sent as part of both the descriptions at $P^{*}$
be denoted by $R_{12}^{*}$. Denote by $\bar{R}_{VKG}$ the following
cross-section of $\tilde{\mathcal{RD}}_{VKG}^{3}$: 
\begin{eqnarray*}
\bar{R}_{VKG}=\inf\Bigl\{ R_{3}:\{R_{1},R_{2},D_{1},D_{2},D_{12}\}=P^{*},\\
(R_{1},R_{2},R_{3},D_{1},D_{2},D_{12},D_{13})\in\tilde{\mathcal{RD}}_{VKG}^{3}\Bigr\}
\end{eqnarray*}
and for the corresponding cross-section of $\tilde{\mathcal{RD}}_{CMS}^{3}$,
the infimum $R_{3}$ rate by $\overline{R}_{CMS}(P^{*})$. As $P^{*}$
is achievable by ZB and not by EC, the constraint $\{R_{1},R_{2},D_{1},D_{2},D_{12}\}=P^{*}$
forces a common codeword to be sent in descriptions 1 and 2. To use
this common codeword, VKG requires it to be sent as part of description
3 as well leading to decoder $\{1,3\}$ receiving the same codeword
twice. However, CMS sends this common codeword only as part of descriptions
1 and 2 leading to a smaller rate for the third description. Hence
it follows that $\overline{R}_{CMS}$ is smaller than $\overline{R}_{VKG}$
by at least the rate of this common codeword which is equal to $R_{12}^{*}$.
It can be verified from \cite{ZB} that for a Binary symmetric source
under Hamming distortion measure:
\begin{eqnarray}
\overline{R}_{VKG}-\overline{R}_{CMS} & \geq & 1-H_{b}(0.30585)\nonumber \\
 & = & 0.1117\,\,\mbox{bits}
\end{eqnarray}
where $H_{b}(\cdot)$ denotes the binary entropy function. Hence,
we have shown that the CMS region is strictly larger than VKG region
whenever there exists a 2-descriptions subset for which ZB outperforms
EC. Moreover, the gap between CMS and VKG regions is at least as large
as the gap between the EC and the ZB regions for the corresponding
2-descriptions subset.
\end{proof}

\subsection{General Sources and Distortion Measures\label{sec:Strict-Improvement2}}

It is clear from Theorem \ref{thm:ZB_imp} that for any given distribution
for $X$ and distortion measures, if there exists a 2-description
subset such that $\mathcal{RD}_{EC}\subset\mathcal{RD}_{ZB}$, then
CMS strictly outperforms VKG. In this section we show a more surprising
result that the common layer codewords in CMS play a critical role
in achieving a strictly larger region for a fairly general (and larger)
class of source distributions and distortion measures. What makes
it particularly interesting is the fact that under MSE, a Gaussian
source ($X\sim\mathcal{N}(0,1)$) belongs to this class, i.e., we
will show that CMS achieves strictly larger RD region compared to
VKG for the $L-$channel quadratic Gaussian MD problem $\forall L\geq3$.
This result is in striking contrast to the corresponding results for
the 2-descriptions setting. It follows from Ozarow's results in \cite{Ozarow} that
for a Gaussian source under MSE, the EC coding scheme achieves the
complete RD region, i.e., sending a common codeword among the two
descriptions does not provide any improvement in the RD region ($\mathcal{RD}_{ZB}(\mathcal{N}(0,1))=\mathcal{RD}_{EC}(\mathcal{N}(0,1))$).
This result has lead to a natural belief that common codewords do
not play a role in the $L-$channel quadratic Gaussian MD problem
and all the explicit characterizations for the achievable regions
in the past make no use of common random variables (see for example, \cite{VKG,wang_Vishwanath,Tian_Chen_symmetric_K_descriptions,Jun_Chen_ind_central,Asymmetric_MD}).
Surprisingly, in the $L-$descriptions framework, the common codewords
in CMS can be used constructively for better coordination between
the descriptions leading to achievable points outside the VKG region.

We begin by defining a set of random variables, denoted by $\mathcal{Z}_{ZB}$,
that plays a critical role in the subsequent theorems. Let us consider
a two descriptions MD setup. For any given distortion measures, $d_{1},d_{2},d_{12}$,
we say that $X\in\mathcal{Z}_{ZB}(d_{1},d_{2},d_{12})$, if there
exists an operating point $(R_{1},R_{2},D_{1},D_{2},D_{12})$ that
belongs to $\mathcal{RD}_{ZB}$, but \textit{cannot} be achieved by
an `independent quantization' mechanism using the ZB coding scheme.
A formal definition of $\mathcal{Z}_{ZB}$ is given below.


\begin{defn}
\textbf{\label{definition:Defn_ZZB}$\mathcal{Z}_{ZB}$}: Let $\epsilon\geq0$.
Let us denote $\mathcal{RD}_{ZB}^{IQ}(\epsilon)$ to be the rate-distortion region achievable by the ZB coding scheme when the closure in (\ref{eq:ZB})
if defined only over joint distributions for the auxiliary random
variables satisfying the following conditions\footnote{The superscript of `IQ' in $\mathcal{RD}_{ZB}^{IQ}(\epsilon)$ refers to `independent quantization'}:

\begin{eqnarray}
I(U_{1};U_{2}|X,V_{12}) & < & \epsilon\nonumber \\
E\left[d_{\mathcal{K}}(X,\psi_{\mathcal{K}}(U_{\mathcal{K}}))\right] & \leq & D_{\mathcal{K}},\,\,\mathcal{K}\in\{1,2,12\}\nonumber \\
I(U_{12};X|U_{1},U_{2},V_{12}) & < & \epsilon\label{eq:ZZB-1}
\end{eqnarray}
We say that $X\in\mathcal{Z}_{ZB}(d_{1},d_{2},d_{12})$, if: 
\[
\lim_{\epsilon\rightarrow0}\mathcal{RD}_{ZB}^{IQ}(\epsilon)\subset\mathcal{RD}_{ZB}
\]
i.e., if $\lim_{\epsilon\rightarrow0}\mathcal{RD}_{ZB}^{IQ}(\epsilon)$
is strictly subsumed in $\mathcal{RD}_{ZB}$. \end{defn}
\begin{rem}
Note that if $\mathcal{RD}_{ZB}^{IQ}(0)\subset\mathcal{RD}_{ZB}$,
then clearly $X\in\mathcal{Z}_{ZB}(d_{1},d_{2},d_{12})$, i.e., if
there is a strict suboptimality in the ZB region when the closure
in (\ref{eq:ZB}) is defined only over joint distributions for the
auxiliary random variables satisfying:
\begin{eqnarray}
U_{1}\leftrightarrow & (X,V_{12}) & \leftrightarrow U_{2}\nonumber \\
E\left[d_{\mathcal{K}}(X,\psi_{\mathcal{K}}(U_{\mathcal{K}}))\right] & \leq & D_{\mathcal{K}},\,\,\mathcal{K}\in\{1,2,12\}\nonumber \\
I(U_{12};X|U_{1},U_{2},V_{12}) & = & 0\label{eq:ZZB-1-1}
\end{eqnarray}
then $X\in\mathcal{Z}_{ZB}(d_{1},d_{2},d_{12})$. Observe that the
last constraint $I(U_{12};X|U_{1},U_{2},V_{12})=0$ can be replaced
with $U_{12}=f(U_{1},U_{2},V_{12})$, without any loss of optimality. This follows from the fact that for any given joint distribution $P(X,U_{1},U_{2},V_{12},U_{12})$ that satisfies all the conditions in (\ref{eq:ZZB-1-1}), it is possible to construct a conditional distribution as follows:

\begin{eqnarray}
&Q(U_{12}|U_{1},U_{2},V_{12})=&\nonumber\\
&\begin{cases}
1 & \mbox{if }U_{12}=\arg\min_{u_{12}}E\left[d_{12}(X,u_{12})\Bigl|U_{1},U_{2},V_{12}\right]\\
0 & \mbox{otherwise}
\end{cases} &
\end{eqnarray}The joint distribution $P(X,U_1,U_2,V_{12})Q(U_{12}|U_{1},U_{2},V_{12})$ achieves the smallest distortion for $D_{12}$. Therefore, it follows that the constraint  $I(U_{12};X|U_{1},U_{2},V_{12})=0$ can be replaced
with $U_{12}=f(U_{1},U_{2},V_{12})$, without any loss of optimality. This simplified definition of $\mathcal{Z}_{ZB}$ will be used in the proof of Theorem \ref{thm:General_CMS-1}.
\end{rem}
We will show in Theorem \ref{thm:General_CMS-2} that $\forall X\in\mathcal{Z}_{ZB}(d_{1},d_{2},d_{12})$,
$\mathcal{RD}_{VKG}\subset\mathcal{RD}_{CMS}$. We will later provide
more intuition on the underlying relation between $\mathcal{Z}_{ZB}(d_{1},d_{2},d_{12})$
and the reason for strict improvement of CMS over VKG. Hereafter,
we will often drop the parenthesis and abbreviate $\mathcal{Z}_{ZB}(d_{1},d_{2},d_{12})$
by $\mathcal{Z}_{ZB}$, whenever the distortion measures are obvious. 

\begin{figure}
\centering\includegraphics[scale=0.4]{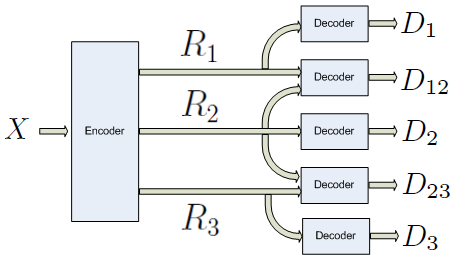}

\caption{The cross-section that we consider to prove that CMS achieves points
outside the VKG region for a general class of source and distortion
measures. CMS achieves the the complete RD region for this setup for
several distortion regimes for the quadratic Gaussian MD problem.
\label{fig:3_des_new-1}}
\end{figure}

Before stating the result we describe the particular cross-section
of the RD region that we will use to prove strict improvement in Theorem
\ref{thm:General_CMS-2}. Consider a 3-descriptions MD setup for a
source $X$ wherein we impose constraints only on distortions $(D_{1},D_{2},D_{3},D_{12},D_{23})$
and set the limit on the rest of the distortions, $(D_{13},D_{123})$
to $\infty$. This cross-section is schematically shown in Fig. \ref{fig:3_des_new-1}.
To illustrate the gains underlying CMS, let us further restrict to
the setting where we impose $D_{1}=D_{3}$ and $D_{12}=D_{23}$. We
also assume that the distortion measures $d_{1}(\cdot,\cdot)$ and
$d_{12}(\cdot,\cdot)$ are same as the distortion measures $d_{3}(\cdot,\cdot)$
and $d_{23}(\cdot,\cdot)$, respectively. The points in this cross-section,
achievable by VKG and CMS, are denoted by $\overline{\mathcal{RD}}_{VKG}(X)$
and $\overline{\mathcal{RD}}_{CMS}(X)$, respectively. We note that
the symmetric setting is considered \textit{only} for ease of understanding
the proof. The arguments can be easily extended to the asymmetric
framework. 

This particular symmetric cross-section of the 3-descriptions MD problem
is equivalent to the corresponding 2-descriptions problem, in the
sense that, one could use any coding scheme to generate bit-streams
for descriptions 1 and 2, respectively. Description 3 would then carry
a replica (exact copy) of the bits sent in description 1. Due to the
underlying symmetry in the problem setup, the distortion constraints
at all the decoders are satisfied. Hence an achievable region based
on the ZB coding scheme can be derived as follows. Let $(G_{12},F_{1},F_{2},F_{12})$
be any random variables jointly distributed with $X$ and taking values
over arbitrary finite alphabets. Then the following RD-region is achievable
for which there exist functions $(\psi_{1},\psi_{2},\psi_{12})$ such
that $R_{1}=R_{3}$, $D_{1}=D_{3}$, $D_{12}=D_{23}$ and: 
\begin{eqnarray}
R_{1} & \geq & I(X;F_{1},G_{12})\nonumber \\
R_{2} & \geq & I(X;F_{2},G_{12})\nonumber \\
R_{1}+R_{2} & \geq & 2I(X;G_{12})+H(F_{1}|G_{12})+H(F_{2}|G_{12})\nonumber \\
 &  & +H(F_{12}|F_{1},F_{2},G_{12})\nonumber \\
 &  & -H(F_{1},F_{2},F_{12}|X,G_{12})\\
D_{\mathcal{K}} & \geq & E\left[d_{\mathcal{K}}(X,\psi_{\mathcal{K}}(F_{\mathcal{K}}))\right],\,\,\mathcal{K}\subseteq\{1,2\},\mathcal{K}\neq\emptyset\nonumber\label{eq:RD(X)-1}
\end{eqnarray}
The closure of achievable RD-tuples over all random variables $(G_{12},F_{1},F_{2},F_{12})$
is denoted by $\overline{\mathcal{RD}}(X)$. In the following theorem,
we will show that $\overline{\mathcal{RD}}(X)\subseteq\overline{\mathcal{RD}}_{CMS}(X)$.
We also show that the VKG coding scheme \textit{cannot} achieve the
above RD region, i.e., $\overline{\mathcal{RD}}_{VKG}(X)\subset\overline{\mathcal{RD}}(X)$,
if $X\in\mathcal{Z}_{ZB}(d_{1},d_{2},d_{12})$. Throughout proof of
Theorem \ref{thm:General_CMS-2}, we use the notations $\overline{\mathcal{R}}_{VKG}(X,D_{1},D_{2},D_{12})$,
$\overline{\mathcal{R}}_{CMS}(X,D_{1},D_{2},D_{12})$ and $\overline{\mathcal{R}}(X,D_{1},D_{2},D_{12})$
to denote the rate-region cross-sections of $\overline{\mathcal{RD}}_{VKG}(X)$,
$\overline{\mathcal{RD}}_{CMS}(X)$ and $\overline{\mathcal{RD}}(X)$
at distortions $(D_{1},D_{2},D_{1},D_{12},D_{12})$, respectively.
We note that in Theorem \ref{thm:General_CMS-2}, we focus only on
the 3-descriptions setting. However, the results can be easily extended
to the general $L-$descriptions scenario. Also note that $\overline{\mathcal{RD}}_{CMS}(X)$
could be strictly larger than $\overline{\mathcal{RD}}(X)$, in general. 
\begin{thm}
\label{thm:General_CMS-2}(i) For the setup shown in Fig. \ref{fig:3_des_new-1}
the CMS scheme achieves $\overline{\mathcal{RD}}(X)$, i.e., $\overline{\mathcal{RD}}(X)\subseteq\overline{\mathcal{RD}}_{CMS}(X)$.

(ii) If $X\in\mathcal{Z}_{ZB}(d_{1},d_{2},d_{12})$,\textup{ }then
there exists points in $\overline{\mathcal{RD}}(X)$ that \textbf{cannot}
be achieved by the VKG encoding scheme, i.e., \textup{$\overline{\mathcal{RD}}_{VKG}\subset\overline{\mathcal{RD}}(X)$}, \end{thm}
\begin{rem}
It directly follows from (i) and (ii) that $\mathcal{RD}_{VKG}\subset\mathcal{RD}_{CMS}$
for the $L-$channel MD problem $\forall L\geq3$, if $X\in\mathcal{Z}_{ZB}$.\textit{ }

\textit{Intuition}: We first provide an intuitive argument to justify
the claim. Due to the underlying symmetry in the setup the CMS scheme
introduces common layer random variables $V_{123}=G_{12}$ and $V_{13}=F_{1}$.
It then sends the codeword of $V_{13}$ is both descriptions 1 and
3 (i.e., $U_{1}=U_{3}=V_{13}$). Hence it is sufficient for the encoder
to generate enough codewords of $U_{2}=F_{2}$ (conditioned on $V_{123}$)
to maintain joint typicality with the codewords of $V_{13}$. However,
the VKG scheme is forced to set the common layer random variable $V_{13}$
to a constant. Thus, in this case, the encoder needs to generate enough
number of codewords of $U_{2}$ (and $U_{12}$) so as to maintain
joint typicality individually with the codewords of $U_{1}$ and $U_{3}$,
which are now generated independently conditioned on $V_{123}$. It
is possible to show that this entails some excess rate on $U_{2}$,
unless $(U_{1},U_{2},U_{3})$ are pairwise independent conditioned
on $V_{123}$. However, if $X\in\mathcal{Z}_{ZB}$, then enforcing
independence of $U_{1}$ and $U_{2}$ conditioned on $V_{123}$ leads
to a strictly smaller rate-distortion region. Therefore, $\forall X\in\mathcal{Z}_{ZB}$,
the VKG scheme leads to a strictly smaller rate-distortion region
compared to the CMS scheme. \end{rem}
\begin{proof}
Part (i) of the theorem is straightforward to prove. We set,
\begin{eqnarray}
V_{123} & = & G_{12}\nonumber \\
V_{13} & = & F_{1}\nonumber \\
U_{2} & = & F_{2}\nonumber \\
U_{12}=U_{23} & = & F_{12}\nonumber \\
U_{1}= & U_{3} & =V_{13}\label{eq:Thm3_1}
\end{eqnarray}
and the rest of the random variables in Theorem \ref{thm:main} to
constants. It is easy to verify that the rate-distortion region given
in (\ref{eq:RD(X)-1}) is achievable by the CMS scheme using the above
joint distribution for the auxiliary random variables and by choosing
the auxiliary rate $R_{123}^{''}=I(X;V_{123})$.

Next, let $X\in\mathcal{Z}_{ZB}(d_{1},d_{2},d_{12})$ with respect
to the given distortion measures, i.e., there exists a strict suboptimality
in the ZB region when the closure of rates is defined only over joint
densities for the auxiliary random variables satisfying (\ref{eq:ZZB-1}).
For ease of understanding the proof, we make a simplifying assumption
that there exists at least one `corner point' in the ZB region that
is not achievable using joint densities satisfying (\ref{eq:ZZB-1}).
Specifically, let the distortions be $(D_{1},D_{2},D_{12})$, respectively,
at the three decoders. We assume that:
\begin{eqnarray}
 &  & \inf_{R_{2}:\{R_{1}=R_{X}(D_{1})\}}R_{ZB}(D_{1},D_{2},D_{12})\label{eq:corner_point_assumption}\\
 &  & <\inf_{R_{2}:\{R_{1}=R_{X}(D_{1})\}}R_{ZB}^{IQ}(D_{1},D_{2},D_{12})\nonumber 
\end{eqnarray}
i.e., if $R_{1}=R_{X}(D_{1})$, then the infimum over all $R_{2}$
required to achieve $(D_{1},D_{2},D_{12})$, using joint densities
satisfying (\ref{eq:ZZB-1}), is strictly larger than the infimum
over all $R_{2}$ achievable using ZB. However, it is important to
note that the proof can be easily extended to the general case of
strict suboptimality occurring at any intermediate point. We will
briefly discuss the extension towards the end of the proof. 

Towards proving (ii), we consider one particular boundary point of
(\ref{eq:RD(X)-1}) and show that this cannot be achieved by the VKG
encoding scheme. Let $\epsilon>0$ and $D_{1},D_{2}$ and $D_{12}$
be fixed. Define the following quantity:
\begin{eqnarray}
R_{VKG}^{*}(D_{1},D_{2},D_{12},\epsilon)=\inf \Bigl\{ R_{2}:R_{1}<R_{X}(D_{1})+\epsilon,\label{eq:Thm3_2}\\
R_{3}<R_{X}(D_{1})+\epsilon,\nonumber \\
(R_{1},R_{2},R_{3})\in\overline{\mathcal{R}}_{VKG}(X,D_{1},D_{2},D_{12}) \Bigl\} \nonumber 
\end{eqnarray}
The corresponding quantity defined for the CMS region is given by:
\begin{eqnarray}
R_{CMS}^{*}(D_{1},D_{2},D_{12},\epsilon)=\inf \Bigl\{ R_{2}:R_{1}<R_{X}(D_{1})+\epsilon,\label{eq:Thm3_2-1}\\
R_{3}<R_{X}(D_{1})+\epsilon,\nonumber \\
(R_{1},R_{2},R_{3})\in\overline{\mathcal{R}}_{CMS}(X,D_{1},D_{2},D_{12}) \Bigl\} \nonumber 
\end{eqnarray}
We will show that:
\[
\lim_{\epsilon\rightarrow0}R_{CMS}^{*}(\epsilon)<\lim_{\epsilon\rightarrow0}R_{VKG}^{*}(\epsilon)
\]
Similarly, we use the notation $R^{*}(\epsilon)$ to denote the same
quantity defined over $\overline{\mathcal{R}}(X,D_{1},D_{2},D_{12})$,
i.e.,
\begin{eqnarray}
R^{*}(D_{1},D_{2},D_{12},\epsilon)=\inf R_{2}: \Bigl\{ R_{1}<R_{X}(D_{1})+\epsilon,\label{eq:Thm3_2-1-1}\\
R_{3}<R_{X}(D_{1})+\epsilon,\nonumber \\
(R_{1},R_{2},R_{3})\in\overline{\mathcal{R}}(X,D_{1},D_{2},D_{12}) \Bigl\} \nonumber 
\end{eqnarray}

From CMS region characterization in Theorem \ref{thm:main}, $R_{CMS}^{*}(\epsilon)$
is given by the solution to the following optimization problem:
\begin{eqnarray}
R_{CMS}^{*}(\epsilon) & = & \inf\,\, \Bigl\{ I(X;V_{123})+I(U_{2};X,U_{1}|V_{123})\nonumber \\
 &  & +I(U_{12};X|V_{123},U_{1},U_{2}) \Bigl\} \label{eq:Thm3_3}
\end{eqnarray}
where the infimum is over all joint distributions $P(V_{123},U_{1},U_{2},U_{12}|X)$,
and conditional distributions $P(V_{123},U_{1}|X)$ for which there
exists a function $\psi_{1}(\cdot)$ such that:
\begin{eqnarray}
I(X;V_{123},U_{1}) & < & R_{X}(D_{1})+\epsilon\nonumber \\
E\left[d_{1}(X,\psi_{1}(U_{1}))\right] & \leq & D_{1}\label{eq:CMS_13_constraints}
\end{eqnarray}
i.e., $(V_{123},U_{1})$ achieves a close-to-optimal reconstruction
of $X$ at $D_{1}$ and $P(U_{12},U_{2}|X,U_{1},V_{123})$ is any
distribution for which there exists functions $\psi_{2}(\cdot)$ and
$\psi_{12}(\cdot)$ satisfying the distortion constraints for $D_{2}$
and $D_{12}$, respectively. 

We next specialize and restate $\overline{\mathcal{RD}}_{VKG}(X)$
for the considered cross-section. Let $(V_{123},U_{1},U_{2},U_{3},U_{12},U_{23},U_{13},U_{123})$
be any random variables jointly distributed with $X$ taking values
on arbitrary alphabets. Then, the following rate-distortion tuples
are achievable by the VKG scheme for which there exist functions $\psi_{1}(\cdot),\psi_{2}(\cdot),\psi_{3}(\cdot),\psi_{12}(\cdot),\psi_{23}(\cdot)$,
such that:
\begin{eqnarray}
R_{i} & \geq & I(X;U_{i},V_{123})\,\, i\in\{1,2,3\}\nonumber \\
R_{1}+R_{2} & \geq & 2I(X;V_{123})+I(U_{1};U_{2}|V_{123})\nonumber \\
 &  & +I(X;U_{1},U_{2},U_{12}|V_{123})\nonumber \\
R_{2}+R_{3} & \geq & 2I(X;V_{123})+I(U_{2};U_{3}|V_{123})\nonumber \\
 &  & +I(X;U_{2},U_{3},U_{23}|V_{123})\nonumber \\
R_{1}+R_{3} & \geq & 2I(X;V_{123})+I(U_{1};U_{3}|V_{123})\nonumber \\
 &  & +I(X;U_{1},U_{3},U_{13}|V_{123})\nonumber \\
R_{1}+R_{2}+R_{3} & \geq & 3I(X;V_{123})+\sum_{i=1}^{3}H(U_{i}|V_{123})\nonumber \\
 &  & +\sum_{\mathcal{K}\in\{12,23,13\}}H(U_{\mathcal{K}}|\{U\}_{\{k\in\mathcal{K}\}},V_{123})\nonumber \\
 &  & +H(U_{123}|\{U\}_{\{1,2,3,12,13,23\}},V_{123})\nonumber \\
 &  & -H(\{U\}_{\{1,2,3,12,13,23,123\}}|X,V_{123})\label{eq:RD_VKG_bar-1}
\end{eqnarray}
\begin{eqnarray}
E\left(d_{1}(X,\psi_{1}(U_{1}))\right) & \leq & D_{1}\nonumber \\
E\left(d_{2}(X,\psi_{2}(U_{2}))\right) & \leq & D_{2}\nonumber \\
E\left(d_{1}(X,\psi_{3}(U_{3}))\right) & \leq & D_{1}\nonumber \\
E\left(d_{12}(X,\psi_{12}(U_{12}))\right) & \leq & D_{12}\nonumber \\
E\left(d_{12}(X,\psi_{23}(U_{23}))\right) & \leq & D_{12}\label{eq:RD_VKG_Dist}
\end{eqnarray}
The closure of the above rate-distortion region over all joint distributions
leads to $\overline{\mathcal{RD}}_{VKG}$. Observe that there are
no distortion constraints imposed on $D_{13}$ and $D_{123}$. This
allows us to simplify the region further, without any loss of optimality.
First, the random variables $U_{13}$ and $U_{123}$ can be set to
constants. This is because they do not appear in any of the distortion
constraints and setting them to constants leads to a larger rate-region
for the given distribution over all other random variables. This step
simplifies the constraints on $R_{1}+R_{3}$ and $R_{1}+R_{2}+R_{3}$
to: 
\begin{eqnarray}
R_{1}+R_{3} & \geq & 2I(X;V_{123})+H(U_{1}|V_{123})\nonumber \\
 &  & +H(U_{3}|V_{123})-H(U_{1},U_{3}|X,V_{123})\nonumber \\
R_{1}+R_{2}+R_{3} & \geq & 3I(X;V_{123})+\sum_{i=1}^{3}H(U_{i}|V_{123})\nonumber \\
 &  & +\sum_{\mathcal{K}\in\{12,23\}}H(U_{\mathcal{K}}|\{U\}_{\{k\in\mathcal{K}\}},V_{123})\nonumber \\
 &  & -H(\{U\}_{\{1,2,3,12,23\}}|X,V_{123})\label{eq:Simplification_1}
\end{eqnarray}
As a next step of simplification, we restrict the closure of the region
to be taken only over joint densities, $P(V_{123},U_{1},U_{2},U_{3},U_{12},U_{23}|X)$,
that satisfy the following constraint:
\begin{eqnarray}
 & P(U_{12},U_{23}|X,V_{123},U_{1},U_{2},U_{3})=\nonumber \\
 & P(U_{12}|X,V_{123},U_{1},U_{2})P(U_{23}|X,V_{123},U_{2},U_{3})\label{eq:U_1223_Mark}
\end{eqnarray}
We note that this restriction does not lead to any loss in $\overline{\mathcal{RD}}_{VKG}(X)$
for this particular cross-section. This is because, for any given
joint distribution $Q(X,V_{123},U_{1},U_{2},U_{3},U_{12},U_{23})$,
we can construct another joint distribution that satisfies (\ref{eq:U_1223_Mark})
and leads to a larger rate-distortion region. To see this, consider
the joint distribution constructed as:
\begin{eqnarray}
 & Q(X,V_{123},U_{1},U_{2},U_{3})\times\nonumber \\
 & Q(U_{12}|X,V_{123},U_{1},U_{2})\times Q(U_{23}|X,V_{123},U_{2},U_{3})
\end{eqnarray}
which satisfies (\ref{eq:U_1223_Mark}). Observe that this joint distribution
satisfies the same distortion constraints as $Q(\cdot)$. Moreover,
it leads to the same rate constraints as $Q(\cdot)$, except for the
constraint on $R_{1}+R_{2}+R_{3}$. However, the constraint it imposes
on $R_{1}+R_{2}+R_{3}$ is weaker than that imposed by $Q(\cdot)$
and hence the rate-distortion region is larger than that achievable
by $Q(\cdot)$. Therefore, it is sufficient to consider only those
joint distributions that satisfy (\ref{eq:U_1223_Mark}) for $\overline{\mathcal{RD}}_{VKG}(X)$. 

We next impose the constraints $R_{1}<R_{X}(D_{1})+\epsilon$ and
$R_{3}<R_{X}(D_{1})+\epsilon$ in (\ref{eq:Simplification_1}). This
enforces the conditional density $P(V_{123},U_{1},U_{3}|X)$ to satisfy
the following constraints:
\begin{eqnarray}
I(X;V_{123},U_{1}) & < & R_{X}(D_{1})+\epsilon\nonumber \\
I(X;V_{123},U_{3}) & < & R_{X}(D_{1})+\epsilon\nonumber \\
E\left[d_{1}(X,\psi_{1}(V_{123},U_{1}))\right] & \leq & D_{1}\nonumber \\
E\left[d_{1}(X,\psi_{3}(V_{123},U_{3}))\right] & \leq & D_{1}\nonumber \\
I(U_{1};U_{3}|X,V_{123}) & < & \epsilon\label{eq:VKG_13_constraints}
\end{eqnarray}
where the last condition is required to satisfy the constraint on
$R_{1}+R_{3}$ in (\ref{eq:Simplification_1}). Therefore, using (\ref{eq:Simplification_1})
and (\ref{eq:U_1223_Mark}), $R_{VKG}^{*}(\epsilon)$ can be written
as the solution to the following optimization problem:
\begin{eqnarray}
R_{VKG}^{*}(\epsilon)=\inf\,\, \Bigl\{ I(X;V_{123})+I(U_{2};U_{1},U_{3},X|V_{123})\nonumber \\
+I(X;U_{12}|U_{1},U_{2},V_{123})+I(X;U_{23}|U_{2},U_{3},V_{123})\Bigl\}\label{eq:R_s_VKG}
\end{eqnarray}
where the infimum is over all conditional densities $P(V_{123},U_{1},U_{2},U_{3},U_{12},U_{23}|X)$
satisfying (\ref{eq:VKG_13_constraints}) for which there exist functions
$\psi_{2}(\cdot),\psi_{12}(\cdot),\psi_{23}(\cdot)$ satisfying the
distortion constraints in (\ref{eq:RD_VKG_Dist}). 

Summarizing what we have so far:

\begin{eqnarray}
\lim_{\epsilon\rightarrow0}R_{VKG}^{*}(\epsilon) & \leq & R_{VKG}^{*}(\epsilon)\nonumber \\
 & = & \inf\,\, \Bigl\{ I(X;V_{123})+I(U_{2};U_{1},U_{3},X|V_{123})\nonumber \\
 &  & +I(X;U_{12}|V_{123},U_{1},U_{2})\nonumber \\
 &  & +I(X;U_{23}|V_{123},U_{2},U_{3})\Bigl\}\label{eq:RVKG_CMS_1}\\
\lim_{\epsilon\rightarrow0}R_{CMS}^{*}(\epsilon) & \leq & R_{CMS}^{*}(\epsilon)\nonumber \\
 & = & \inf\,\, \Bigl\{ I(X;V_{123})+I(U_{2};U_{1},X|V_{123})\nonumber \\
 &  & +I(X;U_{12}|V_{123},U_{1},U_{2}) \Bigl\} \label{eq:RVKG_RCMS}
\end{eqnarray}
where the infimum for $R_{VKG}^{*}(\epsilon)$ is subject to (\ref{eq:VKG_13_constraints})
and the infimum for $R_{CMS}^{*}(\epsilon)$ is subject to (\ref{eq:CMS_13_constraints}).

We first note that the optimization objective for $R_{VKG}^{*}(\epsilon)$
is always greater than or equal to the optimization objective for
$R_{CMS}^{*}(\epsilon)$, for a given joint distribution over the
random variables. This follows from the following standard information
theoretic inequalities:
\begin{eqnarray*}
I(U_{2};U_{1},U_{3},X|V_{123}) & \geq & I(U_{2};U_{1},X|V_{123})\\
I(X;U_{23}|V_{123},U_{2},U_{3}) & \geq & 0
\end{eqnarray*}
Moreover, the constraints for $R_{CMS}^{*}(\epsilon)$ are a subset
of the constraints for $R_{VKG}^{*}(\epsilon)$. Therefore, for $\lim_{\epsilon\rightarrow0}R_{VKG}^{*}(\epsilon)$
to be equal to $\lim_{\epsilon\rightarrow0}R_{CMS}^{*}(\epsilon)$,
there should exist a small enough $\epsilon>0$ for every $\lambda>0$,
such that:
\begin{equation}
R_{VKG}^{*}(\epsilon)-R_{CMS}^{*}(\epsilon)<\lambda\label{eq:eps_delta_constraint}
\end{equation}
We will next show that this cannot happen if $X\in\mathcal{Z}_{ZB}(d_{1},d_{2},d_{12})$.
i.e., we will show that $R_{VKG}^{*}(\epsilon)-R_{CMS}^{*}(\epsilon)$
is lower bounded by a value that is strictly greater than $0$. 

We prove this claim by contradiction. We assume that it is possible
to find a small enough $\epsilon>0$ for every $\lambda>0$ such that
(\ref{eq:eps_delta_constraint}) holds and then arrive at a contradiction.
Let $\epsilon_{1}>0$. Let $P(V_{123},U_{1},U_{2},U_{3},U_{12},U_{23}|X)$
be any conditional distribution that satisfies the constraints in
(\ref{eq:VKG_13_constraints}) and satisfies:
\begin{eqnarray}
I^{P}(X;V_{123})+I^{P}(U_{2};U_{1},U_{3},X|V_{123})\nonumber \\
+I^{P}(X;U_{12}|V_{123},U_{1},U_{2})\nonumber \\
+I^{P}(X;U_{23}|V_{123},U_{2},U_{3})-R_{VKG}^{*}(\epsilon) & < & \epsilon_{1}
\end{eqnarray}
where the superscript $P$ has been added to emphasize that the mutual
information is with respect to the distribution $P$. The distribution
$P$ achieves the optimization objective in (\ref{eq:RVKG_CMS_1}), within $\epsilon_{1}$ of $R_{VKG}^{*}(\epsilon)$. For (\ref{eq:eps_delta_constraint})
to hold, the following condition must be satisfied by all conditional
distributions $\tilde{P}(V_{123},U_{1},U_{2},U_{12}|X)$ that satisfy
the constraints in (\ref{eq:CMS_13_constraints}): 

\begin{eqnarray}
 &  & I^{P}(X;V_{123})-I^{\tilde{P}}(X;V_{123})\nonumber \\
 &  & +I^{P}(U_{2};U_{1},U_{3},X|V_{123})-I^{\tilde{P}}(U_{2};U_{1},X|V_{123})\nonumber \\
 &  & +I^{P}(X;U_{12}|V_{123},U_{1},U_{2})-I^{\tilde{P}}(X;U_{12}|V_{123},U_{1},U_{2})\nonumber \\
 &  & +I^{P}(X;U_{23}|V_{123},U_{2},U_{3})<\epsilon_{1}+\lambda\label{eq:More_eps_delta}
\end{eqnarray}
We will next show that this leads to a contradiction if $X\in\mathcal{Z}_{ZB}(d_{1},d_{2},d_{12})$.
Towards proving that, we set $\tilde{P}(V_{123},U_{1},U_{2},U_{12}|X)=P(V_{123},U_{1},U_{2},U_{12}|X)$.
Observe that $P$ is a valid candidate for $\tilde{P}$ as the constraints
in (\ref{eq:CMS_13_constraints}) are a subset of (\ref{eq:VKG_13_constraints}).
With $\tilde{P}=P$, for (\ref{eq:More_eps_delta}) to hold, we need:
\begin{eqnarray}
I^{P}(U_{2};U_{3}|V_{123},U_{1},X) & < & \epsilon_{1}+\lambda\nonumber \\
I^{P}(X;U_{23}|V_{123},U_{2},U_{3}) & < & \epsilon_{1}+\lambda\label{eq:contra_cond}
\end{eqnarray}
First, observe that the constraint $I^{P}(U_{2};U_{3}|V_{123},U_{1},X)<\epsilon_{1}+\lambda$
implies that $H^{P}(U_{3}|V_{123},U_{1},X)<H^{P}(U_{3}|V_{123},U_{1},U_{2},X)+\epsilon_{1}+\lambda$.
However, as $P$ satisfies (\ref{eq:VKG_13_constraints}), we have
$H^{P}(U_{3}|V_{123},X)<H^{P}(U_{3}|V_{123},X,U_{1})+\epsilon$. On
substituting, we get:
\begin{equation}
H^{P}(U_{3}|V_{123},X)<H^{P}(U_{3}|V_{123},U_{1},U_{2},X)+\epsilon_{1}+\lambda+\epsilon\label{eq:contra_cond2}
\end{equation}
i.e.,
\begin{equation}
I^{P}(U_{3};U_{2},U_{1}|V_{123},X)<\epsilon_{1}+\lambda+\epsilon\label{eq:contra_cond2-1}
\end{equation}
which implies that:
\begin{equation}
I^{P}(U_{3};U_{2}|V_{123},X)<\epsilon_{1}+\lambda+\epsilon\label{eq:contra_cond2-1-1}
\end{equation}
Hence, from (\ref{eq:contra_cond}) and (\ref{eq:contra_cond2-1-1}),
for (\ref{eq:eps_delta_constraint}) to hold, we need the following
constraints to be satisfied for the joint distribution $P$ that achieves
close to optimality in (\ref{eq:RVKG_CMS_1}):
\begin{eqnarray}
I^{P}(U_{3};U_{2}|V_{123},X) & < & \epsilon_{1}+\lambda+\epsilon\nonumber \\
I^{P}(X;U_{23}|V_{123},U_{2},U_{3}) & < & \epsilon_{1}+\lambda+\epsilon\nonumber \\
E^{P}\left[d_{\mathcal{K}}(X,\psi_{\mathcal{K}}(U_{\mathcal{K}}))\right] & \leq & D_{\mathcal{K}},\,\,\mathcal{K}\in\{2,3,23\}
\end{eqnarray}
Under the limits of $\epsilon\rightarrow0$, $\epsilon_{1}\rightarrow0$
and $\lambda\rightarrow0$, if $X\in\mathcal{Z}_{ZB}(d_{1},d_{2},d_{12})$,
the above constraints imply that there is a strict sub-optimality
in the ZB region. From our assumption in (\ref{eq:corner_point_assumption}),
it immediately follows that:
\begin{equation}
\lim_{\epsilon\rightarrow0}R_{VKG}^{*}(\epsilon)>\lim_{\epsilon\rightarrow0}R^{*}(\epsilon)
\end{equation}
where $R^{*}(\epsilon)$ is defined in (\ref{eq:Thm3_2-1-1}). However,
from (i) of the theorem, it follows that:
\begin{equation}
\lim_{\epsilon\rightarrow0}R^{*}(\epsilon)\geq\lim_{\epsilon\rightarrow0}R_{CMS}^{*}(\epsilon)
\end{equation}
This leads to a contradiction as it implies that (\ref{eq:eps_delta_constraint})
is not satisfied. It therefore follows that if $X\in\mathcal{Z}_{ZB}(d_{1},d_{2},d_{12})$,
CMS achieves a strictly larger RD region compared to VKG, proving
the theorem.

We note that if the strict sub-optimality in (\ref{eq:corner_point_assumption})
exists at any other boundary point of the ZB region, the above proof
holds by changing the definitions of $R_{VKG}^{*}(\epsilon)$, $R_{CMS}^{*}(\epsilon)$
and $R^{*}(\epsilon)$ accordingly. We only consider this particular
corner point in this proof for ease of understanding. 
\end{proof}
\textbf{Discussion:} A direct consequence of the above theorem is
that, if $X\in\mathcal{Z}_{ZB}$, then the common layer codewords
of CMS are needed to achieve strict improvement in the region, i.e.,
if $X\in\mathcal{Z}_{ZB}$, $\mathcal{RD}_{VKG}\Bigr|_{V_{\mathcal{L}}=\Phi}\subseteq\mathcal{RD}_{VKG}\subset\mathcal{RD}_{CMS}$,
where $\mathcal{RD}\Bigr|_{V_{\mathcal{L}}=\Phi}$ denotes the VKG
region when the common layer random variable (denoted by $V_{\mathcal{L}}$)
is set to a constant %
\footnote{Note that setting $V_{\mathcal{L}}$ to a constant in VKG is equivalent
to setting all the common layer random variables to constants in CMS. %
}. In fact, it is possible to show that, whenever $X\in\mathcal{Z}_{EC}$,
$\mathcal{RD}_{VKG}\Bigr|_{V_{\mathcal{L}}=\Phi}\subset\mathcal{RD}_{CMS}$,
where $\mathcal{Z}_{EC}$ is similarly defined as the set of all random
variables $X$ for which there exists an operating point, with respect
to the given distortion measures, that \textit{cannot} be achieved
by an `independent quantization' mechanism using the EC coding scheme,
i.e., if there exists an operating point in the EC region that \textit{cannot}
be achieved by the EC coding scheme using a joint density for the
auxiliary random variables satisfying:
\begin{eqnarray}
P(U_{1},U_{2}|X) & = & P(U_{1}|X)P(U_{2}|X)\nonumber \\
E\left[d_{\mathcal{K}}(X,\psi_{\mathcal{K}}(U_{\mathcal{K}}))\right] & \leq & D_{\mathcal{K}}\,\,\,\mathcal{K}\in\{1,2,12\}\nonumber \\
U_{12} & = & f(U_{1},U_{2})\label{eq:ZEC}
\end{eqnarray}
where $f$ is any deterministic function%
\footnote{A formal definition of $\mathcal{Z}_{EC}$ would be in lines of the
definition of $\mathcal{Z}_{ZB}$ in Definition \ref{definition:Defn_ZZB}.
We state the simpler version here for brevity.%
}. Note that the set $\mathcal{Z}_{ZB}$ is a subset of $\mathcal{Z}_{EC}$.
It is possible to construct distributions and distortion measures
such that $\mathcal{Z}_{ZB}$ is a strict subset of $\mathcal{Z}_{EC}$,
but as the construction does not provide any useful insights into
the MD problem, we choose to omit the details. Also observe that if
$X\notin\mathcal{Z}_{EC}$, the concatenation of two independent optimal
quantizers is optimal in achieving a joint reconstruction. While this
condition could be satisfied for specific values of $D_{1},D_{2}$
and $D_{12}$, it is seldom achieved \textit{for all} values of $(D_{1},D_{2},D_{12})$.
Though such sources are of some theoretical interest, the multiple
descriptions encoding for such sources is degenerate. Hence with some
trivial exceptions, it can be asserted that the common layer codewords
in CMS can be used to achieve a strictly larger region (compared to
not using any common codewords) for all sources and distortion measures,
$\forall L\geq3$. 

We note in passing that, for the two descriptions setting, Zhang and
Berger studied an important cross-section of the problem in \cite{MD_No_Excess_Marginal_Rate}
called the ``no excess marginal rate'' setting, where $R_{1}=R_{X}(D_{1})$
and $R_{2}=R_{X}(D_{2})$. They derived upper and lower bounds on
the achievable $D_{12}$ and showed that the gap is negligible for
a binary source under Hamming distortion measure. The constraints
$R_{1}=R_{3}=R_{X}(D_{1})$ imposed in the proof of Theorem \ref{thm:General_CMS-2}
resemble the constraints imposed in \cite{MD_No_Excess_Marginal_Rate}
and hence the results in \cite{MD_No_Excess_Marginal_Rate} may seem
relevant to the setting considered in this paper. However, it is important
to note that in the proof of Theorem \ref{thm:General_CMS-2}, $R_{2}$
can be greater than $R_{X}(D_{2})$. As the interaction is only between
descriptions $\{1,2\}$ and $\{2,3\}$, the problem considered here
is not directly related to the `no excess marginal rate' case and
requires new tools to prove the results.

At a first glance, it is tempting to conclude from Ozarow's results in \cite{Ozarow}
that under MSE a Gaussian random variable belongs to both $\mathcal{Z}_{EC}$
and $\mathcal{Z}_{ZB}$. However, the formal proof is non-trivial.
We will formally show in the next section that this is indeed the
case and an independent quantization scheme (with or without a common
codeword) leads to strict suboptimality for the 2-descriptions quadratic
Gaussian MD problem.

\section{Gaussian MSE Setting\label{sec:Gaussian-MSE-Setting}}

In this section we present a series of new results for the $L-$descriptions
quadratic Gaussian MD problem. We will first show that CMS achieves
a strictly larger RD region, by proving that under MSE, a Gaussian
source belongs to $\mathcal{Z}_{ZB}$. We then use similar encoding
principles to derive the complete rate region in several asymmetric
distortion regimes. Throughout this section, we will assume that $X\sim\mathcal{N}(0,1)$
and the distortion at all the decoders is the squared error, i.e., $d(x_{1},x_{2})=(x_{1}-x_{2})^{2}$. 

Before stating the results formally, we review Ozarow's result for
the 2-descriptions MD setting. Ozarow showed that the complete region
for the 2-descriptions Gaussian MD problem can be achieved using a
`correlated quantization' scheme which imposes the following joint
distribution for $(U_{1},U_{2},U_{12})$ in the EC scheme:
\begin{eqnarray}
U_{1}=X+W_{1}\nonumber \\
U_{2}=X+W_{2}\label{eq:Ozarow_Result-1}
\end{eqnarray}
$U_{12}=E(X|U_{1},U_{2})$, where $W_{1}$ and $W_{2}$ are zero mean
Gaussian random variables independent of $X$ with covariance matrix
$K_{W_{1}W_{2}}$, and the functions $\psi_{\mathcal{K}}(U{}_{\mathcal{K}})$
are given by the respective MSE optimal estimators, e.g., $\psi_{1}(U_{1})=E\left[X|U_{1}\right]$.
The covariance matrix $K_{W_{1}W_{2}}$ is set to satisfy all the
distortion constraints. Specifically, the optimal $K_{W_{1}W_{2}}$
is given by:
\begin{equation}
K_{W_{1}W_{2}}=\left[\begin{array}{cc}
\sigma_{1}^{2} & \rho_{12}\sigma_{1}\sigma_{2}\\
\rho_{12}\sigma_{1}\sigma_{2} & \sigma_{2}^{2}
\end{array}\right]\label{eq:Ozarow_K-1}
\end{equation}
where $\sigma_{i}^{2}=\frac{D_{i}}{1-D_{i}}\,\, i\in\{1,2\}$ and
the optimal $\rho_{12}$, denoted by $\rho_{12}^{*}$, is given by
(see \cite{Zamir}):
\begin{eqnarray}
\rho_{12}^{*} & = & \begin{cases}
-\frac{\sqrt{\pi D_{12}^{2}+\gamma}-\sqrt{\pi D_{12}^{2}}}{(1-D_{12})\sqrt{D_{1}D_{2}}} & D_{12}\leq D_{12}^{max}\\
0 & D_{12}\geq D_{12}^{max}
\end{cases}\nonumber \\
\gamma & = & (1-D_{12})\Bigl[(D_{1}-D_{12})(D_{2}-D_{12})\nonumber \\
 &  & +D_{12}D_{1}D_{2}-D_{12}^{2}\Bigr]\nonumber \\
D_{12}^{max} & = & D_{1}D_{2}/(D_{1}+D_{2}-D_{1}D_{2})\nonumber \\
\pi & = & (1-D_{1})(1-D_{2})\label{eq:other_defn-1}
\end{eqnarray}
We denote the complete Gaussian-MSE $L$-descriptions region by $\mathcal{RD}_{G}^{L}$.
The characterization of $\mathcal{RD}_{G}^{2}$ is given in \cite{EGC}
(see also \cite{Zamir}) and we omit restating it explicitly here
for brevity.

\subsection{Strict Improvement for the Quadratic Gaussian Case}

Equipped with these results, we next show that CMS achieves points
outside the VKG region, by proving that a Gaussian source under MSE
belongs to $\mathcal{Z}_{ZB}$. 
\begin{thm}
\label{thm:General_CMS-1}(i) CMS achieves the \textbf{complete} RD
region for the symmetric 3-descriptions quadratic Gaussian setup shown
in Fig. \ref{fig:3_des_new-1}. 
\begin{eqnarray}
R_{1} & \geq & \frac{1}{2}\log\frac{1}{D_{1}}\nonumber \\
R_{2} & \geq & \frac{1}{2}\log\frac{1}{D_{2}}\nonumber \\
R_{1}+R_{2} & \geq & \frac{1}{2}\log\frac{1}{D_{1}D_{2}}+\delta\nonumber \\
R_{1} & = & R_{3}\nonumber \\
D_{12}=D_{23} &  & D_{3}=D_{1}\label{eq:Symm_RD}
\end{eqnarray}
where $\delta=\delta(D_{12},D_{1},D_{2})$ is defined by:
\begin{equation}
\delta=\frac{1}{2}\log\left(\frac{1}{1-(\rho_{12}^{*})^{2}}\right)\label{eq:delta_defn}
\end{equation}
where $\rho_{12}^{*}$ is defined in (\ref{eq:other_defn-1}).

(ii) The VKG encoding scheme cannot achieve all the points in the
region, i.e., \textup{$\overline{\mathcal{RD}}_{VKG}\subset\overline{\mathcal{RD}}_{CMS}$. }\end{thm}
\begin{rem}
It follows from (i) and (ii) that $\mathcal{RD}_{VKG}\subset\mathcal{RD}_{CMS}$
for the $L-$channel quadratic Gaussian MD problem $\forall L>2$. \end{rem}
\begin{proof}
To prove achievability using CMS, we set,
\begin{eqnarray}
V_{13} & = & X+W_{1}\nonumber \\
U_{2} & = & X+W_{2}\nonumber \\
U_{1}= & U_{3} & =V_{13}\nonumber \\
U_{12}= & U_{23} & =E\left[X|V_{12},U_{2}\right]\label{eq:Cor1_1}
\end{eqnarray}
where $W_{1}$ and $W_{2}$ are zero mean Gaussian random variables
independent of $X$ with covariance matrix $K_{W_{1}W_{2}}$ and the
functions $\psi_{\mathcal{K}}(U{}_{\mathcal{K}})$ are given by the
respective MSE optimal estimators, e.g., $\psi_{1}(U_{1})=E\left[X|U_{1}\right]$.
We set all the remaining auxiliary random variables to constants.
It follows directly that the rate-distortion region given in (\ref{eq:Symm_RD})
is achievable by the CMS scheme. Following the footsteps of Ozarow in \cite{Ozarow}, it is  straightforward to show that the
above region is also complete for the symmetric setup considered. 

Note that Ozarow's results suggest that if $D_{12}\leq D_{12}^{max}$,
then an `independent quantization' scheme does not achieve the smallest
$D_{12}$. It might be tempting to conclude that (ii) follows directly
from this observation and Theorem \ref{thm:General_CMS-2}. However,
a closer inspection reveals that the above argument holds only after
we prove the optimality of Gaussian codebooks under `independent quantization'
mechanism. We relegate the proof of (ii) to Appendix B as the underlying
principles are quite orthogonal to the rest of the paper. \end{proof}
\begin{rem}
Note that, as $\mathcal{Z}_{ZB}\subseteq\mathcal{Z}_{EC}$, a Gaussian
source under MSE belongs to $\mathcal{Z}_{EC}$. Hence, the `correlated
quantization' scheme (an extreme special case of VKG) which has been
proven to be complete for several cross-sections of the $L-$descriptions
quadratic Gaussian MD problem \cite{Jun_Chen_ind_central}, is strictly
suboptimal in general. 
\end{rem}

\subsection{Points on the Boundary - 3-Descriptions Setting\label{sub:Points-on-the}}

In this section we show that CMS achieves the complete RD region for
several cross-sections of the general quadratic Gaussian $L-$channel
MD problem. We again begin with the 3-descriptions case and then extend
the results to the $L$ channel framework. Recall the setup shown
in Fig. \ref{fig:3_des-1}, i.e, a cross-section of the general 3-descriptions
rate-distortion region wherein we impose constraints only on distortions
$(D_{1},D_{2},D_{3},D_{12},D_{23})$ and set the rest of the distortions,
$(D_{13},D_{123})$ to $1$. Here we consider the general asymmetric
case, i.e. $D_{1}\neq D_{3}$ and $D_{12}\neq D_{23}$ and show that
the CMS scheme achieves the complete rate region in several distortion
regimes. 

In the following theorem, without loss of generality we assume that
$D_{1}\leq D_{3}$. If $D_{3}\leq D_{1}$, then the theorem holds
by interchanging `$1$' and `$3$' everywhere. Let $D_{12}$ be any
distortion such that $D_{12}\leq\min\{D_{1},D_{2}\}$. We define $D_{23}^{*}=D_{23}^{*}(D_{1},D_{2},D_{3},D_{12})$
as:
\begin{equation}
D_{23}^{*}=\frac{\sigma_{2}^{2}\sigma_{3}^{2}\left(1-\rho{}^{2}\right)}{\sigma_{2}^{2}\sigma_{3}^{2}\left(1-\rho{}^{2}\right)+\sigma_{2}^{2}+\sigma_{3}^{2}-2\sigma_{2}\sigma_{3}\rho}\label{eq:defn_D23}
\end{equation}
where $\sigma_{i}^{2}=\frac{D_{i}}{1-D_{i}}$ $i\in\{2,3\}$ and
\begin{equation}
\rho=\rho_{12}^{*}\frac{\sigma_{1}}{\sigma_{3}}\label{eq:defn_rho}
\end{equation}
where $\rho_{12}^{*}$ is defined in (\ref{eq:other_defn-1}). In
the following theorem, we will show that CMS achieves the complete
rate-region if $D_{23}=D_{23}^{*}$.
\begin{thm}
\label{thm:Sum_Rate_Gauss}For the setup shown in Fig. \ref{fig:3_des-1},
let $D_{1}\leq D_{3}$. Then,

(i) CMS achieves the complete rate-region if:
\begin{eqnarray}
D_{23} & = & D_{23}^{*}(D_{1},D_{2},D_{3},D_{12})\label{eq:sum_rate_thm1}
\end{eqnarray}
where $D_{23}^{*}$ is defined in (\ref{eq:defn_D23}). The rate region
is given by:
\begin{eqnarray}
R_{i} & \geq & \frac{1}{2}\log\frac{1}{D_{i}}\,\, i\in\{1,2,3\}\nonumber \\
R_{1}+R_{2} & \geq & \frac{1}{2}\log\frac{1}{D_{1}D_{2}}+\delta(D_{1},D_{2},D_{12})\nonumber \\
R_{2}+R_{3} & \geq & \frac{1}{2}\log\frac{1}{D_{2}D_{3}}+\delta(D_{2},D_{3},D_{23})\label{eq:sum_rate_thm2}
\end{eqnarray}
where $\delta(\cdot)$ is defined in (\ref{eq:delta_defn}).

(ii) Moreover, CMS achieves the minimum sum-rate if one of the following
hold:

(a) For a fixed $D_{12}$, $D_{23}\geq D_{23}^{*}(D_{1},D_{2},D_{3},D_{12})$

(b) For a fixed $D_{23}$, $D_{12}\in\{D_{12}:\delta(D_{2},D_{3},D_{23})\geq\delta(D_{1},D_{2},D_{12})\}$\end{thm}
\begin{rem}
We note that the above rate region \textit{cannot} be achieved by
VKG. We omit the details of the proof here as it can be proved in
same lines as the proof of Theorem \ref{thm:General_CMS-1}.
\end{rem}

\begin{rem}
An achievable rate-distortion region can be derived for general distortions
using the encoding principles we derive as part of this proof. However,
it is hard to prove outer bounds if the conditions in (\ref{eq:sum_rate_thm1})
are not satisfied and hence we omit stating the results explicitly
here. 
\end{rem}

\begin{rem}
Both CMS and VKG achieve the complete rate region when $D_{12}\geq D_{12}^{max}$
and $D_{23}\geq D_{23}^{max}$, where $D_{12}^{max}$ and $D_{23}^{max}$
are defined in (\ref{eq:other_defn-1}). It can be easily verified that in
this case an independent quantization scheme is optimal and the complete
achievable rate-region is given by $R_{i}\geq\frac{1}{2}\log\frac{1}{D_{i}},\,\, i\in\{1,2,3\}$. 
\end{rem}

\begin{rem}
It follows from the above theorem that CMS achieves the minimum sum-rate whenever
$D_{12}=D_{23}$ for any $D_{1},D_{3}$. \end{rem}
\begin{proof}
We begin with the proof of (i). The proof of (ii) then follows almost
directly from (i). First we show the converse, which is quite obvious.
Conditions on $R_{i}$ follow from the converse to the source coding
theorem. Conditions on $R_{1}+R_{2}$ and $R_{2}+R_{3}$ follow from
Ozarow's results in \cite{Ozarow}, to achieve $(D_{1},D_{2},D_{12})$ using descriptions
$\{1,2\}$ and to achieve $(D_{2},D_{3},D_{23})$ using descriptions
$\{2,3\}$ at the respective decoders.

We next prove that CMS achieves the rate region in (\ref{eq:sum_rate_thm2})
if (\ref{eq:sum_rate_thm1}) holds. We first give an intuitive argument
to explain the encoding scheme. Description 3 carries an RD-optimal
quantized version of $X$ (which achieves distortion $D_{3}$). Description
1 carries all the bits embedded in description 3 along with `refinement
bits' which assist in achieving distortion $D_{1}\leq D_{3}$. This
entails no loss in optimality as a Gaussian source is successively
refinable under MSE \cite{Successive_Refinement}. Description 2 then
carries a quantized version of the source which is correlated with
the information in descriptions 1 and 3. We will show that if $D_{23}=D_{23}^{*}(D_{1},D_{2},D_{3},D_{12})$,
then the correlations can be set such that description 2 is optimal
with respect to both descriptions 1 and 3. 

Formally, to achieve the rate region in (\ref{eq:sum_rate_thm2}),
we set the auxiliary random variables in the CMS coding scheme as
follows:
\begin{eqnarray}
V_{13} & = & X+W_{1}+W_{3}\nonumber \\
U_{3} & = & V_{13}\nonumber \\
U_{1} & = & X+W_{1}\nonumber \\
U_{2} & = & X+W_{2}\nonumber \\
U_{12}=\Phi &  & U_{23}=\Phi\label{eq:Sum_rate_pf_3}
\end{eqnarray}
and the functions $\psi(\cdot)$ as the respective MSE optimal estimators,
where $W_{1},W_{2},W_{3}$ are zero mean Gaussian random variables
independent of $X$ with a covariance matrix:
\begin{equation}
K_{W_{1}W_{2}W_{3}}=\left[\begin{array}{ccc}
\tilde{\sigma}_{1}^{2} & \rho_{12}\tilde{\sigma}_{1}\tilde{\sigma}_{2} & 0\\
\rho_{12}\tilde{\sigma}_{1}\tilde{\sigma}_{2} & \tilde{\sigma}_{2}^{2} & 0\\
0 & 0 & \tilde{\sigma}_{3}^{2}
\end{array}\right]\label{eq:Sum_rate_pf_4}
\end{equation}
where $\tilde{\sigma}_{1}^{2}=\sigma_{1}^{2}=\frac{D_{1}}{1-D_{1}}$,
$\tilde{\sigma}_{2}^{2}=\sigma_{2}^{2}=\frac{D_{2}}{1-D_{2}}$, $\tilde{\sigma}_{3}^{2}=\sigma_{3}^{2}-\sigma_{1}^{2}=\frac{D_{3}}{1-D_{3}}-\frac{D_{1}}{1-D_{1}}$.
The correlation coefficient $\rho_{12}$ is set to achieve distortion
$D_{12}$, i.e. $\rho_{12}=\rho_{12}^{*}$ defined in (\ref{eq:other_defn-1}).
Let us denote by $W_{13}=W_{1}+W_{3}$. Observe that the encoding
for descriptions 2 and 3 resembles Ozarow's correlated quantization
scheme with $U_{2}=X+W_{2}$ and $U_{3}=X+W_{13}$. Let us denote
the correlation coefficient between $W_{2}$ and $W_{13}$ be $\rho$.
We have the following equation relating $\rho_{12}$ and $\rho$ (which
is equivalent to (\ref{eq:defn_rho})):
\begin{equation}
\rho_{12}^{*}\tilde{\sigma}_{1}=\rho\sqrt{\tilde{\sigma}_{1}^{2}+\tilde{\sigma}_{3}^{2}}\label{eq:rho_defn2}
\end{equation}
Note that the above relation is derived using the independence of
$W_{3}$ and $(W_{1},W_{2})$, which follows from our choice of $K_{W_{1}W_{2}W_{3}}$.
Hence the minimum distortion $D_{23}$ achievable using the above
choice for the joint density of the auxiliary random variables is
given by:
\begin{eqnarray}
D_{23} & = & \mbox{Var}(X|U_{2},U_{3},V_{13})\nonumber \\
 & = & \mbox{Var}(X|U_{2},V_{13})\nonumber \\
 & = & D_{23}^{*}\label{eq:min_D23}
\end{eqnarray}

We next derive the rates required by this choice of $K_{W_{1}W_{2}W_{3}}$.
Application of Theorem \ref{thm:main} using the above joint
density leads to the following achievable rate region for any given
distortions $D_{1},D_{2},D_{3},D_{12},D_{23}$:
\begin{eqnarray*}
R_{13}^{''} & \geq & \frac{1}{2}\log\frac{1}{D_{3}}\\
R_{2}^{'} & \geq & \frac{1}{2}\log\frac{1}{D_{2}}\\
R_{1}^{'}+R_{13}^{''} & \geq & \frac{1}{2}\log\frac{1}{D_{1}}\\
R_{2}^{'}+R_{13}^{''} & \geq & H(V_{13})+H(U_{2})-H(V_{13},U_{2}|X)\\
 & = & H(U_{3})+H(U_{2})-H(U_{3},U_{2}|X)\\
 & = & \frac{1}{2}\log\frac{1}{D_{3}D_{2}}+\frac{1}{2}\log\left(\frac{1}{1-\rho{}^{2}}\right)\\
 & = & \frac{1}{2}\log\frac{1}{D_{3}D_{2}}+\delta(D_{2},D_{3},D_{23}^{*})
\end{eqnarray*}
\begin{eqnarray*}
R_{1}^{'}+R_{2}^{'}+R_{13}^{''} & \geq & H(V_{13})+H(U_{1}|V_{13})+H(U_{2})\\
 &  & -H(U_{1},V_{13},U_{2}|X)\\
 & = & I(X;U_{1},V_{13})+I(U_{2};X,U_{1},V_{13})\\
 & =^{(a)} & I(X;U_{1})+I(X;U_{2})\\
 &  & +I(U_{2};U_{1},V_{13}|X)\\
 & = & I(X;U_{1})+I(X;U_{2})\\
 &  & +I(U_{2};U_{1},U_{3}|X)\\
 & =^{(b)} & I(X;U_{1})+I(X;U_{2})\\
 &  & +I(W_{2};W_{1},W_{1}+W_{3})\\
 & = & I(X;U_{1})+I(X;U_{2})+I(W_{2};W_{1})\\
 &  & +I(W_{2};W_{3}|W_{1})\\
 & =^{(c)} & I(X;U_{1})+I(X;U_{2})+I(W_{2};W_{1})\\
 & = & \frac{1}{2}\log\frac{1}{D_{1}D_{2}}+\frac{1}{2}\log\left(\frac{1}{1-(\rho_{12}^{*})^{2}}\right)\\
 & = & \frac{1}{2}\log\frac{1}{D_{1}D_{2}}+\delta(D_{1},D_{2},D_{12})
\end{eqnarray*}
\begin{eqnarray}
R_{1} & = & R_{13}^{''}+R_{1}^{'}\nonumber \\
R_{2} & = & R_{2}^{'}\nonumber \\
R_{3} & = & R_{13}^{''}\label{eq:CMS_Gauss_ach_region}
\end{eqnarray}
where $(a)$ follows from the Markov chain $X\leftrightarrow U_{1}\leftrightarrow V_{13}$,
$(b)$ from the independence of $X$ and $(W_{1},W_{2},W_{3})$ and
$(c)$ from the independence of $W_{3}$ and $(W_{1},W_{2})$. 

At a first glance, it might be tempting to conclude that the region
for the tuple $(R_{1},R_{2},R_{3})$ in (\ref{eq:CMS_Gauss_ach_region})
is equivalent to the region given by (\ref{eq:sum_rate_thm2}). This
is not the case in general as the equations in (\ref{eq:CMS_Gauss_ach_region})
have an implicit constraint on the auxiliary rates $R_{13}^{''},R_{1}^{'},R_{2}^{'}\geq0$.
However, we will show that if $D_{3}\geq D_{1}$, then the two regions
are indeed equivalent. We denote the rate region given in (\ref{eq:sum_rate_thm2})
by $\mathcal{R}$ and the region in (\ref{eq:CMS_Gauss_ach_region})
by $\mathcal{R}^{*}$. Clearly, $\mathcal{R}^{*}\subseteq\mathcal{R}$,
as any $(R_{1},R_{2},R_{3})$ that satisfies (\ref{eq:CMS_Gauss_ach_region})
also satisfies (\ref{eq:sum_rate_thm2}). We need to show that $\mathcal{R}^{*}\supseteq\mathcal{R}$.
Towards proving this claim, note that both $\mathcal{R}$ and $\mathcal{R}^{*}$
are convex regions bounded by hyper-planes. Hence, it is sufficient
for us to show that all the corner points of $\mathcal{R}$ lie in
$\mathcal{R}^{*}$. Clearly, $\mathcal{R}$ has 6 corner points denoted
by $P_{ijk}$ $i,j,k\in\{1,2,3\}$ defined as:
\begin{eqnarray}
P_{ijk} & = & \{r_{i},r_{j},r_{k}\}\nonumber \\
r_{i} & = & \min R_{i}\nonumber \\
r_{j} & = & \min_{R_{i}=r_{i}}R_{j}\nonumber \\
r_{k} & = & \min_{R_{i}=r_{i},R_{j}=r_{j}}R_{k}
\end{eqnarray}
To prove $\mathcal{R}^{*}\supseteq\mathcal{R}$, we need to prove
that every corner point $(r_{1},r_{2},r_{3})\in\mathcal{R}$ is achieved
by some non-negative $(R_{13}^{''},R_{1}^{'},R_{2}^{'},R_{1},R_{2},R_{3})\in\mathcal{R}^{*}$
such that $R_{i}=r_{i},\,\, i\in\{1,2,3\}$. We set $R_{13}^{''}=R_{3}=r_{3}$
and $R_{2}^{'}=R_{2}=r_{2}$ and show that we can always find $R_{1}^{'}\geq0$
satisfying (\ref{eq:CMS_Gauss_ach_region}) such that $R_{1}=R_{1}^{'}+R_{13}^{'}=r_{1}$.
Let us first consider the points $P_{213}=P_{231}$ given by:
\begin{eqnarray}
r_{1} & = & \frac{1}{2}\log\frac{1}{D_{1}}+\delta(D_{1},D_{2},D_{12})\nonumber \\
r_{2} & = & \frac{1}{2}\log\frac{1}{D_{2}}\nonumber \\
r_{3} & = & \frac{1}{2}\log\frac{1}{D_{3}}+\delta(D_{2},D_{3},D_{23})
\end{eqnarray}
This can be achieved by using the following auxiliary rates, $R_{2}^{'}=r_{2}$,
$R_{13}^{''}=r_{3}$ and 
\begin{eqnarray}
R_{1}^{'} & = & \frac{1}{2}\log\frac{D_{3}}{D_{1}}+\delta(D_{1},D_{2},D_{12})\nonumber \\
 &  & -\delta(D_{2},D_{3},D_{23})\nonumber \\
 & = & \frac{1}{2}\log\frac{(1-D_{1})D_{3}-(\rho_{12}^{*})^{2}D_{1}(1-D_{3})}{(1-D_{1})D_{1}(1-(\rho_{12}^{*})^{2})}
\end{eqnarray}
It is easy to verify that $R_{1}^{'}\geq0$ if $D_{3}\geq D_{1}$.
Hence $P_{213}=P_{231}\in\mathcal{R}^{*}$. Let us next consider the
points $P_{132}=P_{312}$ given by: 
\begin{eqnarray}
r_{1} & = & \frac{1}{2}\log\frac{1}{D_{1}}\nonumber \\
r_{2} & = & \frac{1}{2}\log\frac{1}{D_{2}}\nonumber \\
 &  & +\max\{\delta(D_{1},D_{2},D_{12}),\delta(D_{2},D_{3},D_{23})\}\nonumber \\
r_{3} & = & \frac{1}{2}\log\frac{1}{D_{3}}
\end{eqnarray}
Again it is straightforward to show that $(R_{13}^{''},R_{1}^{'},R_{2}^{'})=(r_{3},r_{1}-r_{3},r_{2})$
belongs to $\mathcal{R}^{*}$. Finally, we consider the remaining
two points $P_{123}$ and $P_{321}$. $P_{123}$ is given by:
\begin{eqnarray}
r_{1} & = & \frac{1}{2}\log\frac{1}{D_{1}}\nonumber \\
r_{2} & = & \frac{1}{2}\log\frac{1}{D_{2}}+\delta(D_{1},D_{2},D_{12})\nonumber \\
r_{3} & = & \frac{1}{2}\log\frac{1}{D_{3}}\nonumber \\
 &  & +\left(\delta(D_{2},D_{3},D_{23})-\delta(D_{1},D_{2},D_{12})\right)^{+}
\end{eqnarray}
where $x^{+}=\max\{x,0\}$. Consider the following auxiliary rates:
$R_{13}^{''}=r_{3}$, $R_{2}^{'}=r_{2}$ and $R_{1}^{'}=\frac{1}{2}\log\frac{D_{3}}{D_{1}}$.
Clearly the first three constraints in (\ref{eq:CMS_Gauss_ach_region})
are satisfied by these auxiliary rates. The following inequalities
prove that the last two constraints are also satisfied by these rates
and hence $P_{123}\in\mathcal{R}^{*}$.
\begin{eqnarray*}
R_{2}^{'}+R_{13}^{''} & = & r_{2}+r_{3}\\
 & = & \frac{1}{2}\log\frac{1}{D_{2}D_{3}}+\delta(D_{1},D_{2},D_{12})\\
 &  & +\left(\delta(D_{2},D_{3},D_{23})-\delta(D_{1},D_{2},D_{12})\right)^{+}\\
 & \geq & \frac{1}{2}\log\frac{1}{D_{2}D_{3}}+\delta(D_{2},D_{3},D_{23})
\end{eqnarray*}
\begin{eqnarray}
R_{2}^{'}+R_{1}^{'}+R_{13}^{''} & = & \frac{1}{2}\log\frac{1}{D_{1}D_{2}}+\delta(D_{1},D_{2},D_{12})\nonumber \\
 &  & +\left(\delta(D_{2},D_{3},D_{23})-\delta(D_{1},D_{2},D_{12})\right)^{+}\nonumber \\
 & \geq & \frac{1}{2}\log\frac{1}{D_{1}D_{2}}+\delta(D_{1},D_{2},D_{12})
\end{eqnarray}
Next consider $P_{321}$:
\begin{eqnarray}
r_{1} & = & \frac{1}{2}\log\frac{1}{D_{1}}\nonumber \\
 &  & +\left(\delta(D_{1},D_{2},D_{12})-\delta(D_{2},D_{3},D_{23})\right)^{+}\nonumber \\
r_{2} & = & \frac{1}{2}\log\frac{1}{D_{2}}+\delta(D_{2},D_{3},D_{23})\nonumber \\
r_{3} & = & \frac{1}{2}\log\frac{1}{D_{3}}
\end{eqnarray}
Using same arguments as before, it can be shown that $P_{321}\in\mathcal{R}^{*}$
by using the following auxiliary rates: $R_{13}^{''}=r_{3}$, $R_{2}^{'}=r_{2}$
and $R_{1}^{'}=\frac{1}{2}\log\frac{D_{3}}{D_{1}}+\left(\delta(D_{1},D_{2},D_{12})-\delta(D_{2},D_{3},D_{23})\right)^{+}$.
Therefore, it follows that $\mathcal{R}=\mathcal{R}^{*}$ and hence
CMS achieves the complete rate region, proving (i). 

We next prove (ii)(a). It follows from (i) that the following rate
point is achievable $\forall D_{23}\geq D_{23}^{*}$:
\begin{eqnarray}
\left\{ R_{1},R_{2},R_{3}\right\}  & = & \Bigl\{\frac{1}{2}\log\frac{1}{D_{1}},\frac{1}{2}\log\frac{1}{D_{2}}+\delta(D_{1},D_{2},D_{12}),\nonumber \\
 &  & \frac{1}{2}\log\frac{1}{D_{3}}\Bigr\}
\end{eqnarray}
Also observe that $\forall D_{23}\geq D_{23}^{*}$, $\delta(D_{1},D_{2},D_{12})\geq\delta(D_{2},D_{3},D_{23})$
and hence a lower bound to the sum rate is $\frac{1}{2}\log\frac{1}{D_{1}D_{2}D_{3}}+\delta(D_{1},D_{2},D_{12})$.
Therefore the above point achieves the minimum sum rate $\forall D_{23}\geq D_{23}^{*}$. 

The proof of (ii)(b) follows similarly by noting that if $D_{12}\in\{D_{12}:\delta(D_{2},D_{3},D_{23})\geq\delta(D_{1},D_{2},D_{12})\}$,
the minimum sum rate is given by $\frac{1}{2}\log\frac{1}{D_{1}D_{2}D_{3}}+\delta(D_{2},D_{3},D_{23})$
which is achieved by the point:
\begin{eqnarray}
\left\{ R_{1},R_{2},R_{3}\right\}  & = & \Bigl\{\frac{1}{2}\log\frac{1}{D_{1}},\frac{1}{2}\log\frac{1}{D_{2}}+\delta(D_{2},D_{3},D_{12}),\nonumber \\
 &  & \frac{1}{2}\log\frac{1}{D_{3}}\Bigr\}
\end{eqnarray}
This proves the theorem. 
\end{proof}
It is interesting to observe that the optimal encoding scheme introduces
common codewords (creates an interaction) between descriptions 1 and
3, even though these two descriptions are never received simultaneously
at the decoder. While common codewords typically imply redundancy
in the system, in this case, introducing them allows for better co-ordination
between the descriptions leading to a smaller rate for the common
branch. We will use similar principles in the following section for
the $L-$descriptions framework and show that the CMS scheme achieves
the complete RD region for several distortion regimes. 

\begin{figure}
\centering\includegraphics[scale=0.185]{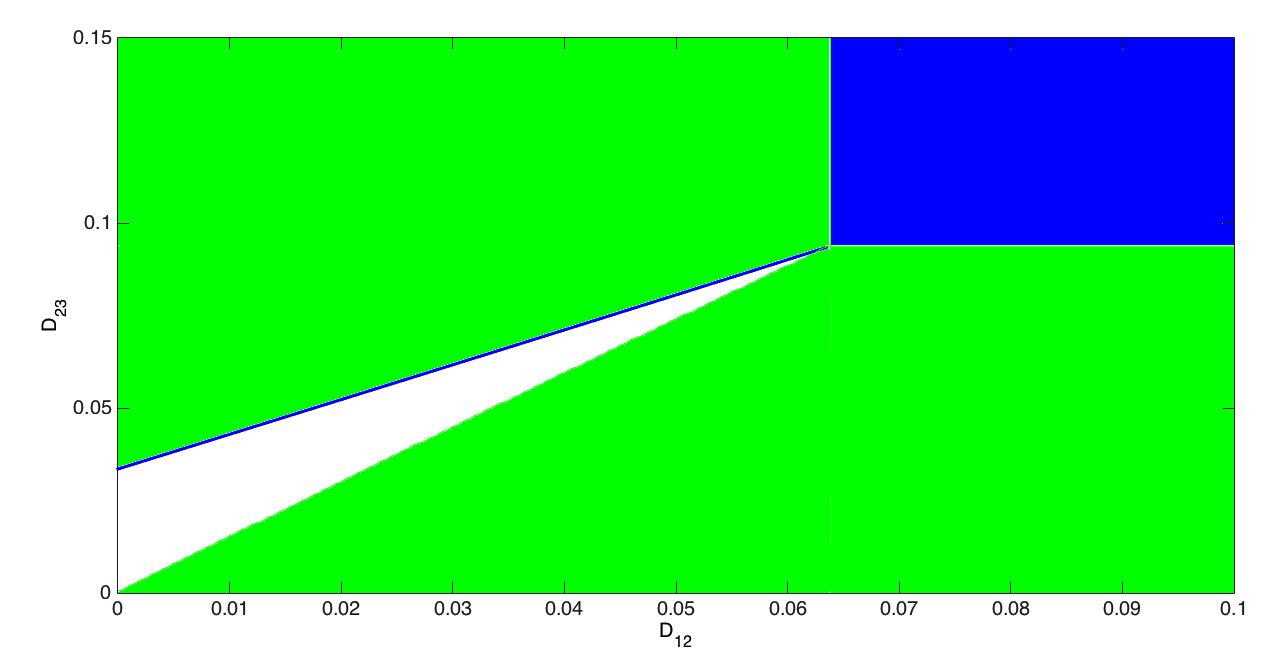}

\caption{Example: This figure denotes the regime of distortions wherein the
CMS scheme achieves the complete rate region and the minimum sum rate.
Here $D_{1}=0.1,D_{2}=0.15$ and $D_{3}=0.2$. The blue points correspond
to the region of distortions wherein the CMS scheme achieves the complete
rate-region and the green points represent the region where the CMS
scheme achieves the minimum sum rate \label{fig:Example:-This-figure}.}

\end{figure}

\begin{example}
We consider an asymmetric setting where $D_{1}=0.1,D_{2}=0.15$ and
$D_{3}=0.2$. Fig. \ref{fig:Example:-This-figure} shows the regime
of distortions where CMS achieves the complete rate-region and minimum
sum rate. The blue region corresponds to the set of distortion pairs
$(D_{12},D_{23})$ wherein the CMS rate-region is complete. The green
region denotes the minimum sum rate points. It is clearly evident
from the figure that CMS achieves the minimum sum rate for a fairly
large regime of distortions. 
\end{example}

\subsection{Points on the Boundary - $L$-Descriptions\label{sub:Points-on-the_L}}

We next extend the coding scheme described in the previous section
to the $L-$descriptions framework. Observe that, in all the setups
considered so far, we have used only a single common layer random
variable. We will see, in the proof of Theorem \ref{thm:boundary_points_L}, 
that multiple common layer random variables are necessary to achieve
the complete rate region of several cross-sections of the rate-distortion
region. We first describe the particular cross-section we are interested
in and then state the result as part of Theorem \ref{thm:boundary_points_L}.
We impose only the following distortion constraints : $D_{1},D_{2},\ldots,D_{L},D_{12},D_{13},\ldots,D_{1L}$
and limit the remaining distortions to 1. This essentially corresponds
to an $L-$descriptions framework wherein each of the descriptions
are either received individually or the first description is received
along with another description from the set $\{2,3,\ldots,L\}$. Observe
that the 3-descriptions framework considered in the previous section
is a special case of this $L-$channel setup%
\footnote{Note that `description 1' here plays the role of `descriptions 2'
considered in the 3-descriptions setting in the previous section.%
}. In the following theorem, we assume without loss of generality that
$D_{2}\leq\ldots\leq D_{L}$. The results follow accordingly for any
other permutation of the ordering. 
\begin{thm}
\label{thm:boundary_points_L}Consider a cross-section of the $L-$descriptions
quadratic Gaussian rate-distortion region, wherein we only impose
distortions $D_{1},D_{2},\ldots,D_{L},D_{12},D_{13},\ldots,D_{1L}$.
Without loss of generality, assume that $D_{2}\leq\ldots\leq D_{L}$
holds. Then the CMS scheme achieves the following complete rate-region:
\begin{eqnarray}
R_{i} & \geq & \frac{1}{2}\log\frac{1}{D_{i}},\,\,\forall i=\{1,2,\ldots,L\}\label{eq:sum_rate_L_thm1}\\
R_{1}+R_{i} & \geq & \frac{1}{2}\log\frac{1}{D_{1}D_{i}}+\delta(D_{1},D_{L},D_{1L})
\end{eqnarray}
if the distortions $D_{12},D_{13},\ldots,D_{1L}$ satisfy the following
conditions:
\begin{equation}
D_{1i}=\frac{\sigma_{1}^{2}\sigma_{i}^{2}(1-\rho_{1i}^{2})}{\sigma_{1}^{2}\sigma_{i}^{2}(1-\rho_{1i}^{2})+\sigma_{1}^{2}+\sigma_{i}^{2}-2\sigma_{1}\sigma_{i}\rho_{1i}}\label{eq:sum_rate_L_dist}
\end{equation}
where $\sigma_{j}^{2}=\frac{D_{j}}{1-D_{j}}$, $\forall j\in\{1,2,\ldots,L\}$
and $\rho_{12},\rho_{13},\ldots,\rho_{1L}$ are given by:
\begin{eqnarray}
\rho_{1i} & = & \rho_{1L}^{*}\frac{\sigma_{2}}{\sigma_{i}},\,\,\forall i\in\{1,2,\ldots,L\}\label{eq:rho_L_dist}
\end{eqnarray}
\end{thm}
\begin{rem}
It is possible to show that the CMS scheme achieves the minimum sum rate
in several distortion regimes similar to Theorem \ref{thm:Sum_Rate_Gauss}.
However, we omit explicitly stating the distortion regimes here as
they can be derived directly from the above theorem. \end{rem}
\begin{proof}
Following the arguments as in the proof of Theorem \ref{thm:Sum_Rate_Gauss},
the converse for the rate region in (\ref{eq:sum_rate_L_thm1}) follows
directly using Ozarow's converse results in \cite{Ozarow}. The encoding scheme which achieves
this sum rate resembles that in Theorem \ref{thm:Sum_Rate_Gauss}.
First, a quantized version of the source, optimized to achieve $D_{L}$,
is sent in description $L$. This information is also sent as part
of all descriptions $\{2\ldots,L-l\}$. Then, refinement information
is sent as part of description $L-1$ which helps in achieving a lower
distortion $D_{L-1}$. This first layer of refinement information
is also sent as part of all descriptions $\{2,\ldots,L-2\}$. Such
an encoding mechanism is repeated successively. The refinement information
generated at layer $l$ is sent in description $l$, as well as all
descriptions $\{2\ldots,L-l-1\}$. Note that successive refinability
of a Gaussian source under MSE \cite{Successive_Refinement} ensures
that there is no rate loss in encoding the source using such a mechanism.
Finally, the first description carries enough information to achieve
$D_{1}$ and $D_{1i}$, when the respective descriptions are received
at the decoder. 

Specifically, we choose the following joint density for the auxiliary
random variables:
\begin{eqnarray*}
U_{1} & = & X+W_{1}\\
U_{2} & = & X+W_{2}\\
V_{23} & = & X+W_{2}+W_{3}\\
V_{234} & = & X+W_{2}+W_{3}+W_{4}\\
 & \vdots\\
V_{23\ldots L} & = & X+W_{2}+W_{3}+\ldots+W_{L}
\end{eqnarray*}
\begin{equation}
U_{3}=V_{23},\, U_{4}=V_{234},\ldots,U_{L}=V_{23\ldots L}\label{eq:Joint_density_L}
\end{equation}
We set all the other auxiliary random variables to constants and set the functions
$\psi(\cdot)$ to be the respective optimal MSE estimators where $\{W_{1},W_{2},\ldots,W_{L}\}$
are zero mean jointly Gaussian random variables with the covariance
matrix given by: 
\begin{equation}
K_{W_{1}\ldots W_{L}}=\left[\begin{array}{ccccc}
\tilde{\sigma}_{1}^{2} & \rho_{12}\tilde{\sigma}_{1}\tilde{\sigma}_{2} & 0 & \cdots & 0\\
\rho_{12}\tilde{\sigma}_{1}\tilde{\sigma}_{2} & \tilde{\sigma}_{2}^{2} & 0 & \cdots & 0\\
0 & 0 & \ddots &  & 0\\
 & \vdots &  &  & \vdots\\
0 & 0 & 0 & \cdots & \tilde{\sigma}_{L}^{2}
\end{array}\right]
\end{equation}
where $\tilde{\sigma}_{1}^{2}=\sigma_{1}^{2}=\frac{D_{1}}{1-D_{1}}$,
$\tilde{\sigma}_{2}^{2}=\sigma_{2}^{2}=\frac{D_{2}}{1-D_{2}}$ and
$\tilde{\sigma}_{j}^{2}=\sigma_{j}^{2}-\sigma_{j-1}^{2}=\frac{D_{j}}{1-D_{j}}-\frac{D_{j-1}}{1-D_{j-1}}$, 
$\forall j\in\{3,\ldots,L\}$. $\rho_{12}$ is equal to $\rho_{12}^{*}$
($\rho_{12}^{*}$ is defined in eq. (\ref{eq:other_defn-1})). This
induces a correlation of $\rho_{1i}$ (defined in (\ref{eq:rho_L_dist}))
between $U_{1}$ and $U_{i}$. Observe that for any fixed $\rho_{12},\rho_{13},\ldots,\rho_{1L}$,
any distortion tuple satisfying the following conditions is achievable:
\begin{eqnarray}
D_{1i} & \geq & \mbox{Var}\left(X|U_{1},U_{i}\right)\,\, i\in\{2,3,\ldots,L-1\}\nonumber \\
 & = & \frac{\sigma_{1}^{2}\sigma_{i}^{2}(1-\rho_{1i}^{2})}{\sigma_{1}^{2}\sigma_{i}^{2}(1-\rho_{1i}^{2})+\sigma_{1}^{2}+\sigma_{i}^{2}-2\sigma_{1}\sigma_{i}\rho_{1i}}\label{eq:distortion_constraints}
\end{eqnarray}
Next, we have from Theorem \ref{thm:main} that the following rates
are achievable: $\forall l\in\{3,\ldots,L\}$:
\begin{eqnarray}
R_{1}^{'} & \geq & \frac{1}{2}\log\frac{1}{D_{1}}\nonumber \\
R_{2}^{'}+\sum_{i=3}^{L}R_{2\ldots l}^{''} & \geq & \frac{1}{2}\log\frac{1}{D_{2}}\nonumber \\
\sum_{i=l}^{L}R_{2\ldots i}^{''} & \geq & \frac{1}{2}\log\frac{1}{D_{l}}\nonumber \\
R_{1}^{'}+R_{2}^{'}+\sum_{i=l}^{L}R_{2\ldots i}^{''} & \geq & \frac{1}{2}\log\frac{1}{D_{1}D_{2}}+\delta(D_{1},D_{2},D_{12})\nonumber \\
R_{1}^{'}+\sum_{i=l}^{L}R_{2\ldots i}^{''} & \geq & \frac{1}{2}\log\frac{1}{D_{1}D_{l}}+\delta(D_{1},D_{L},D_{1L})\nonumber \\
R_{1} & = & R_{1}^{'}\nonumber \\
R_{2} & = & R_{2}^{'}+\sum_{i=l}^{L}R_{2\ldots i}^{''}\nonumber \\
R_{l} & = & \sum_{i=l}^{L}R_{2\ldots i}^{''}
\end{eqnarray}
Following similar arguments as in Theorem \ref{thm:Sum_Rate_Gauss},
it follows directly that the above rate-region is equivalent to (\ref{eq:sum_rate_L_thm1}),
completing the proof.
\end{proof}

\section{Conclusion}

A novel encoding scheme for the general $L-$channel multiple descriptions
problem was proposed which results in a new achievable rate-distortion
region that subsumes the achievable region due to Venkataramani, Kramer
and Goyal. The proposed encoding scheme adds controlled redundancy
by including a common codeword in every subset of the descriptions.
The common codewords assist in better coordination between the descriptions
leading to a strictly larger region for a fairly general class of
sources and distortion measures. In particular, we showed that the
proposed scheme achieves a strictly larger rate-distortion region
for a Gaussian source under MSE and for a binary symmetric source
under Hamming distortion for all $L\geq3$. We further showed that
CMS achieves the complete rate-distortion region for several asymmetric
cross-sections of the $L-$channel quadratic Gaussian MD problem.
The possible impact of the new ideas presented here, on practical
multiple descriptions encoder design and on certain special cases
of the general setting, such as the symmetric multiple descriptions
problem, will be studied as part of the future work. 


\begin{thebibliography}{10}

\bibitem{VKG}
R.~Venkataramani, G.~Kramer, and V.~Goyal.
\newblock Multiple descriptions coding with many channels.
\newblock {\em IEEE Trans. on Information Theory}, 49(9):2106--2114, Sep 2003.

\bibitem{EGC}
A.~Gamal and T.M. Cover.
\newblock Achievable rates for multiple descriptions.
\newblock {\em IEEE Trans. on Information Theory}, IT-28:851--857, Nov 1982.

\bibitem{ZB}
Z.~Zhang and T.~Berger.
\newblock New results in binary multiple descriptions.
\newblock {\em IEEE Trans. on Information Theory}, IT-33:502--521, Jul 1987.

\bibitem{Ahlswede}
R.~Ahlswede.
\newblock The rate-distortion region for multiple descriptions without excess
  rate.
\newblock {\em IEEE Trans. on Information Theory}, 31(6):721 -- 726, nov 1985.

\bibitem{Ozarow}
L.~Ozarow.
\newblock On a source-coding problem with two channels and three receivers.
\newblock {\em Bell Systems Tech. Journal}, 59(10):1909--1921, Dec 1980.

\bibitem{wang}
J.~Wang, J.~Chen, L.~Zhao, P.~Cuff, and H.~Permuter.
\newblock On the role of the refinement layer in multiple description coding
  and scalable coding.
\newblock {\em IEEE Trans. on Information Theory}, 57(3):1443--1456, Mar 2011.

\bibitem{Ramchandran-1}
S.S. Pradhan, R.~Puri, and K.~Ramchandran.
\newblock n-channel symmetric multiple descriptions - part {I}: (n, k)
  source-channel erasure codes.
\newblock {\em IEEE Trans. on Information Theory}, 50(1):47 -- 61, Jan. 2004.

\bibitem{Ramchandran}
R.~Puri, S.S. Pradhan, and K.~Ramchandran.
\newblock n-channel symmetric multiple descriptions - part {II}: {A}n achievable
  rate-distortion region.
\newblock {\em IEEE Trans. on Information Theory}, 51(4):1377 -- 1392, Apr
  2004.

\bibitem{wang_Vishwanath}
H.~Wang and P.~Viswanath.
\newblock Vector gaussian multiple descriptions with individual and central
  receivers.
\newblock {\em IEEE Trans. on Information Theory}, 53(6):2133 --2153, Jun 2007.

\bibitem{Jun_Chen_ind_central}
J.~Chen.
\newblock Rate region of gaussian multiple description coding with individual
  and central distortion constraints.
\newblock {\em IEEE Trans. on Information Theory}, 55(9):3991 --4005, Sep 2009.

\bibitem{Vaishampayan}
V.A. Vaishampayan.
\newblock Design of multiple description scalar quantizers.
\newblock {\em IEEE Trans. on Information Theory}, 39(3):821 --834, May 1993.

\bibitem{MC_DA}
P.~Koulgi, S.L. Regunathan, and K.~Rose.
\newblock Multiple description quantization by deterministic annealing.
\newblock {\em IEEE Trans. on Information Theory}, 49(8):2067 -- 2075, Aug
  2003.

\bibitem{Viswanatha_2_levels}
H.~Wang and P.~Viswanath.
\newblock Vector gaussian multiple description with two levels of receivers.
\newblock {\em IEEE Trans. on Information Theory}, 55(1):401 --410, Jan 2009.

\bibitem{Asymmetric_MD}
S.~Mohajer, C.~Tian, and S.N. Diggavi.
\newblock Asymmetric multilevel diversity coding and asymmetric gaussian
  multiple descriptions.
\newblock {\em IEEE Trans. on Information Theory}, 56(9):4367 --4387, Sep 2010.

\bibitem{SW}
D.~Slepian and J.~K. Wolf.
\newblock Noiseless coding of correlated information sources.
\newblock {\em IEEE Trans. on Information Theory}, 19:471--480, Jul 1973.

\bibitem{WZ76}
A.~D. Wyner and J.~Ziv.
\newblock The rate-distortion function for source coding with side information
  at the decoder.
\newblock {\em IEEE Trans. on Information Theory}, 22:1--10, Jan 1976.

\bibitem{Tian_Chen_symmetric_K_descriptions}
C.~Tian and J.~Chen.
\newblock New coding schemes for the symmetric k -description problem.
\newblock {\em IEEE Trans. on Information Theory}, 56(10):5344--5365, Oct 2010.

\bibitem{Approximating}
C.~Tian, S.~Mohajer, and S.N. Diggavi.
\newblock Approximating the gaussian multiple description rate region under
  symmetric distortion constraints.
\newblock {\em IEEE Trans. on Information Theory}, 55(8):3869 --3891, Aug 2009.

\bibitem{Song_Shuo_Chen}
J.~Chen L.~Song, S.~Shuo.
\newblock A lower bound on the sum rate of multiple description coding with
  symmetric distortion constraints.
\newblock {\em IEEE Trans. on Information Theory}, 60:7547--7567, Dec 2014.

\bibitem{DIR}
K.~Viswanatha, E.~Akyol, and K.~Rose.
\newblock On optimum communication cost for joint compression and dispersive
  information.
\newblock In {\em Proceedings of IEEE Information Theory Workshop (ITW)}, pages
  1--5, Sep 2010.

\bibitem{fusion}
J.~Nayak, S.~Ramaswamy, and K.~Rose.
\newblock Correlated source coding for fusion storage and selective retrieval.
\newblock In {\em Proceedings of IEEE International Symposium on Information
  Theory}, pages 92--96, Sep 2005.

\bibitem{our_ISIT}
K.~Viswanatha, E.~Akyol, and K.~Rose.
\newblock Combinatorial message sharing for a refined multiple descriptions
  achievable region.
\newblock In {\em Proceedings of IEEE International Symposium on Information
  Theory Proceedings (ISIT)}, pages 1312 --1316, Aug 2011.

\bibitem{our_ITW}
K.~Viswanatha, E.~Akyol, and K.~Rose.
\newblock A strictly improved achievable region for multiple descriptions using
  combinatorial message sharing.
\newblock In {\em Proceedings of IEEE Information Theory Workshop (ITW)}, pages
  100 --104, Oct 2011.

\bibitem{Binned_CMS}
E.~Akyol, K.~Viswanatha, and K.~Rose.
\newblock Combinatorial message sharing and random binning for multiple
  description coding.
\newblock In {\em Proceedings of IEEE International Symposium on Information
  Theory Proceedings (ISIT)}, pages 1371 --1375, Jul 2012.

\bibitem{Binned_CMS_ITW}
E.~Akyol, K.~Viswanatha, and K.~Rose.
\newblock On random binning versus conditional codebook methods in multiple
  descriptions coding.
\newblock In {\em Proc. IEEE Information Theory Workshop (ITW)}, Sep 2012.


\bibitem{our_isit_13}
E.~Akyol, K.~Viswanatha, and K.~Rose.
\newblock On the role of common codewords in quadratic Gaussian multiple descriptions coding
\newblock In {\em Proceedings of IEEE International Symposium on Information
  Theory Proceedings (ISIT)}, pages 1367 -- 1371, Jul 2013.


\bibitem{Shirani_Pradhan_14}
F.~Shirani and S.S. Pradhan.
\newblock An achievable rate-distortion region for the multiple descriptions
  problem.
\newblock In {\em IEEE International Symposium on Information Theory (ISIT)},
  pages 576--580, June 2014.

\bibitem{Cover-book}
T.M. Cover and J.A. Thomas.
\newblock {\em Elements of information theory}.
\newblock Wiley-{I}nterscience, 1991.

\bibitem{Gamal_notes}
A.~Gamal and Y~Kim.
\newblock {\em Network information theory}.
\newblock Cambridge University Press, 2011.

\bibitem{Wyner}
A.D. Wyner.
\newblock The rate-distortion function for source coding with side information
  at the decoder ii : {G}eneral sources.
\newblock {\em Information and Controls}, 38:60--80, Jul 1978.

\bibitem{MD_No_Excess_Marginal_Rate}
Z.~Zhang and T.~Berger.
\newblock Multiple description source coding with no excess marginal rate.
\newblock {\em IEEE Trans. on Information Theory}, 41(2):349--357, Mar 1995.

\bibitem{Zamir}
R.~Zamir.
\newblock Gaussian codes and shannon bounds for multiple descriptions.
\newblock {\em IEEE Trans. on Information Theory}, 45(7):2629 --2636, nov 1999.

\bibitem{Successive_Refinement}
W.H.R. Equitz and T.M. Cover.
\newblock Successive refinement of information.
\newblock {\em IEEE Trans. on Information Theory}, 37(2):269 --275, Mar 1991.

\bibitem{Han_Kobayashi}
T.S. Han and K.~Kobayashi.
\newblock A unified achievable rate region for a general class of multiterminal
  source coding systems.
\newblock {\em IEEE Trans. on Information Theory}, IT-26:277--288, May 1980.

\bibitem{Thomas}
J.~Thomas.
\newblock Feedback can at most double gaussian multiple access channel capacity
  ({C}orresp.).
\newblock {\em IEEE Trans. on Information Theory}, 33(5):711 -- 716, Sep 1987.

\end{thebibliography}
\bibliographystyle{IEEEtran}

\appendix

\section*{Appendix A: Error Bounds in Theorem 1\label{app:Error_Bounds_Thm1}}
\begin{proof}
We follow the notation and the notion of strong typicality defined
in \cite{Han_Kobayashi}. We refer to \cite{Han_Kobayashi} (section
3) for formal definitions and basic Lemmas associated with typicality.\textit{ }

Let $\mathcal{E}$ denote the event of an encoding error. We have:
\begin{eqnarray}
P(\mathcal{E}) & = & P(\mathcal{E}|x^{n}\in\mathcal{T}_{\epsilon}^{n})P(x^{n}\in\mathcal{T}_{\epsilon}^{n})\nonumber \\
 &  & +P(\mathcal{E}|x^{n}\notin\mathcal{T}_{\epsilon}^{n})P(x^{n}\notin\mathcal{T}_{\epsilon}^{n})
\end{eqnarray}
 From standard typicality arguments \cite{Cover-book}, we have $P(x^{n}\notin\mathcal{T}_{\epsilon}^{n})<\epsilon$
as $n\rightarrow\infty$. Hence, it is sufficient to find conditions
on the rates to bound $P(\mathcal{E}|x^{n}\in\mathcal{T}_{\epsilon}^{n})$. 

Towards finding conditions on the rate to bound $P(\mathcal{E}|x^{n}\in\mathcal{T}_{\epsilon}^{n})$,
we define event $\mathcal{A}(\{j\}_{\mathcal{Q}})$, for any index
tuple $\{j\}_{\mathcal{Q}}$, as, 
\begin{eqnarray}
\mathcal{A}(\{j\}_{\mathcal{Q}})=\Bigl\{\left(x^{n},v^{n}(\{j\}_{\mathcal{Q}}),u^{n}(\{j\}_{\mathcal{Q}})\right)\in\mathcal{T}_{\epsilon}^{n}\Bigr\}&\label{eq:event_F-1}
\end{eqnarray}
$\mathcal{Q}\subseteq\mathcal{J}(\mathcal{L})$, where for any set $\mathcal{Q}\in\mathcal{Q}^{*}$ (defined in (\ref{eq:cond_Q_main})),
$v^{n}(\{j\}_{\mathcal{Q}})$ denotes the codeword tuple $\{v_{\mathcal{K}}^{n}(\{j\}_{\mathcal{I}_{W+}(\mathcal{K})},j_{\mathcal{K}})\,\,\forall\mathcal{K}\in\mathcal{Q}\}$
and $u^{n}(\{j\}_{\mathcal{Q}})$ denotes the tuple $\{u_{\mathcal{K}}^{n}(\{j\}_{\mathcal{I}_{1+}(\mathcal{K})},\{j\}_{\mathcal{K}})\,\,\forall\mathcal{K}\subseteq[\mathcal{Q}]_{1},\mathcal{K}\neq\emptyset\}$.
Define random variables
\begin{equation}
\chi(\{j\}_{\mathcal{J}(\mathcal{L})})=\begin{cases}
1 & \mbox{if}\,\,\mathcal{A}(\{j\}_{\mathcal{J}(\mathcal{L})})\,\,\mbox{occurs}\\
0 & \mbox{else}
\end{cases}\label{eq:xi_defn-1}
\end{equation}
We have $P(\mathcal{E}|x^{n}\in\mathcal{T}_{\epsilon}^{n})=P(\Psi=0)$
where $\Psi=\sum_{\mathcal{J}(\mathcal{L})}\chi(\{j\}_{\mathcal{J}(\mathcal{L})})$.
From Chebyshev's inequality, it follows that:
\begin{equation}
P(\Psi=0)\leq P\left[|\Psi-E(\Psi)|\geq E(\Psi)/2\right]\leq\frac{4\mbox{Var}(\Psi)}{\left(E(\Psi)\right)^{2}}\label{eq:P_3-1}
\end{equation}

From standard typicality arguments, we can bound $E(\Psi)$ as follows:
\begin{equation}
E(\Psi)\geq2^{n\sum_{\mathcal{K}\in\mathcal{I}_{1+}}R_{\mathcal{K}}^{''}+n\sum_{l\in\mathcal{L}}R_{l}^{'}-n\left(\alpha(\mathcal{J}(\mathcal{L}))+\epsilon\right)}\label{eq:E_xi}
\end{equation}
where for any set $\mathcal{Q}$ satisfying (\ref{eq:cond_Q_main}),
we define: 
\begin{eqnarray}
\alpha(\mathcal{Q}) & = & \sum_{\mathcal{K}\in\mathcal{Q}-[\mathcal{Q}]_{1}}H\left(V_{\mathcal{K}}|\{V\}_{\mathcal{I}_{|\mathcal{K}|+}(\mathcal{K})}\right)\nonumber \\
 &  & +\sum_{\mathcal{K}\in2^{[\mathcal{Q}]_{1}}-\emptyset}H\left(U_{\mathcal{K}}|\{V\}_{\mathcal{I}_{1+}(\mathcal{K})},\{U\}_{2^{\mathcal{K}}-\emptyset-\mathcal{K}}\right)\nonumber \\
 &  & -H\left(\{V\}_{\mathcal{Q}-[\mathcal{Q}]_{1}},\{U\}_{2^{[\mathcal{Q}]_{1}}-\emptyset}|X\right)\label{eq:alpha_defn_1}
\end{eqnarray}
We follow the convention, $\alpha(\emptyset)=0$. Next consider $\mbox{Var}(\Psi)=E(\Psi^{2})-\left(E(\Psi)\right)^{2}$
where, 
\begin{eqnarray}
E(\Psi^{2})=\sum_{\{j\}_{\mathcal{J}(\mathcal{L})}}\sum_{\{k\}_{\mathcal{J}(\mathcal{L})}}E\left[\chi(\{j\}_{\mathcal{J}(\mathcal{L})})\chi(\{k\}_{\mathcal{J}(\mathcal{L})})\right]\nonumber \\
=\sum_{\{j\}_{\mathcal{J}(\mathcal{L})}}\sum_{\{k\}_{\mathcal{J}(\mathcal{L})}}P\left[\mathcal{A}(\{j\}_{\mathcal{J}(\mathcal{L})})\cap\mathcal{A}(\{k\}_{\mathcal{J}(\mathcal{L})})\right]\label{eq:E_xi_s}
\end{eqnarray}
The probability in (\ref{eq:E_xi_s}) depends on whether $\{j\}_{\mathcal{J}(\mathcal{L})}$
and $\{k\}_{\mathcal{J}(\mathcal{L})}$ are equal for a subset of
indices. Let $\mathcal{Q}\subseteq\mathcal{J}(\mathcal{L})$ such
that $\{j\}_{\mathcal{Q}}=\{k\}_{\mathcal{Q}}$. Observe that, due
to the hierarchical structure in the conditional codebook generation
mechanism, for $\left\{ v^{n}(\{j\}_{\mathcal{Q}}),u^{n}(\{j\}_{\mathcal{Q}})\right\} =\left\{ v^{n}(\{k\}_{\mathcal{Q}}),u^{n}(\{k\}_{\mathcal{Q}})\right\} $
to hold, $\mathcal{Q}$ must satisfy (\ref{eq:cond_Q_main}), i.e.,
it must be such that,
\begin{equation}
\mbox{if }\mathcal{K}\in\mathcal{Q}\,\,\Rightarrow\,\,\,\mathcal{I}_{|\mathcal{K}|+}(\mathcal{K})\subset\mathcal{Q}\label{eq:cond_Q}
\end{equation}
It follows from codebook generation that given the codeword tuple
$\{v^{n}(\{j\}_{\mathcal{Q}}),u^{n}(\{j\}_{\mathcal{Q}})\}$, tuples
$\{v^{n}(\{j\}_{\mathcal{J}(\mathcal{L})-\mathcal{Q}}),u^{n}(\{j\}_{\mathcal{J}(\mathcal{L})-\mathcal{Q}})\}$
and $\{v^{n}(\{k\}_{\mathcal{J}(\mathcal{L})-\mathcal{Q}}),u^{n}(\{k\}_{\mathcal{J}(\mathcal{L})-\mathcal{Q}})\}$
are independent and identically distributed. Hence we can rewrite
the probability in (\ref{eq:E_xi_s}) as:
\begin{eqnarray}
 &  & P\left[\mathcal{A}(\{j\}_{\mathcal{J}(\mathcal{L})})\cap\mathcal{A}(\{k\}_{\mathcal{J}(\mathcal{L})})\right]\nonumber \\
 &  & =\left(P\left[\mathcal{A}(\{j\}_{\mathcal{J}(\mathcal{L})-\mathcal{Q}})\biggl|\mathcal{A}(\{j\}_{\mathcal{Q}}\right]\right)^{2}\nonumber \\
 &  & \times P\left[\mathcal{A}(\{j\}_{\mathcal{Q}})\right]\label{eq:upper_bound_p_var}\\
 &  & =\left(P\left[\mathcal{A}(\{j\}_{\mathcal{J}(\mathcal{L})})\biggl|\mathcal{A}(\{j\}_{\mathcal{Q}}\right]\right)^{2}\nonumber \\
 &  & \times P\left[\mathcal{A}(\{j\}_{\mathcal{Q}})\right]\\
 &  & =\frac{\left(P\left[\mathcal{A}(\{j\}_{\mathcal{J}(\mathcal{L})})\right]\right)^{2}}{P\left[\mathcal{A}(\{j\}_{\mathcal{Q}})\right]}
\end{eqnarray}
Note that if $\mathcal{Q}=\emptyset$ (i.e., $\{j\}_{\mathcal{J}(\mathcal{L})}$
are $\{k\}_{\mathcal{J}(\mathcal{L})}$ are not equal for any subset
of the indices), then $P\left[\mathcal{A}(\{j\}_{\mathcal{J}(\mathcal{L})})\cap\mathcal{A}(\{k\}_{\mathcal{J}(\mathcal{L})})\right]=\left(P\left[\mathcal{A}(\{j\}_{\mathcal{J}(\mathcal{L})})\right]\right)^{2}$.
Hence, we bound $\mbox{Var}(\Psi)$ by (see \cite{VKG} for a similar
argument):
\begin{eqnarray}
\mbox{Var}(\Psi) & \leq & \sum\biggl\{ N(\mathcal{Q})P\left[\mathcal{A}(\{j\}_{\mathcal{Q}})\right]\nonumber \\
 &  & \times\left(P\left[\mathcal{A}(\{j\}_{\mathcal{J}(\mathcal{L})})\biggl|\mathcal{A}(\{j\}_{\mathcal{Q}})\right]\right)^{2}\biggl\}\nonumber \\
 & = & \sum\biggl\{ N(\mathcal{Q})\frac{\left(P\left[\mathcal{A}(\{j\}_{\mathcal{J}(\mathcal{L})})\right]\right)^{2}}{P\left[\mathcal{A}(\{j\}_{\mathcal{Q}})\right]}\biggl\}
\end{eqnarray}
where the summation is over all $\mathcal{Q}\subseteq2^{\mathcal{L}}-\emptyset$
such that (\ref{eq:cond_Q_main}) holds (i.e., over all $\mathcal{Q}\in\mathcal{Q}^{*}$)
and $\mathcal{A}(\{j\}_{\mathcal{Q}})$ denotes the event that $\Bigl\{\mathcal{A}(\{j\}_{\mathcal{Q}})=\left(x^{n},v^{n}(\{j\}_{\mathcal{Q}}),u^{n}(\{j\}_{\mathcal{Q}})\right)\in\mathcal{T}_{\epsilon}^{n}\Bigr\}$.
$N(\mathcal{Q})$ denotes the total number of ways of choosing $\{j\}_{\mathcal{J}(\mathcal{L})}$
and $\{k\}_{\mathcal{J}(\mathcal{L})}$ such that they overlap in
the subset $\mathcal{Q}$. An upper bound on $N(\mathcal{Q})$ is
given in eq. (\ref{eq:upper_bound_rateE_var}). Also observe that
the term corresponding to $\mathcal{Q}=\emptyset$ gets canceled with the
`$-\left(E(\Psi)\right)^{2}$' term in $\mbox{Var}(\Psi)$ leaving
a summation over all non-empty $\mathcal{Q}\in\mathcal{Q}^{*}$.

\begin{figure*}[!t]
\begin{eqnarray}
N(\mathcal{Q}) & = & 2^{n\left(\sum_{\mathcal{K}\in\mathcal{Q}-[\mathcal{Q}]_{1}}R_{\mathcal{K}}^{''}+\sum_{l\in[\mathcal{Q}]_{1}}R_{l}^{'}\right)}\prod_{\mathcal{K}\in\mathcal{J}(\mathcal{L})-\mathcal{Q}-[\mathcal{J}(\mathcal{L})-\mathcal{Q}]_{1}}2^{nR_{i,\mathcal{K}}^{'}}(2^{nR_{i,\mathcal{K}}^{'}}-1)\prod_{l\in[\mathcal{J}(\mathcal{L})-\mathcal{Q}]_{1}}2^{nR_{i,l}^{'}}(2^{nR_{i,l}^{'}}-1)\nonumber \\
 & \leq & 2^{n\{\sum_{\mathcal{K}\in\mathcal{Q}-[\mathcal{Q}]_{1}}R_{i,\mathcal{K}}^{''}+\sum_{l\in[\mathcal{Q}]_{1}}R_{i,l}^{'}+2\sum_{\mathcal{K}\in\mathcal{J}(\mathcal{L})-\mathcal{Q}-[\mathcal{J}(\mathcal{L})-\mathcal{Q}]_{1}}R_{i,\mathcal{K}}^{''}+2\sum_{l\in[\mathcal{J}(\mathcal{L})-\mathcal{Q}]_{1}}R_{i,l}^{'}\}}\nonumber \\
 & = & 2^{n\{2\sum_{\mathcal{K}\in I_{1+}}R_{i,\mathcal{K}}^{''}+2\sum_{l\in\mathcal{L}}R_{i,l}^{'}-\sum_{\mathcal{K}\in\mathcal{Q}-[\mathcal{Q}]_{1}}R_{i,\mathcal{K}}^{''}-\sum_{l\in[\mathcal{Q}]_{1}}R_{i,l}^{'}\}}\label{eq:upper_bound_rateE_var}
\end{eqnarray}
\end{figure*}

On substituting the upper bound for $N(\mathcal{Q})$, we have:

\begin{eqnarray}
\mbox{Var}(\Psi) & \leq & \sum\biggl\{2^{-2n\left(\alpha(\mathcal{J}(\mathcal{L}))-\sum_{\mathcal{K}\in\mathcal{I}_{1+}}R_{\mathcal{K}}^{''}-\sum_{l\in\mathcal{L}}R_{l}^{'}\right)}\nonumber \\
 &  & 2^{n\left(\alpha(\mathcal{Q})-\sum_{\mathcal{K}\in\mathcal{Q}-[\mathcal{Q}]_{1}}R_{\mathcal{K}}^{''}-\sum_{l\in[\mathcal{Q}]_{1}}R_{l}^{'}\right)+3n\epsilon}\biggl\}\label{eq:Var_xi}
\end{eqnarray}
where the summation is over all $\mathcal{Q}\in2^{\mathcal{L}}-\emptyset$
such that (\ref{eq:cond_Q_main}) holds, i.e., over all $\mathcal{Q}\in\mathcal{Q}^{*}$.

Inserting (\ref{eq:Var_xi}) and (\ref{eq:E_xi}) in (\ref{eq:P_3-1}),
we get,
\begin{equation}
P(\mathcal{E})\leq4\sum2^{n\left(\alpha(\mathcal{Q})-\sum_{\mathcal{K}\in\mathcal{Q}-[\mathcal{Q}]_{1}}R_{\mathcal{K}}^{''}-\sum_{l\in[\mathcal{Q}]_{1}}R_{l}^{'}\right)+5n\epsilon}\label{eq:P_4}
\end{equation}
where the summation is over all $\mathcal{Q}\in2^{\mathcal{L}}-\emptyset$
such that (\ref{eq:cond_Q_main}) holds, i.e., over all $\mathcal{Q}\in\mathcal{Q}^{*}$
. $P(\mathcal{E})$ can be made arbitrarily small if $\sum_{\mathcal{K}\in\mathcal{Q}-[\mathcal{Q}]_{1}}R_{\mathcal{K}}^{''}+\sum_{l\in[\mathcal{Q}]_{1}}R_{l}^{'}>\alpha(\mathcal{Q})\,\,\forall\mathcal{Q}\in\mathcal{Q}^{*}$. 
\end{proof}

\section*{Appendix B : Proof of Theorem \ref{thm:General_CMS-1}.iii\label{sec:Appendix-B}}

In this appendix, we will show that a Gaussian random variable, under
MSE distortion belongs to both $\mathcal{Z}_{EC}$ and $\mathcal{Z}_{ZB}$
and hence CMS achieves points strictly outside the VKG region. Throughout this section, we use the notation $\mathcal{C}_G$, to denote the set of all jointly Gaussian distributions.

First, let us recall the definitions of $\mathcal{Z}_{ZB}$ and $\mathcal{Z}_{EC}$.
$\mathcal{Z}_{ZB}$ is defined in Definition \ref{definition:Defn_ZZB}
as the set of all random variables $X$ for which there exists a strict
sub-optimality in the ZB region (with respect to the given distortion
measures) when the closure of the rates in (\ref{eq:ZB}) is defined
only over joint distributions that satisfy the following conditions:
\begin{eqnarray}
U_{1}\leftrightarrow & (X,V_{12}) & \leftrightarrow U_{2}\nonumber \\
E\left[d_{\mathcal{K}}(X,\psi_{\mathcal{K}}(U_{\mathcal{K}}))\right] & \leq & D_{\mathcal{K}},\,\,\mathcal{K}\in\{1,2,12\}\nonumber \\
U_{12} & = & f(U_{1},U_{2},V_{12})\label{eq:ZZB-1-1-1}
\end{eqnarray}
where $f$ is some deterministic function. Note that the statement
in Definition \ref{definition:Defn_ZZB} is more general than the
above description. However, for a Gaussian source, under MSE, it is
possible to show strict sub-optimality in ZB region even for $\epsilon=0$,
which leads to the above constraints for the joint distributions.
Similarly, $X$ is said to belong to $\mathcal{Z}_{EC}$ if there
exists a strict sub-optimality in the EC region when the closure of
the rates in (\ref{eq:EGC}) is defined only over joint distributions
that satisfy the following conditions:
\begin{eqnarray}
U_{1}\leftrightarrow & X & \leftrightarrow U_{2}\nonumber \\
E\left[d_{\mathcal{K}}(X,\psi_{\mathcal{K}}(U_{\mathcal{K}}))\right] & \leq & D_{\mathcal{K}},\,\,\mathcal{K}\in\{1,2,12\}\nonumber \\
U_{12} & = & f(U_{1},U_{2})\label{eq:ZZB-1-1-1-1}
\end{eqnarray}
where $f$ is some deterministic function. We next show that a Gaussian
source under MSE belongs to both $\mathcal{Z}_{EC}$ and $\mathcal{Z}_{ZB}$.
We first prove the result for $\mathcal{Z}_{EC}$ and then extend
similar arguments for $\mathcal{Z}_{ZB}$. 


1) \textbf{Proof for $\mathcal{Z}_{EC}$}: We first give an intuitive
argument to justify the claim and then follow it up with the formal
proof. It follows from Ozarow's results (see also \cite{EGC}) that there exists
a regime of distortions $(D_{1},D_{2},D_{12})$ for which the following
rate region is achievable (and complete):
\begin{eqnarray}
R_{1} & \geq & \frac{1}{2}\log\frac{1}{D_{1}}\nonumber \\
R_{2} & \geq & \frac{1}{2}\log\frac{1}{D_{2}}\nonumber \\
R_{1}+R_{2} & \geq & \frac{1}{2}\log\frac{1}{D_{12}}\label{eq:high_dist_region}
\end{eqnarray}
i.e., there is no excess rate incurred due to encoding the source
using two descriptions. Equivalently, there exists a regime of distortions
for which the excess sum rate term in the EC region (i.e., $I(U_{1};U_{2})$)
must be set to zero to achieve the complete rate-region. We will show
that, if we restrict the optimization to conditionally independent
joint densities, then it is impossible to simultaneously satisfy all
the distortions and achieve $I(U_{1};U_{2})=0$.

We next provide a formal proof. Our objective is to show that for the 2-descriptions quadratic Gaussian
MD problem, there is a strict sub-optimality if we perform the EC
encoding scheme and restrict the optimization to joint densities satisfying
the following conditions:

\begin{eqnarray}
P(U_{1},U_{2}|X) & = & P(U_{1}|X)\nonumber \\
 &  & \times P(U_{2}|X)\nonumber \\
E\left[(X-\psi_{i}(U_{i}))^{2}\right] & \leq & D_{i},\,\, i\in\{1,2\}\nonumber \\
E\left[(X-\psi_{12}(U_{12}))^{2}\right] & \leq & D_{12}\nonumber \\
U_{12} & = & f(U_{1},U_{2})\label{eq:EC_CMS}
\end{eqnarray}
for some functions $\psi_{1},\psi_{2}$,$\psi_{12}$ and $f$. We
denote by $\mathcal{RD}_{EC}^{IQ}$, the closure of all the achievable
RD tuples using the EC encoding scheme, over all joint densities satisfying
(\ref{eq:EC_CMS}). We need to show that $\mathcal{RD}_{EC}^{IQ}$
is strictly smaller than $\mathcal{RD}_{G}^{2}$. We consider one
point in $\mathcal{RD}_{G}^{2}$ and show that it is not contained
in $\mathcal{RD}_{EC}^{IQ}$. 

Let the distortions be such that $D_{12}\leq D_{1}+D_{2}-1$ holds.
This regime of distortions has been termed as the `high distortion
regime' in the literature and the complete rate region under these
constraints is given by (\ref{eq:high_dist_region}) (see for example \cite{Ozarow,Gamal_notes,Zamir}).
Let us consider the following rate point : $P_{0}\triangleq(R_{1},R_{2})=(\frac{1}{2}\log\frac{1}{D_{1}},\frac{1}{2}\log\frac{D_{1}}{D_{12}})$.
We will show that this point in not contained in $\mathcal{RD}_{EC}^{IQ}$
when $D_{12}\leq D_{1}+D_{2}-1$. We first rewrite the EC region for
any joint density over the random variables $(X,U_{1},U_{2},U_{12})$
and functions $\psi_{1},\psi_{2}$ and $\psi_{12}$:
\begin{eqnarray}
R_{1} & \geq & I(X;U_{1})\nonumber \\
R_{2} & \geq & I(X;U_{2})\nonumber \\
R_{1}+R_{2} & \geq & I(X;U_{1},U_{2},U_{12})+I(U_{1};U_{2})\\
D_{\mathcal{K}} & \geq & E\left[(X-\psi_{\mathcal{K}}(U_{\mathcal{K}}))^{2}\right],\mathcal{K}\in2^{\{1,2\}}-\emptyset\nonumber\label{eq:EGC-1}
\end{eqnarray}
Observe that, to achieve $P_{0}$, we have to impose the joint density
for $(X,U_{1})$ and the function $\psi_{1}$ to be RD optimal at
$D_{1}$. This requires $(X,U_{1})$ to be jointly Gaussian and $\psi_{1}$
to be the optimal estimator of $X$ given $U_{1}$. Note that,
if $P(U_{1}|X)$ is not Gaussian (with $\psi_{1}$ being the MMSE estimator), it is possible to construct a jointly Gaussian distribution
which achieves the same distortion with smaller rate. For example,
one could construct a distribution $(X,\tilde{U}_{1}) \in \mathcal{C}_G$,
with variance of $P(\tilde{U}_{1}|X)$ equal to the variance of $P(\psi_{1}(U_{1})|X)$.
As a Gaussian distribution achieves maximum entropy among all distributions 
with the same variance, $I(X;\tilde{U}_{1})<I(X;U_{1})$, leading
to a smaller rate for $R_{1}$. Therefore, we can assume that $(X,U_{1}) \in \mathcal{C}_G$. We denote this joint
density by $P_{G}(X,U_{1})$. It follows that the infimum $R_{2}$
achievable using an independent quantization scheme when $R_{1}=\frac{1}{2}\log\frac{1}{D_{1}}$
is given by:
\begin{eqnarray}
R_{EC}^{IQ} & = & \inf R_{2}:\Bigl\{ R_{1}=\frac{1}{2}\log\frac{1}{D_{1}},\nonumber \\
 &  & (R_{1},R_{2})\in\mathcal{RD}_{EC}^{IQ}\Bigr\}\nonumber \\
 & = & \inf \Bigl\{ I(X;U_{2},U_{12}|U_{1})+I(U_{1};U_{2}) \Bigl\} \label{eq:R_EC_IQ}
\end{eqnarray}
where the infimum is over all joint densities $P(X,U_{1},U_{2},U_{12})=P_{G}(X,U_{1})P(U_{2},U_{12}|X,U_{1})$
and functions $\psi_{1},\psi_{2},\psi_{12}$ and $f$ satisfying (\ref{eq:EC_CMS}),
where $(X,U_{1}) \in \mathcal{C}_G$ is RD optimal at $D_{1}$. We will next show that $R_{EC}^{IQ}>\frac{1}{2}\log\frac{D_{1}}{D_{12}}$.

In the following Lemma, we begin by proving that $R_{EC}^{IQ}$ is
greater than or equal $R_{EC}^{G}$, where $R_{EC}^{G}$ is defined
as:

\begin{eqnarray}
R_{EC}^{G} & = & \inf \Bigl\{ I(X;\tilde{U}_{2},\tilde{U}_{12}|U_{1})+I(U_{1};\tilde{U}_{2}) \Bigl\} \label{eq:R_EC_G}
\end{eqnarray}
where the infimum is over all jointly Gaussian densities $Q(X,U_{1},\tilde{U}_{2},\tilde{U}_{12})=P_{G}(X,U_{1})Q(\tilde{U}_{2},\tilde{U}_{12}|X,U_{1})$
satisfying the following conditions:
\begin{eqnarray}
Q(U_{1},\tilde{U}_{2}|X) & = & Q(U_{1}|X)\nonumber \\
 &  & \times Q(\tilde{U}_{2}|X)\nonumber \\
E\left[(X-\tilde{U}_{2})^{2}\right] & \leq & D_{2}\nonumber \\
E\left[(X-\tilde{U}_{12})^{2}\right] & \leq & D_{12}\label{eq:EC_CMS-1}
\end{eqnarray}

\begin{lem}
\label{lemma:before-last}Let $R_{EC}^{IQ}$ be defined as in (\ref{eq:R_EC_IQ})
and $R_{EC}^{G}$ be defined as in (\ref{eq:R_EC_G}). Then, we have:
\begin{equation}
R_{EC}^{IQ}\geq R_{EC}^{G}\label{eq:RECIQ_g_RECG}
\end{equation}
 \end{lem}
\begin{proof}
Consider any joint density $P_{G}(X,U_{1})P(U_{2},U_{12}|X,U_{1})$
and functions $\psi_{1},\psi_{2}$,$\psi_{12},$$f$, satisfying (\ref{eq:EC_CMS}),
where $(X,U_{1})$ are distributed according to a jointly Gaussian
density which is RD optimal at $D_{1}$. We now construct a joint 
density $P_{G}(X,U_{1})Q(\tilde{U}_{2},\tilde{U}_{12}|X,\tilde{U}_{1})$
that achieves a smaller value for the quantity in (\ref{eq:R_EC_IQ}) and is in $\mathcal{C}_G$.
Consider $P_{G}(X,U_{1})Q(\tilde{U}_{2},\tilde{U}_{12}|X,U_{1}) \in \mathcal{C}_G$
such that $K_{\tilde{U}_{2},\tilde{U}_{12}|X,U_{1}}=K_{\psi_{2}(U_{2}),\psi_{12}(U_{12})|X,U_{1}}$,
i.e., the random variables $(\tilde{U}_{2},\tilde{U}_{12})$ have
the same covariance matrix as $(\psi_{2}(U_{2}),\psi_{12}(U_{12}))$,
given $(X,U_{1})$. Clearly, all the conditions in (\ref{eq:EC_CMS-1})
are satisfied by the joint density $Q$. We need to show that the
quantity in (\ref{eq:R_EC_IQ}) is smaller for the joint density $Q$
compared to $P$. Recall that, for a fixed covariance matrix, a Gaussian
distribution over the relevant random variables maximizes the conditional
entropy \cite{Thomas}. Hence, we have:
\begin{eqnarray}
I(X;U_{2},U_{12}|U_{1}) & = & H(X|U_{1})-H(X|U_{1},U_{2},U_{12})\nonumber \\
+I(U_{1};U_{2}) &  & +H(U_{1})-H(U_{1}|U_{2})\nonumber \\
 & \geq & H(X|U_{1})-H(X|U_{1},\tilde{U}_{2},\tilde{U}_{12})\nonumber \\
 &  & +H(U_{1})-H(U_{1}|\tilde{U}_{2})\nonumber \\
 & = & I(X;\tilde{U}_{2},\tilde{U}_{12}|U_{1})+I(U_{1};\tilde{U}_{2})\label{eq:Gauss_ineq}
\end{eqnarray}
leading to $R_{EC}^{IQ}\geq R_{EC}^{G}$.
\end{proof}
Equipped with this result, we continue with our proof to show that $R_{EC}^{G}>\frac{1}{2}\log\frac{D_{1}}{D_{12}}$.
Consider the following series of inequalities:
\begin{eqnarray}
R_{EC}^{IQ}\geq R_{EC}^{G} & = & \inf \Bigl\{ I(X;\tilde{U}_{2},\tilde{U}_{12}|U_{1})+I(U_{1};\tilde{U}_{2}) \Bigl\} \nonumber \\
 & = & \inf \Bigl\{ I(X;\tilde{U}_{2},\tilde{U}_{12},U_{1})+I(U_{1};\tilde{U}_{2})\nonumber \\
 &  & -\frac{1}{2}\log\frac{1}{D_{1}} \Bigl\} \nonumber \\
 & \geq^{(a)} & \frac{1}{2}\log\frac{D_{1}}{D_{12}}+\inf I(U_{1};\tilde{U}_{2})\nonumber \\
 & >^{(b)} & \frac{1}{2}\log\frac{D_{1}}{D_{12}}
\end{eqnarray}
where the infimum is over all $Q(\tilde{U}_{2},\tilde{U}_{12}|X,\tilde{U}_{1}) \in \mathcal{C}_G$
satisfying (\ref{eq:EC_CMS-1}). Note that (a) follows from the fact
that $E\left[(X-\tilde{U}_{12})^{2}\right]\leq D_{12}$ (and therefore
$I(X;\tilde{U}_{2},\tilde{U}_{12},U_{1})\geq I(X;\tilde{U}_{12})\geq\frac{1}{2}\log\frac{1}{D_{12}}$).
(b) follows from the fact that $(X,U_{1},\tilde{U}_{2})$ are jointly
Gaussian satisfying the Markov condition $U_{1}\leftrightarrow X\leftrightarrow\tilde{U}_{2}$,
where $I(X;U_{1})=\frac{1}{2}\log\frac{1}{D_{1}}$ and $I(X;\tilde{U}_{2})\geq\frac{1}{2}\log\frac{1}{D_{2}}$.
Therefore, $\inf I(U_{1};\tilde{U}_{2})>0$, proving that there is
a strict sub-optimality if we restrict the EC coding scheme to an
independent quantization mechanism for a Gaussian random variable
under MSE. Hence a Gaussian random variable, under MSE, belongs to
$\mathcal{Z}_{EC}$. 

2) \textbf{Proof for $\mathcal{Z}_{ZB}$}: Our next objective is to
prove that a Gaussian random variable under MSE also belongs to $\mathcal{Z}_{ZB}$.
Equivalently, we need to show that there is a strict sub-optimality
in the ZB region when we restrict the optimization to joint densities
satisfying the following constraints:

\begin{eqnarray}
P(U_{1},U_{2}|X,V_{12}) & = & P(U_{1}|X,V_{12})\nonumber \\
 &  & \times P(U_{2}|X,V_{12})\nonumber \\
E\left[(X-\psi_{i}(U_{i}))^{2}\right] & \leq & D_{i},\,\, i\in\{1,2\}\nonumber \\
E\left[(X-\psi_{12}(U_{12}))^{2}\right] & \leq & D_{12}\nonumber \\
U_{12} & = & f(U_{1},U_{2},V_{12})\label{eq:ZB_CMS_1}
\end{eqnarray}
The proof follows in very similar lines to that for $\mathcal{Z}_{EC}$,
the main difference being that all the quantities are now defined
conditioned on every value of $V_{12}$. Recall that the ZB region
achievable using any joint density $P(X,V_{12},U_{1},U_{2},U_{12})$
is given by:
\begin{eqnarray}
R_{1} & \geq & I(X;V_{12},U_{1})\nonumber \\
R_{2} & \geq & I(X;V_{12},U_{2})\nonumber \\
R_{1}+R_{2} & \geq & I(X;V_{12})+I(U_{1};U_{2}|V_{12})\nonumber \\
 &  & +I(X;V_{12},U_{1},U_{2},U_{12}) \\
D_{\mathcal{K}} & \geq & E\left[(X-\psi_{\mathcal{K}}(U_{\mathcal{K}}))^{2}\right],\,\,\mathcal{K}\subseteq\{1,2\},\mathcal{K}\neq\emptyset\nonumber
\end{eqnarray}
Before we continue with the proof, we take a small digression and
prove certain properties of the ZB region. We summarize these properties
in the following lemma.
\begin{lem}
\label{lemma:End_Lemma}The closure of the RD-tuples in $\mathcal{RD}_{ZB}$
can be restricted to joint distributions that satisfy the following
Markov chain conditions, without any loss in the achievable region:
\begin{eqnarray}
X & \leftrightarrow & U_{1}\leftrightarrow V_{12}\nonumber \\
X & \leftrightarrow & U_{1}\leftrightarrow V_{12}\nonumber \\
X & \leftrightarrow & U_{12}\leftrightarrow {U_{1},U_{2},V_{12}}\label{eq:ZB_end_Markov}
\end{eqnarray}
\end{lem}
\begin{proof}
Let $P(X,V_{12},U_{1},U_{2},U_{12})$ be any joint density, which
may or may not satisfy the above properties, and let $\psi_{1},\psi_{2}$ and
$\psi_{12}$ be functions that satisfy the respective distortion constraints.
We will construct another joint density $Q(X,\tilde{V}_{12},\tilde{U}_{1},\tilde{U}_{2},\tilde{U}_{12})$
that satisfies (\ref{eq:ZB_end_Markov}) and achieves the same RD
region. Towards constructing such a joint density, we set:
\begin{eqnarray*}
\tilde{V}_{12} & = & V_{12}\\
\tilde{U}_{1} & = & (U_{1},V_{12})\\
\tilde{U}_{2} & = & (U_{2},V_{12})\\
\tilde{U}_{12} & = & (U_{12},U_{1},U_{2},V_{12})
\end{eqnarray*}
First, observe that the above choice of the joint density satisfies
(\ref{eq:ZB_end_Markov}). Also observe that, $I(X;V_{12},U_{1})=I(X;\tilde{V}_{12},\tilde{U}_{1})$,
$I(X;V_{12},U_{2})=I(X;\tilde{V}_{12},\tilde{U}_{2})$, $I(X;V_{12},U_{1},U_{2},U_{12})=I(X;\tilde{V}_{12},\tilde{U}_{1},\tilde{U}_{2},\tilde{U}_{12})$
and $I(U_{1};U_{2}|V_{12})=I(\tilde{U}_{1};\tilde{U}_{2}|\tilde{V}_{12})$.
Hence, the rate region achievable by $Q(\cdot)$ is the same as that
of $P(\cdot)$. Moreover, all the distortion constraints can be satisfied
using the same functions $\psi_{1},\psi_{2}$ and $\psi_{12}$. Hence
it follows that the entire ZB region can be achieved by considering
a closure only over random variable that satisfy (\ref{eq:ZB_end_Markov}). \end{proof}
\begin{cor}
\label{cor:last}The closure of the RD-tuples in $\mathcal{RD}_{ZB}^{IQ}$
can be restricted to joint distributions that satisfy (\textup{\ref{eq:ZB_end_Markov}}),
without any loss in $\mathcal{RD}_{ZB}^{IQ}$.\end{cor}
\begin{proof}
The proof is very similar to that of Lemma \ref{lemma:End_Lemma}.
We omit restating it here.
\end{proof}
Equipped with these results, we continue our proof.
Let us again consider the corner point $P_{0}\triangleq(R_{1},R_{2})=(\frac{1}{2}\log\frac{1}{D_{1}},\frac{1}{2}\log\frac{D_{1}}{D_{12}})$
for some $(D_{1},D_{2},D_{12})$ satisfying $D_{12}\leq D_{1}+D_{2}-1$, 
and show that it is not achievable by ZB when we restrict the joint
densities to satisfy (\ref{eq:ZB_CMS_1}). Note that, as a result
of Lemma (\ref{lemma:End_Lemma}) and Corollary (\ref{cor:last}),
it is sufficient to consider joint densities that also satisfy (\ref{eq:ZB_end_Markov}).
Observe that, as $I(X;V_{12},U_{1},U_{2},U_{12})\geq\frac{1}{2}\log\frac{1}{D_{12}}$,
$P_{0}$ can be achieved only by joint densities that satisfy $I(X;V_{12})=0$.
Hence, to prove that a Gaussian random variable under MSE belongs
to $\mathcal{Z}_{ZB}$, it is sufficient to show that $P_{0}$ is
not achievable when we restrict the optimization to joint densities
satisfying (\ref{eq:ZB_CMS_1}), (\ref{eq:ZB_end_Markov}) and $I(X;V_{12})=0$. 

Let $P(V_{12},U_{1},U_{2},U_{12},X)$ be any such joint density and
let $\mathcal{V}_{12}$ be the corresponding alphabet for $V_{12}$.
Then the associated achievable region can be rewritten as:

\begin{eqnarray}
R_{1} & \geq & I(X;U_{1}|V_{12})\nonumber \\
 & = & \sum_{v_{12}\in\mathcal{V}_{12}}P(v_{12})I(X;U_{1}|V_{12}=v_{12})\nonumber \\
R_{2} & \geq & I(X;U_{2}|V_{12})\nonumber \\
 & = & \sum_{v_{12}\in\mathcal{V}_{12}}P(v_{12})I(X;U_{2}|V_{12}=v_{12})\nonumber \\
R_{1}+R_{2} & \geq & I(U_{1};U_{2}|V_{12})\nonumber \\
 &  & +I(X;U_{1},U_{2},U_{12}|V_{12})\nonumber \\
 & = & \sum_{v_{12}\in\mathcal{V}_{12}}P(v_{12})\Bigl[I(U_{1};U_{2}|V_{12}=v_{12})\nonumber \\
 &  & +I(X;U_{1},U_{2},U_{12}|V_{12}=v_{12})\Bigr]\nonumber \\
D_{\mathcal{K}} & \geq & E\left[(X-\psi_{\mathcal{K}}(U_{\mathcal{K}}))^{2}\right],\,\,\mathcal{K}\subset\{1,2\},\mathcal{K}\neq\emptyset\nonumber \\
 & = & E\left[E\left[(X-\psi_{\mathcal{K}}(U_{\mathcal{K}}))^{2}\Bigl|V_{12}\right]\right]\label{eq:ZB_v12_split}
\end{eqnarray}


We will next show that the optimization can be further restricted
to joint densities $P(X,V_{12})Q(\tilde{U}_{1},\tilde{U}_{2},\tilde{U}_{12}|X,V_{12})$
such that $(X,\tilde{U}_{1},\tilde{U}_{2},\tilde{U}_{12}) \in \mathcal{C}_G$ given $V_{12}=v_{12}$, $\forall v_{12}\in\mathcal{V}_{12}$.
First, as $X$ is independent of $V_{12}$, $P(X|V_{12}=v_{12})$
is Gaussian $\forall v_{12}\in\mathcal{V}_{12}$. Recall that
$P_{0}$ is obtained by first minimizing $R_{1}$, followed by minimizing
$R_{2}$ given $R_{1}$ subject to all the distortion constraints.
From (\ref{eq:ZB_v12_split}), we have:
\begin{equation}
R_{1}=\inf\sum_{v_{12}\in\mathcal{V}_{12}}P(v_{12})I(X;U_{1}|V_{12}=v_{12})
\end{equation}
where the infimum is over all joint densities $P(X,V_{12},U_{1})$
satisfying the distortion constraint $D_{1}$. 

Towards proving that $(X,U_{1}) \in \mathcal{C}_G$ given $V_{12}=v_{12}$, 
$\forall v_{12}\in\mathcal{V}_{12}$, let $P(X,V_{12},U_{1})$ be any joint density over the given alphabets and let function $\psi_{1}(\cdot)$
achieve the distortion constraint at $D_{1}$. Consider the following joint
density $P(X,V_{12})Q(\tilde{U}_{1}|X,V_{12})$, such that $(X,\tilde{U}_{1})$
are jointly Gaussian given $V_{12}=v_{12}$ and $K_{X,\tilde{U}_{1}|V_{12}=v_{12}}=K_{X,\psi_{1}(U_{1})|V_{12}=v_{12}}$, 
$\forall v_{12}\in\mathcal{V}_{12}$. Observe that this joint density satisfies the distortion
constraint for $D_{1}$. Moreover, it achieves a smaller rate, as a Gaussian density over the relevant random variables maximizes
the conditional entropy for a fixed covariance matrix, i.e., $I(X;U_{1}|V_{12}=v_{12})\geq I(X;\tilde{U}_{1}|V_{12}=v_{12})$, 
$\forall v_{12}\in\mathcal{V}_{12}$. Therefore, to achieve minimum $R_1$, we can consider only  
densities wherein $(X,\tilde{U}_{1})$ are jointly Gaussian given
$V_{12}$. 


We next focus on minimizing $R_{2}$ given $R_{1}=\frac{1}{2}\log\frac{1}{D_{1}}$.
From (\ref{eq:ZB_v12_split}), we have:
\begin{eqnarray}
R_{2} & = & \inf \Bigl\{ \sum_{v_{12}\in\mathcal{V}_{12}}P(v_{12})\Bigl[I(\tilde{U}_{1};U_{2}|V_{12}=v_{12})\nonumber \\
 &  & +I(X;\tilde{U}_{1},U_{2},U_{12}|V_{12}=v_{12})\Bigr]-R_{1} \Bigl\} \label{eq:last_eq}
\end{eqnarray}
where the infimum is over all joint densities $P(X,V_{12},\tilde{U}_{1})P(U_{2},U_{12}|X,V_{12},\tilde{U}_{1})$
satisfying (\ref{eq:ZB_CMS_1}), (\ref{eq:ZB_end_Markov}) and $I(X;V_{12})=0$
and where $(X,\tilde{U}_{1})$ are jointly Gaussian given $V_{12}=v_{12}$, 
$\forall v_{12}\in\mathcal{V}_{12}$ (required to minimize $R_{1}$).
It is easy to show using similar arguments that, to achieve the infimum
in (\ref{eq:last_eq}), it is sufficient to consider joint densities
where $(X,\tilde{U}_{1},\tilde{U}_{2},\tilde{U}_{12}) \in \mathcal{C}_G$ given $V_{12}=v_{12}$ and $Q(\tilde{U}_{1},\tilde{U}_{2}|X,V_{12})=Q(\tilde{U}_{1}|X,V_{12})Q(\tilde{U}_{2}|X,V_{12})$,
$\forall v_{12}\in\mathcal{V}_{12}$. 

To summarize what we have so far, we have shown that to achieve minimum
$R_{1}$ followed by minimum $R_{2}$ in $\mathcal{RD}_{ZB}^{IQ}$,
it is sufficient to consider only joint densities $P(X,V_{12},U_{1},U_{2},U_{12})$
that satisfy the following properties:
\begin{enumerate}
\item $I(U_{1};U_{2}|X,V_{12})=0$, $U_{12}=f(U_{1},U_{2},V_{12})$ and
$P(X,V_{12},U_{1},U_{2},U_{12})$ satisfies all the distortion constraints
\item $I(X;V_{12})=0$
\item $(X,U_{1},U_{2},U_{12})$ are jointly Gaussian given $V_{12}=v_{12}$
$\forall v_{12}\in\mathcal{V}_{12}$
\item The Markov chain conditions in (\ref{eq:ZB_end_Markov}) are satisfied
\item $U_{1}$ achieves RD-optimality at $D_{1}$
\end{enumerate}
We will show that the point $P_{0}$ cannot be achieved by any
joint density that satisfies the above properties. Observe that, given $V_{12}=v_{12}$, $(X,U_{1},U_{2})$ are
jointly Gaussian random variables satisfying the Markov condition
$U_{1}\leftrightarrow X\leftrightarrow U_{2}$. This leads to two possibilities
for the joint distribution $P(X,V_{12},U_{1},U_{2},U_{12})$:

i) $\exists v_{12}\in\mathcal{V}_{12}$ such that $I(X;U_{1}|V_{12}=v_{12})>0$
and $I(X;U_{2}|V_{12}=v_{12})>0$: As $(X,U_{1},U_{2})$ are jointly
Gaussian, this leads to the conclusion that $I(U_{1};U_{2}|V_{12}=v_{12})>0$
and hence $I(U_{1};U_{2}|V_{12})>0$. This clearly implies that there
is excess rate on $R_{2}$ and hence point $P_{0}$ is not achievable. 

ii) $\forall v_{12}\in\mathcal{V}_{12}$, either $I(X;U_{1}|V_{12}=v_{12})=0$
or $I(X;U_{2}|V_{12}=v_{12})=0$. This implies that the alphabet space
of $V_{12}$ can be divided into two disjoint subsets such that, when
$V_{12}$ takes values in the first set, $X$ and $U_{2}$ are independent,
and when $V_{12}$ takes values in the second set $X$ and $U_{1}$
are independent. Let the two sets be denoted by $\mathcal{V}_{12}^{1}$
and $\mathcal{V}_{12}^{2}$, respectively. We have:
\begin{eqnarray}
I(X;U_{2}|V_{12}\in\mathcal{V}_{12}^{1}) & = & 0\nonumber \\
I(X;U_{1}|V_{12}\in\mathcal{V}_{12}^{2}) & = & 0\label{eq:end3}
\end{eqnarray}
However, recall that the joint distribution can be restricted to satisfy,
$I(X;V_{12})=0$, $I(X;V_{12}|U_{1})=0$ and $I(X;V_{12}|U_{2})=0$.
Hence, we have:
\begin{eqnarray}
H(X) & = & H(X|V_{12}\in\mathcal{V}_{12}^{1})\nonumber \\
 & =^{(a)} & H(X|U_{2},V_{12}\in\mathcal{V}_{12}^{1})\nonumber \\
 & = & H(X|U_{2})
\end{eqnarray}
\begin{eqnarray}
H(X) & = & H(X|V_{12}\in\mathcal{V}_{12}^{2})\nonumber \\
 & =^{(b)} & H(X|U_{1},V_{12}\in\mathcal{V}_{12}^{2})\nonumber \\
 & = & H(X|U_{1})
\end{eqnarray}
where (a) and (b) follow from (\ref{eq:end3}). This implies that
$(X,U_{1})$ and $(X,U_{2})$ must be pair-wise independent and hence
the distortion constraints on $D_{1}$ and $D_{2}$ cannot be satisfied. 

Hence, it follows that the point $P_{0}$ cannot be achieved by an
`independent quantization mechanism' using the ZB coding scheme for
a Gaussian source under MSE, proving the theorem.

\begin{IEEEbiographynophoto}{Kumar B. Viswanatha}(S'08) received his PhD degree in 2013 in the Electrical and Computer Engineering department from University of California at Santa Barbara (UCSB), USA. He is currently working for Qualcomm research center in San Diego. Prior to joining Qualcomm, he was an intern associate in the equity volatility desk at Goldman Sachs Co., New York, USA. His research interests include multi-user information theory, wireless communications, joint compression and routing for networks and distributed compression for large scale sensor networks.
\end{IEEEbiographynophoto}

\begin{IEEEbiographynophoto}{Emrah Akyol}(S'03, M'12) received the Ph.D. degree in 2011 in electrical and computer engineering from the University of California at Santa Barbara. From 2006 to 2007, he held positions at Hewlett-Packard Laboratories and NTT Docomo Laboratories, both in Palo Alto, where he worked on topics in video compression.

Currently, Dr. Akyol is a postdoctoral researcher in the Department of Electrical and Computer Engineering, University of Illinois Urbana-Champaign. Prior to that, he was a post doctoral researcher at Electrical Engineering Department at University of Southern California and at University of California - Santa Barbara. His current research is on the interplay of networked information theory, communications and control. 
\end{IEEEbiographynophoto}

\begin{IEEEbiographynophoto}{Kenneth Rose} (S'85-M'91-SM'01-F'03) received the Ph.D. degree in 1991 from the California Institute of Technology, Pasadena.

He then joined the Department of Electrical and Computer Engineering, University of California at Santa Barbara, where he is currently a Professor. His main research activities are in the areas of information theory and signal processing, and include rate-distortion theory, source and source-channel coding, audio and video coding and networking, pattern recognition, and non-convex optimization. He is interested in the relations between information theory, estimation theory, and statistical physics, and their potential impact on fundamental and practical problems in diverse disciplines.

Dr. Rose was co-recipient of the 1990 William R. Bennett Prize Paper Award of the IEEE Communications Society, as well as the 2004 and 2007 IEEE Signal Processing Society Best Paper Awards.

\end{IEEEbiographynophoto}

\end{document}